\documentclass[12pt,reqno,openany]{report} 
\usepackage[a4paper]{geometry}
\usepackage{amsfonts,amssymb,amsmath,amsthm}
\usepackage[nottoc]{tocbibind}

\def\C{\mathbb{C}}

\def\R{\mathbb{R}}
\def\Z{\mathbb{Z}}
\def\N{\mathbb{N}}
\let\ds\displaystyle

\newcommand{\be}{\begin{equation}}
\newcommand{\ee}{\end{equation}}

\newcommand{\pa}{\partial}

\newcommand{\IdentityMatrix}{{\bf 1}}

\theoremstyle{plain}
\newtheorem{theorem}{Theorem}[chapter]
\newtheorem{lemma}{Lemma}[chapter]
\newtheorem{corollary}{Corollary}[chapter]
\newtheorem{remark}{Remark}[chapter]
\newtheorem{proposition}{Proposition}[chapter]
\newtheorem{conjecture}{Conjecture}[chapter]

\newtheorem{op}{Open problem}[chapter]
\newtheorem{observation}{Observation}[chapter]
 
\def\ldb{\mathopen{\{\!\!\{}} \def\rdb{\mathclose{\}\!\!\}}}

\theoremstyle{remark}
\newtheorem{definition}{Definition}[chapter]
\newtheorem{example}{Example}[chapter]
\newtheorem{problem}{Problem}[chapter]
\newtheorem{exercise}{Exercise}[chapter]
\newtheorem{application}{Application}[chapter]

\DeclareMathOperator{\grad}{grad}

\def\p{\partial }

\begin{document}

\begin{center}

\vspace{100mm}

{\Huge{\sc algebraic structures related \\ to  integrable \\[5mm] differential equations}
}\\


\vspace{35pt}
{\bf \LARGE Vladimir Sokolov$^{a,b}$}\\[25pt]

{${}^a$\normalsize
 Landau Institute for Theoretical Physics, \\
142432 Chernogolovka (Moscow region), Russia\\}e-mail: {\normalsize
 \it
vsokolov@landau.ac.ru}\vspace{15pt}

{${}^b$\normalsize
 Sao Paulo University, \\
 Instituto de Matematica e Estatistica
 \\ 05508-090, Sao Paulo, Brasil\\} {\normalsize
 \it
sokolov@landau.ac.ru}  https://www.ime.usp.br \vspace{10pt}

\vspace{3mm}

\end{center}

\begin{abstract} The survey is devoted to algebraic structures related to integrable ODEs and evolution PDEs. A description of Lax representations is given in terms of vector space decomposition of loop algebras into direct sum of Taylor series and a complementary subalgebras. Examples of complementary subalgebras and corresponding integrable models are presented. In the framework of the bi-Hamiltonian approach compatible associative algebras related affine Dynkin diagrams are considered. A bi-Hamiltonian origin of the classical elliptic Calogero-Moser models is revealed.
\end{abstract}

\tableofcontents

\date{}

\chapter{Introduction}

\qquad A constructive approach to integrability is based upon the study of hidden and rich
algebraic or analytic structures associated with integrable equations.  In this survey algebraic structures associated with integrable ODEs and PDEs with two independent variables are considered.  Some of them are related to Lax representations for differential equations.  Furthermore, the bi-Hamiltonian formalism and the AKS factorization method  are considered.  Structures relevant to Yang-Baxter r-matrix are not discussed since many nice books have been written on the subject (see, for example \cite{FadTah, reysem}).
 
The statements are formulated  in the simplest form  but usually possible ways for generalization are pointed out.  In the proofs only essential points are mentioned while for technical details references are given.  
The text contains many carefully selected examples, which give a sense of the subject. A number of open problems are suggested.

The author is not a scrabble in original references. Instead, some references to reviews, where an information of pioneer works can be found, are given.  

The survey is addressed to both experts in algebra and in classical integrable systems. It is accessible to PhD students and can serve as an introduction to classical integrability for scientists with algebraic inclinations.

The exposition is based on a series of lectures  delivered by the author in USP (Sao Paulo, 2015).

The contribution of my collaborators I. Golubchik, V. Drinfeld, and  A. Odesskii  to results presented in  this survey is difficult to  overestimate. 

The author is thankful to the first readers of the survey A. Zobnin and S. Carpentier who made many suggestions and found a lot of misprints  contributing to the improvement of the text. 

The author is grateful to V. Kac, I. Shestakov and V. Futorny for their attention and to FAPESP for the financial support (grants 2014/00246-2 and 2016/07265-8) of my visits to Brazil, where the survay has been written.

\bigskip

\bigskip

\section{List of basic notation} 

\subsection{Constants, vectors and matrices}

Henceforth, the the field of constants is $\C$; $\bf u$ stands for $N$-dimensional vector, namely ${\bf u}=(u^1,\dots, u^N).$ Moreover, the  standard scalar product $\sum_{i=1}^{N} u^i\, v^i$ is denoted by $\langle {\bf u},\, {\bf v}\rangle$.

The associative algebra of order ``m'' square matrices is denoted by ${\rm Mat}_m$; the matrix  $\{u_{ij}\}\in {\rm Mat}_m$ is denoted by $\bf U.$  The unity matrix is denoted by ${\bf 1}$ or  ${\bf 1}_m$. The notation ${\bf U}^t$ stands for the matrix transpose of ${\bf U}$.

For the set of $n\times m$ matrices we use the notation ${\rm Mat}_{n,m}.$

\subsection{Derivations and differential operators} 
For ODEs the independent variable is denoted by $t$, whereas for PDEs we have two independent variables $t$ and $x$. Notation $u_t$ stands for the partial derivative of $u$ with respect to $t$. For the $x$-partial derivatives of $u$ the notation $u_x=u_1$, $u_{xx}=u_2,$ etc, is used.

The operator $\ds \frac{d}{d x}$ is often denoted  by $D$. For the differential operator $L=\sum_{i=0}^{k} a_i \,D^i$ we define the operator $L^{+}$ as
$$
L^+=\sum_{i=0}^{k} (-1)^i \,D^i\circ\, a_{i},
$$
where $\circ$ means that, in this formula, $a_i$ is the operator of multiplication by $a_i$. By $L_t$ we denote 
$$
L_t = \sum_{i=0}^{k} (a_i)_{t} \,D^i.
$$

\subsection{Differential algebra}
We denote by ${\cal F}$  a differential field. 
For our main considerations one can assume that elements of $ {\cal F}$ are rational functions of a finite number of independent variables $u_i.$ However, very ofter we find some functions solving overdetermind systems of PDEs. In such a case we have to extend the basic field $\cal F$. We will avoid any formal description of such extensions hoping that in any particular case it will be clear what we really need from $ {\cal F}$.
 
The principle derivation 
\be \label{DD}
D \stackrel{def}{=} \frac{\partial}{\partial x}+\sum_{i=0}^\infty u_{i+1} \frac{\partial}{\partial u_i},
\ee generates all independent variables $u_i$ starting from $u_0=u$. 

When we speak of solutions (or common solutions) of ODEs and PDEs, we mean {\it local}  solutions with generic initial date. 

\subsection{Algebra} 
We denote by $A(\circ)$ an $N$-dimensional algebra $A$ over $\C$ with an operation $\circ$.  A basis of $A$ is denoted by ${\bf e}_1, \dots , {\bf e}_N$, and corresponding structural constants  by $C^{i}_{jk}$:
$$
{\bf e}_j \circ {\bf e}_k =  C^i_{jk}\, {\bf e}_i.  
$$
In what follows we assume that the summation is carried out over repeated indices. 
We will use the following notation: 
\begin{equation}\label{as}
As(X, Y, Z) = (X \circ Y) \circ Z - X \circ (Y \circ Z),
\end{equation}
\begin{equation}\label{br}
[X, Y, Z] = As(X, Y, Z) - As(Y, X, Z).
\end{equation}

By ${\cal G}$ and ${\cal A}$ we usually denote a Lie and an associative algebra, respectively.

The algebra of Laurent series of the form
$$
S=\sum_{i=-n}^{\infty} c_{i} \lambda^{i}, \qquad c_{i}\in \C, \qquad n\in \Z
$$
is denoted by $\C ((\lambda)),$ for the subalgebra of Taylor series we  use $\C [[\lambda]]$ and $\C[\lambda]$ stands for polynomials in $\lambda.$  By $S_{+}$ and $S_{-}$ we denote  $$
S_{+}=\sum_{i=0}^{\infty} c_{i} \lambda^{i},\qquad {\rm and} \qquad 
S_{-}=\sum_{i=-n}^{-1} c_{i} \lambda^{i},
$$
respectively. We use a similar notation for the commutative and non-commutative Laurent series with coefficients from Lie and associative algebras.

\section{Lax pairs} 

The modern theory of integrable systems was inspired by
the discovery of the inverse transform method \cite{Z}, \cite[Chapter 1]{AblSeg81}. The main ingredient of this method is a Lax representation for 
a differential equation under investigation.

A Lax representation for a given differential equation is a relation of the form 
\begin{equation} \label{Lax}
L_t=[A,\,L],
\end{equation}
where $L$ and $A$ are some operators,
which is equivalent to the differential equation. To apply the technique of the inverse scattering method  the operators $L$ and $A$ should depend on an additional (complex) parameter $\lambda$.

\subsubsection{ODE case}

A Lax representation for a differential equation
\begin{equation} \label{eq}
{\bf u}_t=\bf F({\bf u}), \qquad {\bf u}=({\rm u^1}, \dots, {\rm u^N}),
\end{equation}
is a relation of the form~\eqref{Lax}, where $L=L({\bf u},\lambda), \,A=A({\bf u},\lambda)$ are some
matrices.

\begin{lemma}\label{tr} \
\begin{itemize}
\item[i)] If $L_1$ and $L_2$ satisfy~\eqref{Lax}, then $L=L_1 L_2$ satisfies~\eqref{Lax};
\item[ii)] $\bar L=L^n$ satisfies~\eqref{Lax} for any $n\in \N$;
\item[iii)] $\hbox{\rm tr}\, L^n$ is an integral of motion for~\eqref{eq};
\item[ i{\rm v})] the coefficients of the characteristic polynomial $\,\hbox{\rm Det}\,(L-\mu \IdentityMatrix)\,$ are integrals of motion.
\end{itemize}
\end{lemma}

\begin{proof} Item i). We have 
$$
L_t=(L_1)_t L_2+L_1 (L_2)_t = [A,\,L_1] L_2+L_1 [A,\, L_2]=A\, L - L\,A.
$$
Item ii) follows from Item i).  
Item iii): if we apply the ${\rm trace}$ functional  to both sides of the identity 
$(L^n)_t=[A,\,L^n]$, we get $({\rm tr}\, L^n)_t=0.$ Item {i\rm v}) follows from Item ii) and from the formula
$${\rm Det} (L-\mu \IdentityMatrix)=\exp{({\rm tr}\, ({\rm log}\,(L-\mu \IdentityMatrix)))}.$$
\end{proof}

\begin{example}
\label{Example1.1}\label{manak}\cite{man}
Let ${\bf U}(t)$ be an $m\times m$-matrix,
$$
L=a \lambda+{\bf U}, \qquad A=\frac{\bf U^2}{\lambda},
$$
where $a=\hbox{diag}(a_1,\dots,a_m)$. Then~\eqref{Lax} is
equivalent to the ODE
\begin{equation}\label{euler}
{\bf U}_t=[{\bf U}^2,\,a].
\end{equation}

If
$$
{\bf U}=\left(%
\begin{array}{ccc}
  0 & u_1 & u_2 \\
  -u_1 & 0 & u_3 \\
  -u_2 & -u_3 & 0
\end{array}%
\right), \qquad
a=\left(%
\begin{array}{ccc}
  a_3 & 0 & 0 \\
  0 & a_2 & 0 \\
  0 & 0 & a_1
\end{array}%
\right),
$$
where $a_i$ are arbitrary parameters, 
then~\eqref{euler} is equivalent to the Euler top
$$
(u_1)_t=(a_3-a_2) u_2 u_3,\qquad
(u_2)_t=(a_1-a_3) u_1 u_3,\qquad
(u_3)_t=(a_2-a_1) u_1 u_2.
$$

The characteristic polynomial $\,\hbox{Det}\,(L-\mu \IdentityMatrix)\,$ is given by
$$
(\mu-a_1 \lambda)(\mu-a_2 \lambda)(\mu-a_3 \lambda) + 
(u_1^2+u_2^2+u_3^2) \mu + (a_1 u_1^2+a_2 u_2^2+a_3 u_3^2) \lambda.
$$
The coefficients of the monomials in $\lambda$ and $\mu$ provide two non-trivial first integrals for the Euler top. The corresponding characteristic curve
$$
\,\hbox{Det}\,(L-\mu \IdentityMatrix)=0\,
$$
is elliptic.
The eigenfunction $\Psi(\lambda,\mu,t)$ satisfying 
\be \label{ee1}
L\Psi=\mu \Psi
\ee 
defines a vector bundle over this curve.
The dependence $\Psi$ on $t$ is described by the linear equation
\begin{equation} \label{ee2} \Psi_t=A\,\Psi.
\end{equation}
Using~\eqref{ee1} and~\eqref{ee2}, one can construct $\Psi(t)$ and after that find the corresponding solution ${\bf U}(t)$.
\end{example}

The assumption that $L$ and $A$ in \eqref{Lax} are functions of $t$ and $\lambda$ with values in a finite-dimensional Lie algebra $\cal G$ is a remarkable specification in the case of generic matrices $L$ and $A$, which reduces the number of unknown functions in the corresponding non-linear system of ODEs.  

\begin{remark}\label{rem11} We may assume also that in \eqref{Lax} the $A$-operator belongs to  $\cal G$ while $L$ belongs to a module over $\cal G$ {\rm (see, for example,~ \cite{sokolgok1}}\rm ).
\end{remark} 

\subsubsection{Lax pairs for   evolution PDEs}
\begin{example} \label{Example1.2}
The Lax pair~\eqref{Lax} for the KdV equation
\begin{equation} \label{kdv} 
u_t=u_{xxx}+ 6\,u\,u_x
\end{equation}
found by P.~Lax in~\cite{Lax} is given by  
$$
L=D^2+u+\lambda, \qquad A=4 D^3+6 u D+ 3 u_x, \qquad D=\frac{d}{d x}. 
$$
In contrast with Example~\ref{Example1.1}, here $L$ and $A$ are differential operators.
The relations \eqref{ee1}, \eqref{ee2} allow to construct $\Psi(x,t)$
by the inverse scattering method and, as a result, to find the corresponding solution $u(x,t)$ for the KdV equation.
\end{example}

\begin{example}\label{Example1.3}
The Lax representation for the nonlinear Schr\"odinger (NLS) equation  written as a system of two equations
\begin{equation}\label{NLS}
u_t=-u_{xx}+2 u^2 v, \qquad v_t=v_{xx}-2 v^2 u
\end{equation}
has been found by V.~Zakharov and A.~Shabat~\cite{ZS}.
The Lax $L$-operator is defined by 
$$
L=D 
+ \lambda \begin{pmatrix}
  1 & 0 \\
  0 & -1
\end{pmatrix} + \begin{pmatrix}
  0 & u \\
  v & 0
\end{pmatrix}.
$$
The operator $A$ is a polynomial in $\lambda$ with matrix coefficients
which depend on $u, v, u_x, v_x, u_{xx}, v_{xx}, \ldots$ (see Section~\ref{SectionMatrixLaxPairs}).
\end{example}

In this example the $L$-operator has the form $L=D-B$, where $B$ is a matrix depending on unknown functions and the spectral parameter $\lambda.$ For this special case equation  \eqref{Lax} can be written as
\begin{equation}\label{zercur}
A_x - B_t=[A,\, B].
\end{equation}
Relation  \eqref{zercur} is called a 
{\it zero-curvature representation}.

In contrast with the ODE case (see Remark \ref{rem11}) we may additionally assume only that $A$ and $B$ in  \eqref{zercur} are functions of $x, t, \lambda$ with values in a finite-dimensional Lie algebra $\cal G$. In the NLS case we have ${\cal G}=\mathfrak{sl}_2.$

\section{Hamiltonian structures}

 Let $y_1,\dots,y_m$ be coordinate functions. Any
Poisson bracket between functions $f(y_1,\dots,y_m)$ and
$g(y_1,\dots,y_m)$ is given by
\begin{equation}\label{brcom}
\{f,\,g\}=\sum_{i,j} P_{i,j}(y_1,\dots,y_m) \frac{\partial
f}{\partial y_i} \frac{\partial g}{\partial y_j},
\end{equation}
where $P_{i,j}=\{y_i,y_j\}$. The functions $P_{ij}$ are not
arbitrary since by definition
\be \label{Hcond1}
\{f,g\}=-\{g,f\}, 
\ee
\be \label{Hcond2}
\{\{f,g\},h\}+\{\{g,h\},f\}+\{\{h,f\},g\}=0.
\ee
Formula \eqref{brcom} can be rewritten as 
\begin{equation}\label{brop}
\{f,\,g\}= \langle \grad \,f, \, {\cal H} (\grad \,g)\rangle,
\end{equation}
where ${\cal H}=\{P_{i,j}\}$ and $\langle \cdot,\cdot \rangle$ is the standard scalar product. ${\cal H}$ is called a {\it Hamiltonian operator} or a {\it Poisson tensor}.
\begin{definition} The Poisson bracket is called {\it degenerate} if ${\rm Det} \,{\cal H}=0.$
\end{definition}

The Hamiltonian dynamics is defined by 
\be \label{Hdyn}
\frac{d y_i}{d t}=\{H,\,y_i\}, \qquad i=1,\dots,m,
\ee
where $H$ is a Hamiltonian function. If $\{K,\,H\}=0,$ then $K$ is
an integral of motion for the dynamical system. In this case the dynamical system 
$$
\frac{d y_i}{d \tau}=\{K,\,y_i\}
$$
is an infinitesimal symmetry for~\eqref{Hdyn} \cite{Olv93}. 

If $\, \{J,\,f\}=0\, $ for any $f$, then $J$ is called a {\it Casimir
function} of the Poisson bracket. The Casimir functions exist iff
the bracket is degenerate.

For the symplectic manifold the coordinates are denoted by $q_i$ and $p_i$,
$i=1,\dots N$. The standard constant Poisson bracket is given by
\begin{equation}\label{stand}\emph{\emph{\emph{}}}
\{p_i,p_j\}=\{q_i,q_j\}=0, \qquad \{p_i,q_j\}=\delta_{i,j},
\end{equation}
where $\delta$ is the Kronecker symbol. 
The corresponding dynamical system has the usual Hamiltonian form
$$\frac{d p_i}{d t}=- \frac{\partial H}{\partial q_i},\qquad
\frac{d q_i}{d t}=\frac{\partial H}{\partial p_i}.$$

For linear Poisson brackets we have
$$
P_{ij}=\sum_k b^k_{ij} x_k, \qquad i,j,k =1,\dots, N.
$$
It is well-known that this formula defines a Poisson bracket iff
$b^k_{ij}$ are {\it structure constants of a Lie algebra}.
Very often the title of this Lie algebra is also used for the corresponding linear Poisson bracket. 

For the spinning top-like systems~\cite{bormam, audin} the Hamiltonian structure is defined by 
linear Poisson brackets.  

\begin{example}
For the models of rigid body dynamics~\cite{bormam} the Poisson
bracket is given by
$$
\{M_{i},M_{j}\}=\varepsilon_{ijk}\,M_{k}, \qquad
\{M_{i},\gamma_{j}\}=\varepsilon_{ijk}\,\gamma_{k}, \qquad
\{\gamma_{i},\gamma_{j}\}=0 .
$$
Here $M_{i}$ and $\gamma_{i}$ are components of 3-dimensional
vectors $\bf{M}$ and $\bf{\Gamma}$, $\varepsilon_{ijk}$ is the
totally skew-symmetric tensor.
The corresponding Lie algebra $e(3)$ is the Lie algebra of the
group of motions in $\R^{3}.$
 This bracket  has two Casimir functions
$$
J_{1}=\langle {\bf M}, {\bf \Gamma}\rangle, \qquad  J_{2}= \vert {\bf
\Gamma}\vert ^{2}. 
$$
\end{example}

\bigskip

The class of quadratic Poisson brackets
\begin{equation}\label{qua}
\{x_i,\,x_j\}=\sum_{p,q}r_{i,j}^{p,q} \,x_p x_q, \qquad i,j=1,\dots,N,
\end{equation}
is of a great importance for the modern mathematical physics.

As for evolution PDEs of the form 
\begin{equation}\label{eveq}
u_t=F(u, u_x,  u_{xx}, \dots , u_n), \qquad u_i=\frac{\partial^i
u}{\partial x^i},
\end{equation}
the Poisson brackets are also defined by formula \eqref{brop}. However, we should take the variational derivative  instead of the gradient. Furthermore, the Hamiltonian operator ${\cal H}$ is not a matrix but a differential (or even pseudo-differential) operator.

\bigskip

\section{Bi-Hamiltonian formalism}\label{sec14}

\begin{definition}\cite{magri}
Two Poisson brackets $\{\cdot, \cdot \}_1$ and $\{\cdot, \cdot
\}_2$ are said to be compatible if
$$
\{\cdot,\,\cdot\}_{\lambda}=\{\cdot, \cdot \}_1 +\lambda \{\cdot,
\cdot \}_2
$$
is a Poisson bracket for any $\lambda$.
\end{definition}
General results on the structure of the Hamiltonian pencil $\{\cdot,\,\cdot\}_{\lambda}$ can be found in \cite{magri1, zakhar}. In particular,
if the bracket $\{\cdot,\,\cdot\}_{\lambda}$ is degenerate, then a set of commuting integrals can be constructed as follows. 
\begin{theorem}\label{biH}{\rm \cite{
magri}}
Let
$$
C(\lambda)=C_0+\lambda C_1+\lambda^2 C_2+\cdots
$$
be a Casimir function for the bracket $\{\cdot, \cdot\}_{\lambda}$.
Then the coefficients $C_i$ commute with each other
with respect to both brackets $\{\cdot, \cdot \}_1$ and $\{\cdot, \cdot\}_2$.
\end{theorem}

It follows from Theorem \ref{biH} that 
$$
\{ C_{k+1}, y \}_1=-\{C_{k}, y \}_2
$$
for any function $y$ and any $k$.
Let us take $C_{k+1}$ for a Hamiltonian $H$. Then dynamical system \eqref{Hdyn} can be written in two different ways:
$$
\frac{d y_i}{d t}=\{C_{k+1},y_i\}_1=-\{C_{k},y_i\}_2. 
$$
All functions $C_j$ are integrals of motion for this system.

We see that the same dynamical system can be represented in two different Hamiltonian forms, with different compatible Hamiltonian structures and different Hamiltonians $C_{k+1}$ and $- C_k$. In this case we say that this system possesses a {\it bi-Hamiltonian representation}~\cite{magri}.

The spinning top-like systems usually are bi-Hamiltonian with respect to two compatible linear Poisson
brackets.  The corresponding algebraic object is a pair of compatible Lie algebras (see Section 3.2). 

\subsection{Shift argument method} 
Here is a standard way of constructing compatible Poisson brackets. 

Let ${\bf a}=(a_1,\dots, a_N)$ be a constant
vector. Then any linear Poisson bracket produces a constant
bracket compatible with the initial linear one by the
transformation $x_i \rightarrow x_i+\lambda a_i$ (see \cite{man, fom}). 

Consider now quadratic Poisson brackets \eqref{qua}. 
The shift $x_i\mapsto x_i+\lambda a_i$ leads to a Poisson bracket of the form   
$\{\cdot,\cdot\}_{\lambda}=\{\cdot,\cdot\}+\lambda
\{\cdot,\cdot\}_1+\lambda^2 \{\cdot,\cdot\}_2.$ If the coefficient of  $\lambda^2$ equals zero, then this formula defines a linear Poisson bracket $\{\cdot,\cdot\}_1$ compatible with~\eqref{qua}.

Thus, in the case of quadratic brackets the shift vector
${\bf a}=(a_1,\dots, a_N)$ is not arbitrary. Its components have to satisfy the overdetermined system of algebraic equations 
\begin{equation}\label{admiss}
\sum_{p,q} r_{i,j}^{p,q} \,a_p a_q=0,  \qquad i,j=1,\dots,N.
\end{equation}
Such a vector $\bf a$ is called {\it admissible}. The admissible vectors are nothing but 0-dimensional 
symplectic leafs for the Poisson bracket~\eqref{qua}.

Any $p$-dimensional vector space of admissible vectors generates $p$ pairwise compatible linear Poisson brackets. Each of them 
is compatible with the initial quadratic bracket~\eqref{qua}.

Many interesting integrable models can be obtained~\cite{olsh, odsok11, sok1} by the shift argument method from the {\it elliptic quadratic Poisson brackets}~\cite{feyod}.

\begin{theorem} For the quadratic Poisson bracket $q_{m n^2,kmn-1}( \tau)$ {\rm (}for the definition of these brackets see {\rm \cite{feyod}}{\rm )}, the set of admissible vectors is
a union of $n^2$ components which are $m$-dimensional vector spaces. The space of generators of the
algebra is the direct sum of these spaces.
\end{theorem}
\begin{op}
Find integrable systems generated by the shift argument method from the elliptic Poisson brackets  $q_{m n^2,\,k m n-1}(\tau)$ . 
\end{op}

\subsection{Bi-Hamiltonian form for KdV equation}

Most of known integrable equations  \eqref{eveq} can be written in a
Hamiltonian form
$$
u_t={\cal H}\left(\frac{\delta \rho}{\delta u}  \right),
$$
where  ${\cal H}$ is a Hamiltonian
operator. The corresponding Poisson bracket is given by 
\begin{equation}\label{evpuas}
\{f, \, g\}=\frac{\delta f}{\delta u}\,{\cal H}\Big(\frac{\delta g}{\delta u} \Big),
\end{equation}
where 
\[
\frac{\delta}{\delta u}=\sum_k (-1)^k D^k \circ \frac{\partial}
{\partial u_k}
\]
is the Euler operator or the variational derivative.
\begin{definition}\label{def28} Two functions $\rho_1, \rho_2$ are called
equivalent $\rho_1\sim \rho_2$ if $\rho_1-\rho_2\in {\rm Im}\,D$.
\end{definition}
\begin{remark} For functions $u(x)$ which are rapidly decreasing at $x \to \pm \infty$,  two equivalent polynomial conserved densities $\rho_1$ and $\rho_2$  with zero constant terms define the same functional
$$
  \int_{-\infty}^{+\infty} \rho(u, u_x, \dots)\, dx.
$$
\end{remark}

\begin{proposition}\label{tgms} If  $a\in  {\rm Im}\, D$, then  
  \[\frac{\delta a}{\delta u}=0\, \]
and therefore the variational derivative is well defined on the equivalent classes. 
\end{proposition}

By definition the Poisson bracket \eqref{evpuas} is defined on the vector space of equivalence classes  and satisfies \eqref{Hcond1}, \eqref{Hcond2}. The finite-dimensional bracket \eqref{brcom} satisfies also the Leibniz rule
$$
\{f,\,g\, h\} = \{f,\,g\}\, h + g\, \{f,\,h\}.
$$
For brackets  \eqref{evpuas} the Leibniz rule has no sense since the product of equivalence classes is not defined.
 
We don't discuss here the bi-Hamiltonian formalism for evolution equations of the form \eqref{eveq} 
in general.  Notice only that KdV equation~\eqref{kdv} 
is a bi-Hamiltonian system~\cite{magri}. Two compatible Poisson brackets are given by the formula \eqref{evpuas}, where the Hamiltonian operators
${\cal H}_i$ are differential ones:
$$
{\cal H}_1 = D, \qquad {\cal H}_2 = D^3+4 u D+ 2 u_x. 
$$
Notice that ${\cal H}_1$ can be obtained from ${\cal H}_2$ by the argument shift $u\to u+\lambda.$

The KdV equation can be written in the bi-Hamiltonian form:
$$
u_t = {\cal H}_1 \frac{\delta \rho_1}{\delta u} = {\cal H}_2 \frac{\delta \rho_2}{\delta u},
$$
where
$$
\rho_1= -\frac{u_x^2}{2} + u^3, \qquad \rho_2=
\frac{u^2}{2}.
$$

\chapter{Factorization of Lie algebras and Lax pairs}

In this chapter we discuss different types of Lax representations for integrable PDEs and some constructions that allow one to find higher symmetries and conservation laws using Lax pairs. For Hamiltonian structures related to Lax operators see, for example, the books \cite{FadTah, reysem} and the original papers \cite{drsok, krich}.  
 
\section{Scalar Lax pairs for evolution equations} 

In this section the Lax $L$-operators are linear differential operators or ratios of linear differential operators. The corresponding $A$-operators are constructed by the use of formal non-commutative ``pseudo-differential'' series. 

For our purporses the language of differential algebra \cite{kapl} is the most adequate one.

Consider evolution equations of the form~\eqref{eveq}.  
Suppose that the right hand side of~\eqref{eveq} as well as other functions in $u,u_x,u_{xx}, \dots$ belong to a differential field 
$ {\cal F}$. For main considerations one can assume that elements of $ {\cal F}$ are rational functions of finite number of independent variables $$u_i=\frac{\partial^i u}{\partial x^i}.$$  In order to integrate a
function with
respect to one of its arguments or to take one of its roots, we sometimes have to
extend the
basic field $\cal F$.  
  
As usual in differential algebra, we have a principle derivation \eqref{DD},
 which generates all independent variables $u_i$ starting from $u_0=u$. This derivation is a formalization of the total $x$-derivative, 
which acts on functions of the form $\ds g\Big(x, u(x), \frac{\partial u}{\partial x}, \dots\Big)$. 

The variable $t$ in the local algebraic theory of evolution equations plays the role of a parameter. 

A higher (or generalized) infinitesimal symmetry (or a commuting flow) for \eqref{eveq} is an evolution equation
\begin{equation}\label{evsym}
u_{\tau}=G(u, u_x,  u_{xx}, \dots , u_m), \qquad m > 1
\end{equation}
which is compatible with~\eqref{eveq}. Compatibility means that $$
\frac{\partial}{\partial t}\frac{\partial u}{\partial \tau}=\frac{\partial}{\partial \tau}\frac{\partial u}{\partial t},
$$
where the partial derivatives are calculated in virtue of~\eqref{eveq} and~\eqref{evsym}. In other words, for any initial value $u_0(x)$ 
there exists a common solution $u(x,t,\tau)$ of equations~\eqref{eveq} and~\eqref{evsym} such that $u(x,0,0)=u_0(x)$.

\begin{example} \label{Example1.6}
The simplest higher symmetry for the Korteweg--de Vries equation \eqref{kdv}
has the following form
\begin{equation}\label{kdvsym}
u_{\tau}=u_{5}+10 u u_{3}+20 u_1 u_2+30 u^2 u_1.
\end{equation}
\end{example}

The infinite-dimensional vector field   
\begin{equation}
\label{Dt} D_F=F \frac{\partial}{\partial u_{0}}+
D(F)\frac{\partial} {\partial u_{1}}+ D^{2}(F)\frac{\partial}{\partial
u_{2}}+ \cdots\,  
\end{equation}
is associated with evolution equation  \eqref{eveq}. 
This vector field commutes with $D$. We shall call vector fields of the form  \eqref{Dt} {\it evolutionary}.
The function $F$ is called {\it generator} of that 
evolutionary vector field. Sometimes we will call  \eqref{Dt} {total} $t$-{derivative}
with respect to  \eqref{eveq} and denote it by $D_t$. The set of all evolutionary vector fields form a Lie algebra over $\C$: $[D_G,\, D_H]=D_K,$ where
\begin{equation}\label{comev}
K=D_G(H)-D_H(G)=H_*(G)-G_*(H).
\end{equation}
Here and in the sequel we use the following  notation:
\begin{definition}\label{dfrechet} For any element $a\in {\cal F}$ the Fr\'echet
derivative $a_{*}$ is a linear differential operator defined by
$$
a_* \stackrel{def}{=} \sum_k\frac{\partial a}{\partial u_k}D^k \, . 
$$
\end{definition}

We defined a generalized symmetry of equation \eqref{eveq} as an evolution equation~\eqref{evsym} that is compatible with~\eqref{eveq}. By definition, the compatibility means that $[D_F,\, D_G]=0.$ It can also be written in the form 
$$
G_*(F)-F_*(G)=D_t(G)-F_*(G)=0.
$$

Formula  \eqref{comev} defines a Lie bracket on our differential field ${\cal F}$. The integrable hierarchy is nothing but an infinite-dimensional commutative subalgebra of this Lie algebra.

A local conservation law for equation~\eqref{eveq} is a pair of functions
$\rho(u,u_x,...)$ and $\sigma(u,u_x,...)$ such that 
\be \label{conlaw}
D_t\Big(\rho(u, u_x,\dots, u_p) \Big)=D\Big(\sigma(u, u_x,\dots, u_q) \Big)
\ee
for any solution $u(x,t)$ of equation~\eqref{eveq}.
The functions $\rho$ and $\sigma$ are called {\it density} and {\it flux} of the conservation law~\eqref{conlaw}. It is easy to see that $q=p+n-1$, where $n$ is the order of equation \eqref{eveq}.

\begin{example} \label{Example1.7}
Functions
$$\rho_1=u,\qquad \rho_2=u^2,\qquad \rho_3=-u_x^2+2u^3 $$
are conserved densities of the Korteweg--de Vries equation~\eqref{kdv}.
Indeed, 
\begin{gather*}
D_t\Big(u\Big)=D\Big(u_2+3u^2\Big),\\[2mm]
D_t\Big(u^2\Big)=D\Big(2 u u_{xx}-u_x^2+4u^3\Big),\\[2mm]
D_t\Big(-u_x^2+2 u^3\Big)=D\Big(9u^4+6u^2 u_{xx}+u_{xx}^2-12u u_x^2-2u_x u_3\Big).
\end{gather*}
\end{example}

For solitonic type solutions $u(x,t)$ of \eqref{eveq}, which are decreasing at $x \to \pm \infty$, it follows from \eqref{conlaw} that  
$$
\frac{\partial}{\partial t} \int_{-\infty}^{+\infty} \rho\, dx = 0.
$$
This justifies the name {\it conserved density} for the function $\rho$.
Analogously, if $u$ is a function periodic in space  with period $L$, then the value of the functional 
$I(u)=\int_0^L \rho\, dx$ does not depend on time and therefore it is a constant
of motion.

Suppose that functions $\rho$ and $\sigma$ satisfy  \eqref{conlaw}. Then for any function $s(u,u_x,
\dots)$ the functions $\bar \rho=\rho+D(s)$ and 
 $\bar \sigma=\sigma+D_t(s)$ satisfy  \eqref{conlaw} as well. We call the conserved densities $\rho$ and $\bar \rho$ {\it equivalent}. It is clear that  equivalent densities define the same functional.

\subsection{Pseudo-differential series}

Consider a skew field of (non-commutative) formal series of the form
\begin{equation}\label{serA}
 S=s_{m}D^m+s_{m-1}D^{m-1}+\cdots + s_0+s_{-1}D^{-1}+
 s_{-2}D^{-2}+\cdots\,
\qquad s_k\in {\cal F}\, .
\end{equation}
The number $m\in \Z$ is called the {\it order} of $S$ and is denoted by ${\rm ord}\, S$.  If $s_i=0$ for $i<0$ that $S$ is called a {\it differential operator}.
\medskip

The product of two formal series is defined by the formula 
$$
 D^k\circ s D^m =s D^{m+k}+C_k^1 D(s)D^{k+m-1} + 
 C_k^2 D^2 (s)D^{k+m-2}+
\cdots \, ,
$$
where $k,m\in \mathbb Z$ and $C^j_n$ is the binomial coefficient
\[
C^j_n=\frac{n(n-1)(n-2)\cdots(n-j+1)}{j!},\qquad n\in \Z.
\]

\begin{remark} \label{remcom} For any series $S$ and $T$ we have ${\rm ord} (S\circ T -T\circ S)\le {\rm ord}\,S + {\rm ord}\,T -1.$
\end{remark}

The formally conjugated formal series $S^+$ is defined as
\[
 S^+=(-1)^m D^m\circ\, s_{m}+(-1)^{m-1}D^{m-1}\circ\, s_{m-1}+\cdots +
 s_0-D^{-1}\circ\,s_{-1}+D^{-2}\circ\, s_{-2}+\cdots\,.
\]

\begin{example}
Let
$$
R=u D^2+u_1 D,\qquad S=-u_1 D^3,\qquad T=u D^{-1};
$$
then
\begin{gather*}
R^+=D^2\circ u-D\circ u_1=R,\\[2mm]
S^+=D^3\circ u_1=u_1D^3+3u_2D^2+3u_3D+u_4,\\[2mm]
T^+=-D^{-1}u=-u D^{-1}+u_1D^{-2}-u_2D^{-3}+\cdots \,.
\end{gather*}
\end{example}

For any series \eqref{serA}
one can uniquely find the inverse series
$$
T=t_{-m}D^{-m}+t_{-m-1}D^{-m-1}+\cdots\, ,\qquad t_k\in  {\cal F}
$$
such that $S\circ T=T\circ S=1$.
Indeed, multiplying $S$ and $T$ and equating the result to 1, we find that $s_m t_{-m}=1$, i.~e., $t_{-m}=1/s_m$. Comparing the coefficients of $D^{-1},$ we get
$$
m s_m \, D(t_{-m})+s_m \,t_{-m-1}+s_{m-1}\,t_{-m}=0
$$
and therefore
$$
t_{-m-1}=-\frac{s_{m-1}}{s_m^2}-m D\Big(\frac{1}{s_m}\Big) \, ,\quad \mbox{etc.}
$$

Furthermore, we can find the $m$-th root of the series $S$, i.~e., a series
$$
R=r_1 D+r_0+r_{-1}D^{-1}+r_{-2}D^{-2}+\cdots
$$
such that $R^m=S$.
This root is unique up to any number factor $\varepsilon$ such that  $\varepsilon^m=1$.

\begin{example}\label{ex10}
Let $S=D^2+u$. Assuming
$$
R=r_1 D+r_0+r_{-1}D^{-1}+r_{-2}D^{-2}+\cdots ,
$$
we compute
$$
R^2=R\circ R= r_1^2 D^2+(r_1 D(r_1)+r_1 r_0 +r_0 r_1)D+r_1D(r_0)+r_0^2+r_1r_{-1}+r_{-1}r_1+\cdots\, ,
$$
and compare the result with $S$. From the coefficients of $D^2$ we find $r_1^2=1$ or
$r_1=\pm 1$. Let  $r_1=1$. Comparing coefficients of $D,$ we get $2r_0=0$, i.~e., $r_0=0$.
From $D^{0}$ we obtain $2r_{-1}=u$, terms of $D^{-1}$
$r_{-2}=-u_1/4$, etc., i.~e.
$$
R=S^{1/2}=D+\frac{u}{2}D^{-1}-\frac{u_1}{4}D^{-2}+\cdots\, .
$$
\end{example}

\begin{definition}
The {\it residue} of a formal series  \eqref{serA} by definition is
the coefficient of $D^{-1}$:
$$
 \hbox{res}\, (S) \stackrel{def}{=} s_{-1}\, .
$$
The {\it logarithmic residue} of $S$ is defined as
$$
 \hbox{res} \log S \stackrel{def}{=} \frac{s_{m-1}}{s_m}\, .
$$
\end{definition}

We will use the following important  
\begin{theorem} \label{adler} {\rm \cite{adler}}
For any two formal series $S$, $T$ the
residue of the commutator belongs to ${\rm Im}\,D$:
$$
{\rm res} [S,\,T]=D(\sigma (S,\,T)),
$$
where
$$
\sigma (S,\,T)=\sum_{i\le {\rm ord}(T),\ j\le
{\rm ord}(S)}^{i+j+1>0}C^{i+j+1}_{j}\, \times  
\sum_{k=0}^{i+j}(-1)^k D^k(s_j)D^{i+j-k}(t_j)\, .
$$
\end{theorem}

\medskip 

\subsection{Korteweg--de Vries hierarchy} 

For the KdV equation~\eqref{kdv} the Lax pair is defined by
\begin{equation}\label{LAkdv}
L=D^2+u, \qquad A=4\,\Big(D^3+\frac{3}{2} u D+ \frac{3}{4} u_x\Big).
\end{equation}
Using these $L$ and $A$ operators, we are going to demonstrate how a scalar differential Lax pair generates 
higher symmetries, conservation laws and explicit solutions of the solitonic type.

One can easily verify that the commutator $[A, L]$ is equal to the right hand side of~\eqref{kdv}. Since $L_t=u_t$ the relation  \eqref{Lax} is equivalent to~\eqref{kdv}.
In particular, the commutator $[A, L]$ does not contain any powers of $D$, i.~e., it is a differential operator of zero order. 

\begin{problem}\label{problem1}
How to describe all differential operators
$$P_n=D^n+\sum_{i=0}^{n-1} p_i D^i$$
such that $[P_n,\,L]$ is a differential operator of zero order?
\end{problem}
It is clear that for such an operator $P_n$ the relation $L_t=[P_n, L]$ is equivalent to an evolution equation of the form  \eqref{eveq}.
  
\medskip

\begin{definition} For any series
$$
P=\sum_{i=-\infty}^{k} p_i D^i
$$
we denote
$$
P_{+}=\sum_{i=0}^{k} p_i D^i, \qquad  P_{-}=\sum_{i=-\infty}^{-1}
p_i D^i.
$$
\end{definition}
\begin{remark} We consider a vector space decomposition of the associative algebra of all
pseudo-differential series into a direct sum of the subalgebra of
differential operators and the subalgebra of series of negative orders.
  The subscripts $+$ and $-$ symbolize the projections onto these subalgebras.  
\end{remark}

\begin{lemma}\label{lem1} Let $P$ be a formal series such that $[L,P]=0;$
  then
\begin{equation}\label{AA}
[P_{+},\,L]=f, \qquad  f\in  {\cal F}.
\end{equation}
\end{lemma}
\begin{proof} Since $[L,\,P_{+}+P_{-}]=0,$ we have
$$
[P_{+},\,L]=-[P_{-},\,L].
$$
The left hand side of this identity is a differential operator while according to Remark \ref{remcom} the order of the right hand side is not positive. \end{proof}

\begin{lemma}\label{lem2} The following relation holds:
$$
[L, \, L^{\frac{1}{2}}]=0. 
$$
\end{lemma}
\begin{proof}  
Let $$[L, \, L^{\frac{1}{2}}]=\sigma D^p+\cdots$$
Then (see Example \ref{ex10}) we have
$$
0=[L, \, L]=[L, \, L^{\frac{1}{2}}] L^{\frac{1}{2}}+L^{\frac{1}{2}}[L, \, L^{\frac{1}{2}}]=2 \sigma D^{p+1}\cdots ,
$$
and therefore $s=0.$
\end{proof} 

\begin{corollary} It follows from Lemmas \ref{lem1} and \ref{lem2} that for any $n\in \Z_{+}$ the differential operator $\ds P=L^{\frac{n}{2}}_{+}$
satisfies the relation 
\begin{equation}\label{AAA}
[P,\,L]=f_P, \qquad  f_P\in  {\cal F}
\end{equation}
for some $f_P$. If $n$ is even, then $f_P=0.$
\end{corollary}

It is clear that the set of all differential operators that satisfy  \eqref{AAA} is a vector space over $\C.$

\begin{lemma}\label{lem3} The differential operators
$\ds L^{\frac{n}{2}}_{+}, \,\,\, n\in \Z_{+},$ form a basis of this vector space.
 \end{lemma}
 \begin{proof} Suppose that $P=\sigma D^p+\cdots$ satisfies  \eqref{AAA}. Equating the coefficients of $D^{p+1},$ we get 
 $D(\sigma)=0$ and therefore $\sigma={\rm const}.$ Since the operator $\ds \sigma L^{\frac{p}{2}}_{+}$ has the same leading coefficient as $P$, the operator $\ds P-\sigma L^{\frac{p}{2}}_{+}$ has strictly lower order than $P$. The induction over $p$ completes the proof. 
 \end{proof} 
 
 Let
 $$
 [L^{\frac{n}{2}}_{+},\, L]=f_n.
 $$
 For even $n$ we have $f_n=0$ and the   evolution equation $u_t=f_n$ that is equivalent to 
\begin{equation}\label{hier}
 L_t= [L^{\frac{n}{2}}_{+},\, L]
 \end{equation}
 is trivial. Denote
 $$
 A_n=L^{\frac{2 n+1}{2}}_{+}.
 $$
It can be easily verified that $A=4 A_1$, where $A$ is defined by  \eqref{LAkdv}. The evolution equation corresponding to $n=0$ is just
$u_t=u_x.$

\begin{theorem} \label{Hier} For any $n,m\in \N$ the evolution equations $u_{\tau}=f_{2 m+1}$ and 
$u_t=f_{2 n+1}$ are compatible\footnote{In other words, the first equation is a higher symmetry for the second one and vice versa}.
\end{theorem}
\begin{proof}  Let us rewrite these equations in the Lax form:
$$
L_t=[A_n,\, L], \qquad L_{\tau}=[A_m,\, L].
$$
We have
$$
(L_{t})_{\tau}-(L_{\tau})_{t}=[(A_n)_{\tau},\, L]-[(A_m)_{t},\,
L] +
[A_n,\, [A_m,\, L]]-[A_m,\, [A_n,\, L]].
$$
Due to the Jacobi identity it suffices to prove that
$$
(A_n)_{\tau}-(A_m)_{t}+[A_n,\, A_m]=0.
$$
Since $(L^p)_t=[A_n, L^p]$ and  $(L^p)_{\tau}=[A_m, L^p]$ for any
$p$, we get
$$
(A_n)_{\tau}=\left([A_m, L^{\frac{2n+1}{2}}]\right)_{+}, \qquad 
(A_m)_{t}=\left([A_n, L^{\frac{2m+1}{2}}]\right)_{+}.
$$
Therefore, we need to verify that
$$
\left([A_m, L^{\frac{2n+1}{2}}]-[A_n, L^{\frac{2m+1}{2}}]+[A_n,\,
A_m]\right)_{+}=0.
$$
Substituting $$A_n=L^{\frac{2n+1}{2}}-\left(
L^{\frac{2n+1}{2}}\right)_{-}$$ and
$$A_m=L^{\frac{2m+1}{2}}-\left(
L^{\frac{2m+1}{2}}\right)_{-}$$ to the latter identity, we obtain
$$
\left[\left(L^{\frac{2n+1}{2}}\right)_{-},\,
\left(L^{\frac{2m+1}{2}}\right)_{-} \right]_{+}=0,
$$
which is obviously true.  
\end{proof}
\begin{corollary}
Any two evolution equations defined by  \eqref{Lax} with different $A$-operators of the form 
$$
A=\sum_{i\ge 0}  c_i L^{\frac{2 i+1}{2}}_{+}, \qquad c_i\in \C,
$$
are compatible. 
\end{corollary}

This infinite-dimensional vector space of compatible  evolution equations is called the {\it KdV hierarchy} \cite{geldik1}. Any two equations of the hierarchy  are higher symmetries for each other.

Thus, the symmetries of the KdV equation are generated by the same $L$-operator but by different $A$-operators.

\subsubsection{Recursion operator for KdV equation}

 Now we are going to find \cite[Section 2A]{GKS} a recursion relation between $f_{2 n+1}$ and
$f_{2 n+3}.$  A similar method was applied for the first time in \cite{sokkn} to find a recursion operator for the Krichever-Novikov equation.

Since $L^{2 n+3 \over 2}=L\, L^{2 n+1 \over 2},$ we have
$$
A_{n+1}=(L\, L^{2 n+1 \over 2})_{+}=L\,( L^{2 n+1 \over 2})_{+}+
(L \, (L^{2 n+1 \over 2})_{-})_{+},
$$
or
$$
A_{n+1}=L A_n+R_n,
$$
where $R_{n}=a_{n}D+b_{n}$ is a differential operator of first order. Hence
$$
 f_{2n+3}=[A_{n+1}, \ L]= L\circ f_{2n+1}+[R_n,\, L].
$$
Now if we equate to zero coefficients of $D^2$, $D$
and $D^0$ in the above equation, we obtain
\[a_{n}={1 \over 2}D^{-1}(f_{2n+1}), \qquad  b_{n}={3 \over 4} f_{2n+1}\]
and \[f_{2n+3}=\Big({1 \over 4}D^{2} +u+{1 \over 2}u_{x}D^{-1}\Big)\, f_{2n+1},\]\\
which gives the standard recursion operator
\begin{equation}\label{recop}
{\cal{R}} = {1 \over 4}D^{2} +u+{1 \over 2}u_{x}D^{-1}
\end{equation}
 for the KdV equation
\begin{equation}\label{kdv1}
 u_t=\frac{1}{4}\Big(u_{xxx}+6 u u_x \Big).
\end{equation}
The factor $\ds \frac{1}{4}$ appears due to the fact that we take for $A$-operator $L^{\frac{3}{2}}_{+}$ instead of $4 L^{\frac{3}{2}}_{+}.$ Of course, this coefficient can be removed by the scaling $t\to 4\,t.$

As it was shown in \cite{GKS} this method for finding a recursion operator can be generalized to  Lax pairs of different type.

 \begin{exercise} Check that the recursion operator \eqref{recop} satisfies the operator identity
$$
{\cal{R}}_t=[F_{*},\, {\cal{R}}],
$$
where
$$
F_{*}=\frac{1}{4}\left(D^3+6 u D+6 u_x\right)
$$
is the Frech{\' e}t derivative of the right hand side of the KdV equation \eqref{kdv1}.
\end{exercise}

 \begin{exercise} (see \cite[Appendix A]{GKS})
Find the recursion operator for the Boussinesq system
$$
u_{t}  = v_{x} ,\qquad 
v_{t}  = - {1 \over 3}(u_{xxx}+8uu_{x}). $$
A Lax pair for this system is given by
$$
L=D^{3}+2uD+u_{x}+v, \qquad A=\left(L^{\frac{2}{3}}\right)_{+}.
$$
\end{exercise}

\subsubsection{Conservation laws}

\begin{proposition} For any $n\in \N$ the function 
\begin{equation}\label{conkdv}
\rho_n={\rm res}\,(L^{\frac{2 n-1}{2}}),
\end{equation}
where $L$ is defined by \eqref{LAkdv}, is a conserved density for the KdV equation.
\end{proposition}
\begin{proof} It is easy to prove (cf. Lemma \ref{tr}) that 
$$
(L^{\frac{2 n-1}{2}})_t=[A,\,(L^{\frac{2 n-1}{2}})].
$$
Finding residue of both sides of this identity and taking into account Theorem \ref{adler}, we arrive at the statement of the proposition.
\end{proof}
It can be verified that formula  \eqref{conkdv} with $n=1,2,3$ produces conserved densities equivalent to the ones presented in Example \ref{Example1.7}.

\subsubsection{Darboux transformation}

The Darboux transformation for the KdV equation is defined by the following relation
\begin{equation}\label{tiL}
\tilde L=T L T^{-1},
\end{equation}
where
$$
T=D^n+a_{n-1} D^{n-1}+\cdots+a_0
$$
is a differential operator. 
In the generic case $\tilde L$ is a  pseudo-differential series, but for special $T$  this series could be a differential operator of the form $\tilde L=D^2+\tilde u.$ 
In this case we have
\begin{equation}\label{Tiden}
(D^n+a_{n-1} D^{n-1}+\cdots+a_0)(D^2+u)\,=\,(D^2+\tilde
u)(D^n+a_{n-1} D^{n-1}+\cdots+a_0).
\end{equation}
Comparing the coefficients of $D^{n},$ we get
$$
\tilde u=u-2 \frac{\partial}{\partial x}  a_{n-1}.
$$
This formula allows us to construct a new solution $\tilde u$ of the  KdV equation starting from a given solution $u$.
 
It follows from  \eqref{Tiden} that
\begin{equation}\label{TT}
L(\hbox{Ker}\, T)\subset \hbox{Ker}\, T.
\end{equation}
The existence of the Euclidean algorithm in the ring of differential operators \cite{ore} guarantees that  \eqref{TT} is a sufficient condition for $\tilde L$ to be a differential operator.

Suppose that the Jordan form of the operator $L$ acting on the finite-dimensional space 
$\hbox{Ker}\,T$ is diagonal. Then a basis of $\hbox{Ker}\,T$ is
given by  some functions $\Psi_1, \dots, \Psi_n$, such that
\begin{equation}\label{xd}
\frac{\partial^2}{\partial x^2}\Psi_i+u \Psi_i=\lambda_i^2 \Psi_i.
\end{equation}
If the functions $\Psi_1, \dots, \Psi_n$ are fixed, then, up to a left factor,  the equation $T(Y)=0$ is given by  the formula
$$
{\bf W}(\Psi_1, \dots, \Psi_n, Y)=0,
$$
where ${\bf W}$ is the Wronskian.  This implies
$$ a_{n-1}=-\frac{\partial}{\partial x} \log {{\bf W}(\Psi_1, \dots, \Psi_n)}.$$
and therefore the Darboux transformation has the following form:
\begin{equation}\label{newu}
\tilde u=u+2 \frac{\partial^2}{\partial x^2} \log {{\bf W}(\Psi_1,
\dots, \Psi_n)} .
\end{equation}
 
As usual, the $t$-dynamics of $\Psi_i$ is defined by the $A$-operator:
\begin{equation}\label{td}
\frac{\partial}{\partial t}(\Psi_i)=A(\Psi_i)\equiv \frac{\partial^3}{\partial
x^3}\Psi_i+\frac{3}{2} u \frac{\partial}{\partial
x}(\Psi_i)+\frac{3}{4} u_x \Psi_i.
\end{equation}

\begin{theorem}
Let $u(x,t)$ be any solution of the KdV equation  \eqref{kdv1}. If functions $\Psi_i$ satisfy  \eqref{xd} and   \eqref{td}, then the function $\tilde u(x,t)$ defined by  \eqref{newu} satisfies the KdV equation.
\end{theorem}
\begin{proof}
The Lax equation  \eqref{Lax} can be rewritten in the commutator form
$\ds [\frac{\partial}{\partial t}-A, \, L]=0.$ This implies
$\ds [\frac{\partial}{\partial t}-\tilde A, \, \tilde L]=0,$ where the differential operator $\tilde L$ is defined by  \eqref{tiL} and 
\begin{equation}\label{tilT}
\tilde A= T A T^{-1}+T_t T^{-1}.
\end{equation}
It suffices to check that the ratio of differential operators 
$T A+T_t$ and $T$ is a differential operator. This is equivalent to the fact that
\begin{equation}\label{id1}
T A (\Psi_i)+T_t (\Psi_i)=0
\end{equation}
for any $i$. We have $0=(T \Psi_i)_t=T_t \Psi_i+T(\Psi_i)_t.$ Substituting $-T(\Psi_i)_t$ for $T_t (\Psi_i)$ into  \eqref{id1} and using  \eqref{td}, we complete the proof.
\end{proof}
\begin{remark}
The numbers $\lambda_i$ from  \eqref{xd} are arbitrary parameters in the solution  \eqref{newu}.
\end{remark}

\begin{exercise} Prove that for any Jordan form of the operator $L$ acting on the finite-dimensional space $\hbox{Ker}\,T$ the condition 
\begin{equation}\label{ATcon}
\Big(\frac{\partial}{\partial t} - A\Big) \hbox{Ker}\,T \subset \hbox{Ker}\,T
\end{equation}
provides the fact that $\tilde A$, defined by \eqref{tilT}, is a differential operator.
\end{exercise}

\subsubsection{Solitons and rational solutions for KdV equation}
Let us start from the trivial solution $u(x,t)=0$ of the KdV
equation.
In this case condition \eqref{TT} means that $\hbox{Ker}\,T$ is any finite-dimensional vector space $\bf V$ of functions invariant with respect to the differential operator $\ds \frac{\partial^2}{\partial x^2}.$  The $t$-dynamics of ${\bf V}$ is defined by the condition 
$$
\Big(\frac{\partial}{\partial t} - \frac{\partial^3}{\partial x^3}\Big) {\bf V} \subset 
{\bf V}.
$$

In the generic case, when the Jordan form of the operator $\ds \frac{\partial^2}{\partial x^2}$ is diagonal, a basis of such a vector space is given by
$$
\Psi_i(x,t)=\exp{(\eta_i)}+c_i \exp{(-\eta_i)}, \qquad {\rm where} \qquad
\eta_i=\lambda_i\, x+\lambda_i^3\, t, \qquad i=1,\dots, n.
$$
The function
$$
\tilde u=2 \frac{\partial^2}{\partial x^2} \log {{\bf W}(\Psi_1,
\dots, \Psi_n)} 
$$
is called $n$-{\it soliton} solution of the KdV equation.

\begin{example} If $n=1$, we have
$$
\tilde u(x,t)=\frac{8 c_1 \lambda_1^2} {(e^{\lambda_1\,
x+\lambda_1^3\,t} +c_1 e^{-\lambda_1\, x-\lambda_1^3 \,t})^2 }
$$
\end{example}
 
\begin{example} The $2$-soliton solution for the KdV equation is given by
$$\tilde u(x,t)=\displaystyle (\lambda_2^2-\lambda_1^2) \frac{\ds\frac{8 c_1
\lambda_1^2} {(e^{\eta_1} +c_1 e^{-\eta_1})^2 }-\frac{8 c_2
\lambda_2^2} {(e^{\eta_2} +c_2 e^{-\eta_2})^2 }}{\ds \left(\lambda_1
\frac{c_1-e^{2 \eta_1}}{c_1+e^{2 \eta_1}}- \lambda_2
\frac{c_2-e^{2 \eta_2}}{c_2+e^{2 \eta_2}}\right)^2}.
$$
\end{example}

If the vector space $\bf V$ consists of polynomials in $x$, we get rational solutions of the KdV equation.   In the simplest case ${\rm dim}\,{\bf V}=1$ we have $\Psi_1=x$ and formula \eqref{newu} produces a stationary rational solution
$$
\ds \tilde u(x,t)= - \frac{2}{x^2}.
$$

\subsection{Gelfand-Dikii hierarchy and generalizations} 

Let 
\begin{equation}\label{AAAa}
L=D^n+\sum_{i=0}^{n-2} u_i D^i, \qquad A=\sum_{i=0}^m c_i L^{\frac{i}{n}}_{+}, \qquad c_i\in \C .
\end{equation}
In the same way as in Lemmas \ref{lem1} and \ref{lem3} it can be proved that the Lax equation  \eqref{Lax} is equivalent to a system of $n-1$ evolution equations for unknown functions $u_{n-2},\dots, u_0.$ Moreover, the systems generated by the same $L$-operator and $A$-operators of the form  \eqref{AAAa} with different $m$ and $c_i$ are higher symmetries for each other. This infinite set of compatible evolution systems  is called the {\it Gelfand-Dikii hierarchy} \cite{geldik}. If $n=2$, we get the KdV hierarchy described above. 

\subsubsection{Factorization of $L$-operator}

Relations between factorization of scalar differential Lax operators, Miura-type transformations and modified KdV-type systems were discussed, for example, in \cite{SokSha80, forgib, drsok}.  We are concerned here with the case of two factors only. 

Consider the following system of the Lax type equations
\begin{equation}\label{drsok}
M_t= A N-M B, \qquad N_t=B N-N A,
\end{equation}
where 
$$
M=D^r+w D^{r-1}+\sum_{i=0}^{r-2} u_i D^i, \qquad N=D^s-w D^{s-1}+\sum_{i=0}^{s-2} v_i D^i,
$$
$$
A=\sum_{i=0}^{m} c_i \Big((M N)^{\frac{i}{r+s}}\Big)_{+}, \qquad B=(M^{-1} A M)_{+}.
$$
System  \eqref{drsok} is related to Lax equation  \eqref{Lax}. Namely, if $M$ and $N$ satisfy   \eqref{drsok}, then $L=M N$ satisfies  \eqref{Lax}. 
\begin{proposition}
Relations  \eqref{drsok} are equivalent to a system of $r+s-1$ evolution equations in $w, u_i, v_i.$
\end{proposition}

\subsubsection{Reductions in differential $L$-operators}

Let us introduce the following notation:
$$
Q_1(n)\stackrel{def}{=}D^{2 n+1}+\sum_{i=0}^{n-1} u_i D^{2 i+1}+ D^{2 i+1} u_i, \qquad 
Q_2(n)\stackrel{def}{=}D^{2 n}+\sum_{i=0}^{n-1} u_i D^{2 i}+ D^{2 i} u_i,
$$
$$
Q_3(n)\stackrel{def}{=}D^{2 n-1}+\sum_{i=1}^{n-1} u_i D^{2 i-1}+ D^{2 i-1} u_i+u_0 D^{-1} u_0.
$$
We call $u_{n-1},\dots, u_0$ {\it functional parameters} of $Q_i(n)$.
By definition, we put
$$
Q_1(0)\stackrel{def}{=}D, \qquad Q_2(0)\stackrel{def}{=}1, \qquad Q_3(0)\stackrel{def}{=}D^{-1}.
$$
Notice that the operators $Q_1(n)$ and $Q_3(n)$ are skew-symmetric:  $Q_1(n)^{+}=-Q_1(n)$ and  $Q_3(n)^{+}=-Q_3(n)$. The operators  $Q_2(n)$ are symmetric: $Q_2(n)^{+}=Q_2(n).$  

There are deep relations between such operators and classical simple Lie algebras \cite[Section 7]{drsok}. The algebra $B_n$  corresponds to an operator of $Q_1(n)$-type,
while the algebras $C_n$ and $D_n$ correspond to operators of types $Q_2(n)$ and $Q_3(n)$.

\begin{theorem} \label{scala} {\rm (see \cite[Section 7]{drsok}, \cite{DriSok81b})} 
Suppose the operators $M=Q_i(r)$ and $N=Q_j(s)$, where $i,j \in \{1,2,3 \},$ have functional parameters $u_{r-1}, \dots, u_0$ and $v_{s-1}, \dots, v_0$, respectively. Then relations  \eqref{drsok}, where
$$
A=\sum_{i=0}^{m} c_i \Big(L^{\frac{2 i+1}{n}}\Big)_{+}, \qquad B=\Big(M^{-1} A M\Big)_{+}, 
$$
$L=M N$ and $n={\rm ord}\,L,$ are equivalent to a system of $r+s$ evolution equations in $u_{r-1},\dots,u_0$, $v_{s-1},\dots, v_0.$
\end{theorem}

Here we present several examples \cite{DriSok81b} with $r+s\le 2,$ where the operator $A$ has a minimal possible order. In the corresponding differential equations we perform some scalings of independent variables and unknown functions to reduce equations to a simple form. If any different transformations were applied, then we point out their form up to constants, which can be easily reconstructed by reader. 

\begin{example} In the cases 
\begin{itemize}
\item[a)] \qquad $L=D^2+u$, $\qquad A=\Big(L^{\frac{3}{2}}\Big)_+$, 
\item[b)] \qquad $L=(D^2+u) D^{-1}$, $\qquad A=\Big(L^{3}\Big)_+$, 
\item[c)] \qquad $L=(D^3+2 u D+u_x) D^{-1}$, $\qquad A=\Big(L^{\frac{3}{2}}\Big)_+$, 
\item[d)] \qquad $L=(D^3+2 u D+u_x) D$, $\qquad A=\Big(L^{\frac{3}{4}}\Big)_+$ 
\end{itemize}
we get the KdV equation  \eqref{kdv}.
\end{example}
 \begin{example} The cases 
\begin{itemize}
\item[a)] \qquad $L=D+u D^{-1} u$, $\qquad A=\Big(L^{3}\Big)_+$, 
\item[b)] \qquad $L=(D+u D^{-1} u) D$, $\qquad A=\Big(L^{\frac{3}{2}}\Big)_+$ 
\end{itemize}
give rise to the modified KdV equation 
$$
u_t=u_{xxx}+6 u^2 u_x.
$$
\end{example}
\begin{example} The operators 
 $$L=(D^2+u) D,  \qquad A=\Big(L^{\frac{5}{3}}\Big)_+$$  
 produce the Savada-Kotera  equation \cite{SawKot74}
 $$
u_t=u_5+5 u u_3+ 5 u_1 u_2+5 u^2 u_1.
 $$
\end{example}
\begin{example} In the case 
 $$L=D^3+2 u D+u_x,  \qquad A=\Big(L^{\frac{5}{3}}\Big)_+$$  
 we obtain the Kaup-Kupershmidt equation \cite{Kau80}
 $$
 u_t=u_5+10 u u_3+ 25 u_1 u_2+20 u^2 u_1.
 $$ 
\end{example}
 
\begin{example} The system
$$
u_t=v v_x, \qquad v_t=v_{xxx} +2 u v_x+v u_x
$$
corresponds to
 $$L=D^3+2 u D+u_x+v D^{-1} v,  \qquad A=L_{+}.$$  
\end{example}
\begin{example} The operators 
$$L=(D^4+ u D^2+D^2 u+v) D^{-1},  \qquad A=L_{+}$$
yield
$$
u_t=w_x, \qquad w_t=w_{xxx} + w u_x+u w_x, 
$$
where $w=v+\alpha \, v_{xx}$ for some constant $\alpha$.  
\end{example}
\begin{example} For 
$$L=(D^5+ u D^3+D^3 u+v D+D v) D^{-1},  \qquad A=\Big(L^{\frac{3}{4}}\Big)_{+}$$
we obtain
$$
u_t=-u_{xxx}+w_x-u u_x, \qquad w_t=2 w_{xxx} +u w_x, 
$$
where $w=v+\beta u_{xx}$ for some constant $\beta$.  
\end{example}
\begin{example} The operators 
$$L=(D^3+ 2 u D+u_x u+v D^{-1} v)\, D,  \qquad A=\Big(L^{\frac{3}{4}}\Big)_{+}$$
correspond to
$$
u_t=u_{xxx}+u u_x-v v_x, \qquad v_t=-2 v_{xxx} - u v_x. 
$$
\end{example}

Several more examples can be found in \cite{drsok,DriSok81b}.
 
\section{Matrix Lax pairs} 
\label{SectionMatrixLaxPairs}

\subsection{The NLS hierarchy }

The nonlinear Schr\"odinger equation (NLS equation) has the form $Z_t=i Z_{xx}+\vert Z\vert^2\, Z.$ After a (complex) scaling of $t$ and $Z$ the equation can be written as a system of two equations
\begin{equation}\label{nls1}
u_t=\frac{1}{2} \left(u_{xx}- 2\,u^2\,v\right),\qquad
v_t=\frac{1}{2} \left(-v_{xx}+ 2\,v^2\,u\right),
\end{equation}
where $u=Z, \, v=\bar Z.$ The Lax representation  \eqref{Lax} for  \eqref{nls1} is defined \cite{ZS} by 
$$
L=D + \left(%
\begin{array}{cc}
  1 & 0 \\
  0 & -1
\end{array}%
\right)\lambda+\left(%
\begin{array}{cc}
  0 & u  \\
  v & 0
\end{array}%
\right),
$$
$$
A = \left(%
\begin{array}{cc}
  1 & 0 \\
  0 & -1
\end{array}%
\right)\lambda^2+\left(%
\begin{array}{cc}
  0 & u  \\
  v & 0
\end{array}%
\right) \lambda+
\frac{1}{2} \left(%
\begin{array}{cc}
  -u v\, & -u_x \\
  v_x\, & u v
\end{array}%
\right).
$$
Notice that all matrix coefficients of $L$ and $A$ belong to the Lie algebra $\mathfrak{sl}_2.$
 
One can verify that the matrix $A$ obeys the following properties:
\begin{itemize}
\item[a)] \quad The commutator $[A,\,L]$ does not depend on $\lambda;$
\item[b)] \quad It has the
following matrix structure:
$$
[A,\,L]=\left(%
\begin{array}{cc}
  0\, & \ast \\
  \ast \, & 0
\end{array}%
\right).
$$
\end{itemize}
It is clear that if these properties hold for a matrix polynomial 
\begin{equation}\label{An}
A_n=\sum_{i=0}^{n} a_i \lambda^i, \qquad a_i\in \mathfrak{sl}_2,
\end{equation}
then the Lax equation $L_t=[A_n,\,L]$ is equivalent to a system of two evolution equations for $u$ and $v$.
 
\begin{problem}\label{problem2} How to describe all matrix
polynomials  \eqref{An} that satisfy the above two properties?
\end{problem}

\subsubsection{Formal diagonalization} 
 
\begin{theorem} {\rm \cite[Section 1]{drsok}}. There exists a unique series
$$
T=\IdentityMatrix+\left( \begin{array}{cc}
  0\, & \alpha_1 \\
  \beta_1 \, & 0
\end{array}%
\right) \frac{1}{\lambda}+\left( \begin{array}{cc}
  0\, & \alpha_2 \\
  \beta_2 \, & 0
\end{array}%
\right) \frac{1}{\lambda^2}+\cdots
$$
such that
$$
T^{-1} L T= L_0,
$$
where
$$
L_0=
D+ \left(%
\begin{array}{cc}
  1 & 0 \\
  0 & -1
\end{array}%
\right)\lambda+\left(%
\begin{array}{cc}
  \rho_0 & 0  \\
  0 & -\rho_0
\end{array}%
\right)+
\left( \begin{array}{cc}
  \rho_1\, & 0 \\
  0 \, & -\rho_1
\end{array}%
\right) \frac{1}{\lambda}+\left( \begin{array}{cc}
  \rho_2\, & 0 \\
  0 \, & -\rho_2
\end{array}%
\right) \frac{1}{\lambda^2}+\cdots
$$
\end{theorem}

\begin{proof} Equating the coefficients of $\lambda^0$ in $L T=T
L_0$, we get
$$
\left[\left(%
\begin{array}{cc}
  1 & 0 \\
  0 & -1
\end{array}%
\right),\, \left( \begin{array}{cc}
  0 & \alpha_1 \\
  \beta_1  & 0
\end{array}%
\right)\right]-\left(%
\begin{array}{cc}
  \rho_0 & 0  \\
  0 & -\rho_0
\end{array}%
\right)=\left(%
\begin{array}{cc}
  0 & u  \\
  v & 0
\end{array}%
\right).
$$
Hence $ \rho_0=0,$ $\ds \alpha_1=\frac{1}{2} u$ and
$\ds \beta_1=-\frac{1}{2} v$.
At each step we have a similar relation of the form
$$
\left[\left(%
\begin{array}{cc}
  1 & 0 \\
  0 & -1
\end{array}%
\right),\, \left( \begin{array}{cc}
  0 & \alpha_k \\
  \beta_k  & 0
\end{array}%
\right)\right]-\left(%
\begin{array}{cc}
  \rho_{k-1} & 0  \\
  0 & -\rho_{k-1}
\end{array}%
\right)= P_k, $$ where $P_k\in \mathfrak{sl}_2$ is a already known matrix. The
functions $\alpha_k, \beta_k, \rho_{k-1}$ are thus uniquely defined.
\end{proof}

\medskip

\begin{proposition} Let 
\begin{equation}\label{Bn}
B_n=T \left(%
\begin{array}{cc}
  1 & 0 \\
  0 & -1
\end{array}%
\right) T^{-1}\, \lambda^n,
\end{equation}
\begin{equation}\label{ABn}
A_n=(B_n)_{+}, 
\end{equation}
where by definition
$$
\left(\sum_{i=-\infty}^m a_i \lambda^i \right)_{+}\stackrel{def}{=}\sum_{i=0}^m
a_i \lambda^i, \qquad \quad
\left(\sum_{i=-\infty}^m a_i \lambda^i
\right)_{-}\stackrel{def}{=} \sum_{i=-\infty}^{-1} a_i \lambda^i.
$$
Then $A_n$ satisfies properties a) and b).
\end{proposition}
\begin{proof}
Since $[L,\,B_n]=0,$ we have
$$
[A_n, \,L]=-[\left(B_n \right)_{-},\,L].
$$
The left hand side is a polynomial in $\lambda$ whereas the right
hand side has the form
$$
\left(%
\begin{array}{cc}
  0\, & \ast \\
  \ast \, & 0
\end{array}%
\right)+\sum_{i=-\infty}^{-1} b_i \lambda^i.
$$
Hence
\begin{equation}\label{fg}
[A_n, \,L]=\left(%
\begin{array}{cc}
  0\, & f_n \\
  g_n \, & 0
\end{array}%
\right).
\end{equation}
\end{proof}
\begin{proposition} For any $n$ and $m$ the system of equations $$u_{\tau}=f_m, \qquad
v_{\tau}=g_m,$$ where $f_i$ and $g_i$ are defined by  \eqref{fg}, is a higher symmetry for the system $u_t=f_n, \quad
v_t=g_n.$
\end{proposition}
The proof is similar to the proof of Theorem \ref{Hier} for the KdV hierarchy. 
\begin{exercise} Prove the proposition.
\end{exercise}

The $A$-operator for the NLS equation is given by  \eqref{ABn} with $n=2.$ Formulas  \eqref{Bn},  \eqref{ABn} for arbitrary $n$ define the 
{\it NLS hierarchy}.
The next member of the NLS hierarchy  $$u_t=-\frac{1}{4} u_{xxx}+\frac{3}{2} v u u_x, \qquad 
v_t=-\frac{1}{4} v_{xxx}+\frac{3}{2} u v v_x
$$
corresponds to
$$
A_3=A_2 \lambda+\frac{1}{4} \left(%
\begin{array}{cc}
  v u_x-u v_x\, & u_{xx}-2 u^2 v \\
  v_{xx}-2 v^2 u\, & u v_x-v u_x
\end{array}%
\right).
$$
The reduction $v=u$ leads to the modified Korteweg-de Vries equation
$$\quad u_t=-\frac{1}{4} u_{xxx}+\frac{3}{2} u^2 u_x. $$

\subsubsection{Recursion operator for NLS equation}
In this section we follow the paper \cite[Section 3A]{GKS}. 
Since $$B_{n+1}=\lambda\, B_{n},$$ we have
$$
A_{n+1}=(\lambda\, B_{n})_{+}=\lambda\, (B_{n})_{+}+ (\lambda\,
(B_{n})_{-})_{+}. $$
The latter formula shows that
$$
A_{n+1}=\lambda\, A_{n}+R_{n},
$$
where $R_{n}$ does not depend on $\lambda$. Substituting this into
the Lax equation $L_{t_{n+1}}=[A_{n+1},\,L],$ we get
$$
L_{t_{n+1}}=\lambda \, L_{t_{n}}+ [R_{n},L]. 
$$
If $$ R_{n}=\left(
\begin{array}{cc}
  a_n\, & b_n \\
  c_n \, & -a_n
\end{array}%
\right),$$
then we find that
$$
b_{n}={1 \over 2}\,f_{n}, \qquad c_{n}=-{1 \over 2}\,g_{n}, \qquad 
a_{n}={ 1 \over 2}\,D^{-1}\,(v f_{n}+u g_{n}).
$$  
Therefore the recursion operator
$$
\cal{R}\left(
\begin{array}{cc}
  f_n \\
  g_n \,
\end{array}%
\right)=\left(
\begin{array}{cc}
  f_{n+1} \\
  g_{n+1} \,
\end{array}%
\right), \qquad \quad \left(
\begin{array}{cc}
  f_{1} \\
  g_{1} \,
\end{array}%
\right)=\left(
\begin{array}{cc}
  u_x \\
  v_x \,
\end{array}%
\right)$$
is given by
$$
\cal{R}=\left(
\begin{array}{cc}
  -{1 \over 2}\,D+u\, D^{-1}v\, &  u\,D^{-1}u \\
   -v\,D^{-1} v \, &  {1 \over 2}\, D-v\,D^{-1}u 
\end{array}%
\right). $$

The operator ${\cal R}^2$ gives rise to a recursion operator for the
mKdV equation by the reduction $v=u.$

\subsection{Generalizations}

Consider the operator
\begin{equation}\label{LL}
L=D+\lambda a+q(x,t),
\end{equation}
where $q$ and $a$ belong to a Lie algebra ${\cal G}$ and $\lambda$
is the spectral parameter. The constant element $a$  is supposed to
be such that
$$
{\cal G}={\rm Ker} \,(ad_{a}) \oplus {\rm Im} \,(ad_{a}).
$$
\begin{theorem}\label{DiaP} There exist unique series
\begin{eqnarray*}
u&=&u_{-1}\, \lambda^{-1}+ u_{-2}\, \lambda^{-2}+\cdots , \qquad \qquad 
u_{i} \in \,{\rm Im} \,(ad_{a}), \label{la2}  \\[3mm]
h&=&h_{0}+h_{-1}\, \lambda^{-1}+ h_{-2}\, \lambda^{-2}+\cdots , \qquad 
h_{i} \in \,{\rm Ker} \,(ad_{a}),
\end{eqnarray*}
such that
$$
e^{ad_{u}}\,(L) \stackrel{def}{=}  L+[u,L]+{1 \over 2}\,[u,[u,L]]+\cdots
=D_{x}+a\lambda+h.
$$
\end{theorem}

Let $b$ be a constant element of ${\cal G}$ such that $$\qquad [b,\,
{\rm Ker}\,(ad_{a})]=\{0\}.\qquad $$
  Since $$[b\, \lambda^{n}, D_{x}+a \lambda+h]=0,$$ we have
$[B_{b,n},\, L]=0$, where
$$
B_{b,n}= e^{-ad_{u}}\,(b\, \lambda^{n}). $$ Then the corresponding
$A$-operator of the form
$$
A_{b,n}=b\, \lambda^{n}+a_{n-1}\, \lambda^{n-1}+\cdots+a_{0}
$$
is defined by the formula
\begin{equation}\label{Bnn}
A_{b,n}=(B_{b,n})_{+}. 
\end{equation}

For the Lie algebra ${\cal G}=\mathfrak{sl}_2$ and $a={\rm diag}\,(1,-1)$ we get the NLS hierarchy. 

\begin{example} Let ${\cal G}=\mathfrak{gl}_m$, the $L$-operator has the form \eqref{LL}, where
$$
a=\left(\begin{array}{cc}
  {\bf 1}_{m-1} & 0 \\
  0 & -1 \\
\end{array}\right),  \qquad 
q=\left(\begin{array}{cc}
  0 & {\bf u} \\
  {\bf v}^t & 0 \\
\end{array}\right).
$$
Here ${\bf u}$ and ${\bf v}$ are column vectors. In this case 
$$
{\rm Ker} \,(ad_{a})=\Big\{\left(\begin{array}{cc}
  {\bf S} & 0 \\
  0 & s \\
\end{array}\right)  \Big\}, \qquad  {\rm Im} \,(ad_{a})=\Big\{\left(\begin{array}{cc}
  0 & {\bf u}_1 \\
  {\bf u}_2^t & 0 \\
\end{array}\right) \Big\},
$$
where ${\bf S}$ is an $(m-1)\times (m-1)$-matrix, $s$ is a scalar and ${\bf u}_i$ are column vectors.
Following the above diagonalization procedure, we find that the coefficients of the operator
$$
A_{a,2} = a \lambda^2 + s_1 \lambda + s_2
$$
are given by 
$$
s_1=\left(%
\begin{array}{cc}
  0 & {\bf u}  \\
  {\bf v}^t & 0
\end{array}%
\right), \qquad s_2=
\frac{1}{2} \left(%
\begin{array}{cc}
  -{\bf u} {\bf v}^t\, & -{\bf u}_x \\
 {\bf v}_x^t\, & \langle {\bf u}, \bf{v} \rangle
\end{array}%
\right).
$$
If $m=2$ the Lax pair coincides with the Lax pair for the NLS equation from Subsection 2.2.1. 
The corresponding non-linear integrable system is (up to scalings of $t,{\bf u}, {\bf v}$) the vector NLS equation \cite{Man}
\be \begin{array}{l}
{\bf u}_t={\bf u}_{xx}+2\langle {\bf u},{\bf v}\rangle \, {\bf u}, \qquad 
{\bf v}_t=-{\bf v}_{xx}-2\langle {\bf v},{\bf u}\rangle \,  {\bf v}.                    \label{ex1}
\end{array} \ee
\end{example}

\begin{example} Let  ${\cal G}=\mathfrak{gl}_m$, $a={\rm diag}\,(a_1,\dots,a_m)$,  $b={\rm diag}\,(b_1,\dots,b_m)$, where $a_i\ne a_j$ for $i\ne j$.  The equation corresponding to $A_1$ given by \eqref{Bnn} is called $m$-{\it wave equation}. It has the form
\begin{equation}\label{Nwave}
{\bf Q}_t={\bf P}_x+[{\bf Q},\,{\bf P}], 
\end{equation}
where ${\bf Q}$ and ${\bf P}$ are $m\times m$-matrices whose entries are related by 
$$
p_{ij}=\frac{b_i-b_j}{a_j-a_i}\, q_{ij}.
$$
\end{example}

Solutions of \eqref{Nwave} that do not depend on $x$ describe the dynamics of an $m$-dimensional rigid body \cite{man}.

\subsubsection{Relations between scalar to matrix Lax pairs}

The Gelfand-Dikii hierarchy (see Subsection 2.1.3) is defined by a scalar linear differential operator of order $n.$
Of course, it is not difficult to replace this operator by a matrix first order differential operator of the form 
\begin{equation}\label{calL}
{\cal L}=D+\Lambda +q,
\end{equation}
where
\begin{equation}\label{Lq}
\Lambda=\left(%
\begin{array}{cccc}
  0 & 0 & \cdots & \lambda\\
  1 & 0 & \cdots & 0\\
  \cdot & \cdot & \cdots & \cdot\\
    0 & \cdots & 1 & 0
\end{array}%
\right), \qquad 
q=\left(%
\begin{array}{cccc}
  0 & 0 & \cdots & u_1\\
  0 & 0 & \cdots & u_2\\
  \cdot & \cdot & \cdots & \cdot\\
   0 & 0 & \cdots & u_n
\end{array}%
\right).
\end{equation}
Let us consider operators  \eqref{calL}, where $q$ is an arbitrary upper-diagonal matrix. Any gauge transformation $\bar {\cal L}=N {\cal L} N^{-1},$ where $N$ is a function with values in the group of upper triangular matrices with ones on the diagonal, preserves the class of such operators. It turns out that the matrix $q$ defined by  \eqref{Lq} is one of the possible canonical forms with respect to this gauge action. 

The approach \cite[Section 6]{drsok} based on this observation allows one to construct an analog of the Gelfand-Dikii hierarchy for any Kac-Moody algebra $G$.  

Let $e_i, f_i, h_i,$ where $i=0,\dots,r,$ be the canonical generators of a Kac-Moody algebra $G$ with the commutator relations
$$
[h_i, \,h_j]=0, \qquad [e_i,\, f_j]=\delta_{ij}\,h_i, \qquad 
[h_i,\, e_j]=A_{ij}\, e_j, \qquad [h_i,\, f_j]=-A_{ij}\, f_j,
$$
where $A$ is the Cartan matrix of the algebra $G.$ 

Let us take $\sum_{i=0}^{r} e_i$ for the element $\Lambda$ in  \eqref{calL}. The potential $q$ depends on a choice of a vertex $c_m$ for the Dynkin diagram of $G.$ We consider the gradation $G=\oplus G_i$ such that $e_m\in G_1,$ $f_m\in G_{-1}$ and the remaining canonical generators belong to $G_0.$ It is well-known that ${\cal G}=G_0$ is a semi-simple finite-dimensional Lie algebra. The potential $q$ is a generic element of the Borel subalgebra ${\cal B}\subset {\cal G}$ generated by $f_i, h_i$, where $i \ne m$. 

If ${\cal L}$ is an operator of the form  \eqref{calL} and $S$ belongs to the corresponding nilpotent subalgebra ${\cal N} \subset  {\cal B}$, then the operator 
$$
\bar {\cal L}=e^{ad\,S} ({\cal L})
$$
has the same form  \eqref{calL} (with different $q$). This follows from the fact that $\,\,[{\cal N}, {\cal B}] \subset {\cal N},\,\,$ $[{\cal N}, e_m]=\{0\}$, \,\, $[{\cal N}, e_i]\subset {\cal B}$.

Any canonical form under these gauge transformations gives rise to a system of $r$ evolution equations. The systems corresponding to different canonical forms are related by invertible polynomial transformations of unknown functions. 

Moreover, any ${\cal L}$-operator  \eqref{calL} generates a commutative hierarchy
of integrable systems. The corresponding $A$-operators can be constructed by a formal diagonalization procedure, which generalizes the construction from Theorem \ref{DiaP}. 

It was proved in \cite{drsok} that the systems related to $L$-operators of the form  \eqref{calL} include the systems from Theorem \ref{scala}. 

For further generalizations see \cite{GHM,GolSok39,kac}.

\section{Decomposition of loop algeras and  Lax pairs}

In all classes of Lax representations described above, $L$-operators are polynomials in the spectral parameter $\lambda$. However, there exist important examples, where $\lambda$ is a parameter on an elliptic curve or on its degenerations \cite{skl, ZaMik}.

An algebraic curve of genus $g > 1$ appears in the following
\begin{example} \cite{golsok} Consider the vector equation 
\begin{equation}\label{lanlif}
{\bf u}_t=\Big({\bf u}_{xx}+{3\over 2} \langle{\bf u}_x, \ {\bf u}_x\rangle\, {\bf u} \Big)_x + 
{3\over 2} \langle{\bf u}, \, {\bf R}\, {\bf u}\rangle\, {\bf u}_x, \qquad \vert {\bf u}\vert=1,  
\end{equation}
where ${\bf u}=(u^1,\dots, u^N)$, ${\bf R}={\rm diag}\,(r_1,\dots, r_N),$  and $\langle\cdot, \ \cdot\rangle$ is the standard scalar product. In the case $N=3$ this equation is a higher symmetry of the famous integrable Landau-Lifshitz equation.
\begin{equation}\label{ll}
{\bf u}_t= {\bf u} \times  {\bf u}_{xx}+ {\bf R} \,{\bf u} \times {\bf u}, \qquad \vert {\bf u}\vert=1. 
\end{equation} 
Here $\times$ stands for the cross product.  It is interesting that for $N\ne 3$ all symmetries of equation  \eqref{lanlif}  have odd orders. In particular, the equation has no symmetry of order 2. 

Equation  \eqref{lanlif} possesses a Lax representation with 
\begin{equation}\label{LAlif}
L=D + \left(\begin{array}{cc}
0\,&\Lambda {\bf u} \\ 
{\bf u}^T \Lambda\,&0  
\end{array} \right), \quad 
\end{equation}
Here   $$\Lambda={\rm diag}(\lambda_{1},\lambda_{2},\cdots,\lambda_{N})$$ is a matrix defined by
$$
\Lambda^2=\frac{{\bf 1}}{\lambda^2} - {\bf R}.
$$
It is clear that  
$$
\lambda_1^2+r_1=\lambda_2^2+r_2=\cdots=\lambda_N^2+r_N.
$$
For generic $\lambda_i, r_i$ this algebraic curve has genus $g=1+(N-3)\, 2^{N-2}.$
 In the case $N=3$ such a form of the elliptic spectral curve has been used in \cite{skl,FadTah}. 
\end{example}

A class of Lax operators related to algebraic curves of genus $g > 1$ was introduced in \cite{Kri, Kri1}.

\subsubsection{Factoring subalgebras}

If we don't want to fix {\it a priori} the $\lambda$-dependence in Lax operators, we may assume that $L$ is a Laurent series in $\lambda$ with coefficients being elements of a finite-dimensional Lie algebra  ${\cal G}$.

The Lie algebra   ${\cal G}((\lambda))$ of formal series of the form
$$
{\cal G}((\lambda)) = \Big\{ \sum^{\infty}_{i=-n} g_{i} \lambda^{i}\quad
\vert \quad g_{i}\in {\cal G}, \quad n\in {\mathbb Z} \Big\}
$$
is called the (extended) {\it loop algebra} over  ${\cal G}$.

If ${\cal G}$ is semi-simple, then the formula 
\begin{equation} \label{forma}
\langle X(\lambda),\, Y(\lambda)\rangle = {\rm res}
\Big(X(\lambda),\,Y(\lambda)\Big), \qquad X(\lambda),Y(\lambda)
\in {\cal G}((\lambda))
\end{equation}
defines an invariant non-degenerate bi-linear form on  ${\cal G}((\lambda))$. 
Here $(\cdot,\cdot)$ is the non-degenerate invariant Killing form on  ${\cal G}$,  $\,{\rm res}\,P$ stands for the coefficient of $\lambda^{-1}$
in a (scalar) Laurent series $P$. The invariance of the form means that
$$
\langle[a,\,b],\, c\rangle= -\langle b,\, [a,\,c]\rangle
$$
for any $a,b,c \in {\cal G}((\lambda)).$

If we assume that $L$ and $A$ in Lax equation ({\ref{Lax}}) are elements of  ${\cal G}((\lambda))$, then  ({\ref{Lax}}) is equivalent to an infinite set of evolution equations. To get a finite system of PDEs we need some additional assumptions on the structure of $L$ and $A$. 

The basic ingredient for constructing of  Lax pairs in ${\cal G}((\lambda))$ is a vector space decomposition  (see \cite{chered,skryp})
\begin{equation}\label{decompgen}
{\cal G}((\lambda))={\cal G}[[\lambda]]\oplus {\cal U},
\end{equation}
 where ${\cal G}[[\lambda]]$ is the subalgebra of all Taylor series and ${\cal U}$ is a so called {\it factoring}, or {\it complementary}, Lie subalgebra. Obviously, the subalgebra ${\cal G}[[\lambda]]$ is isotropic with respect to the form  \eqref{forma}. 
 
 Let us denote by $\pi_{+}$ and $\pi_{-}$ the projection operators onto $\cal U$ and  ${\cal G}[[\lambda]],$ respectively.

The following statement is evident: 
 \begin{lemma}\label{prope} Let  ${\cal U}$ be a factoring subalgera. Than
   for any principle part $P=\sum_{i=-n}^{-1} g_{i}\lambda^{i},$ where $g_{i}\in {\cal G}$, there exists  a unique element $\bar P\in {\cal U}$ of the form $\bar P=P+O(1).$
 \end{lemma}
 \begin{example} The simplest factoring subalgebra is given by
 \begin{equation}\label{polyn}
{\cal U}^{st}=\Big\{ \sum^{n}_{i=1} g_{i} \lambda^{-i}\quad \vert
\quad g_{i}\in {\cal G}, \quad n \in {\mathbb N} \Big\}.
\end{equation}
This subalgebra is called {\it standard}.
 
 \end{example}
 
 Two factoring subalgebras are called {\it equivalent} if they are related by a transformation of the parameter $\lambda$ of the form
 \begin{equation}\label{replam}
\lambda \rightarrow \lambda +k_{2} \lambda^{2}+k_{3}
\lambda^{3}+\cdots , \qquad k_i\in \C,
\end{equation}
or by an automorphism of the form 
\begin{equation}\label{avt}
\exp{(\hbox{ad}_{g_{1} \lambda+g_{2} \lambda^{2}+\cdots})}, \qquad
g_{i}\in
 {\cal G}.
\end{equation}
It is clear the  \eqref{replam} and  \eqref{avt} preserve the subalgebra ${\cal G}[[\lambda]]$.  

Suppose that $r$-dimensional Lie algebra ${\cal G}$ is semi-simple. Let ${\bf e}_1,\dots,{\bf e}_r$ be a basis in  ${\cal G}$. 
According  Lemma \ref{prope} for any $i$ there exists a unique element ${\bf E}_i \in {\cal U}$ such that 
\begin{equation}\label{ei}
{\bf E}_i=\frac{{\bf e}_i}{\lambda}+O(1).
\end{equation}
\begin{proposition}
The elements ${\bf E}_i$ generate ${\cal U}$.
\end{proposition}
\begin{proof}
We have to show that for any $i,k$ an element ${\bf E}_{ik}$ of the form 
$${\bf E}_{ik}=\frac{{\bf e}_i}{\lambda^k}+O(-k+1)$$
can be obtained as a commutator of length $k$ of the elements ${\bf E}_j$, where $j=1,\dots,r.$ The proof is by induction on $k$. The induction step follows from the well-known property of the semi-simple Lie algebras: 
$[{\cal G},\,{\cal G}]={\cal G}.$
\end{proof}

If we take generic elements of the form  \eqref{ei}, the Lie subalgebra they generate will contain Taylor series. All of them should be equal to zero. This imposes  strong restrictions on generators  \eqref{ei}.

\begin{lemma}\label{lem255} A subalgebra ${\cal U}$ is factoring iff for any $k$ the dimension $d_k$ of the vector space $V_k$ of all elements from ${\cal U}$ of the form  
$$ \sum^{\infty}_{i=-n} g_{i} \lambda^{i}, \qquad  n \le k$$
 is the same as for the standard subalgebra \eqref{polyn}.
\end{lemma}
\begin{proof} Let ${\bf F}_{ij},$ where $i\le r,\, j\le k,$ be elements of ${\cal U}$ such that 
$$
{\bf F}_{ij}=\frac{{\bf e}_i}{\lambda^j}+O(1).
$$
It is clear that  ${\bf F}_{ij}$ form a basis of $V_k.$ 
\end{proof}

\begin{conjecture} Let $\cal G$ be a simple Lie algebra not isomorphic to $\mathfrak{sl}_2$. Then elements ${\bf E}_i \in {\cal G}((\lambda))$ of the form \eqref{ei} generate a factoring subalgebra iff the dimension of the vector space spanned by $[{\bf E}_{j},\, {\bf E}_k]$ and ${\bf E}_i$ is equal to $2\, {\rm dim}\,{\cal G}.$

\end{conjecture}

When the factoring subalgebra ${\cal U}$ is isotropic with respect to  \eqref{forma}, the description of factoring subalgebras is closely related to a  classification of the Yang-Baxter $r$-matrices \cite{beldr}.  Without this assumption the problem has not been deeply considered yet.   
In Subsection 2.3.1 we solve it for ${\cal G}=\mathfrak{so}_3$ \cite{sokol1}.

\subsubsection{Multiplicands}

\begin{definition}
A (scalar) Laurent series
$$
{\bf m}=\sum_{i=-n}^{\infty} c_{i} \lambda^{i}, \qquad c_{i}\in \C,
$$
is called a {\it   multiplicand} of $\cal U$ if ${\bf m}\,{\cal U}\subset {\cal
U}$. The number $n$ is called {\it the order} of the multiplicand
${\bf m}$. If ${\cal U}$ admits a multiplicand {\bf m} of order $n=1,$ then
${\cal U}$ is called {\it homogeneous}.
\end{definition}

Let ${\cal G}$ be a simple Lie
algebra.
The following construction allows us to associate an algebraic curve with any factoring subalgebra ${\cal U}.$

\begin{theorem} {\rm \cite{ostap}}  For any factoring subalgebra the following
statements are fulfilled:
\begin{itemize}
\item[ i)] multiplicands of negative orders do not exist;
\item[ ii)] the complement of the set of all multiplicand orders with respect to
the set of natural numbers is finite.
\end{itemize}
\end{theorem}

It follows from the statement ii) that any two multiplicands are related by an algebraic relation. So, the set of all multiplicands is isomorphic to a coordinate ring 
of some algebraic curve.  Examples are given in  Subsection 2.3.1. 

This canonical relation between factoring subalgebras and algebraic curves allows one to use methods of algebraic geometry for the  investigation of factoring subalgebras. 

\subsection{Factoring subalgebras for ${\cal G}=\mathfrak{so}_3$}

In this section we follow the paper \cite{sokol1}.

Consider the standard basis  
$$
{\bf e_{1}}=\left(%
\begin{array}{ccc}
  0 & 1 & 0 \\
  -1 & 0 & 0 \\
  0 & 0 & 0
\end{array}%
\right), \qquad
{\bf e_{2}}=\left(%
\begin{array}{ccc}
  0 & 0 & 1 \\
  0 & 0 & 0 \\
  -1 & 0 & 0
\end{array}%
\right), \qquad
{\bf e_{3}}=\left(%
\begin{array}{ccc}
  0 & 0 & 0 \\
  0 & 0 & 1 \\
  0 & -1 & 0
\end{array}%
\right)
$$
in $\mathfrak{so}_3.$ Let ${\cal U}$ be a factoring subalgebra.  Define elements ${\bf E}_i\in {\cal U} $ by  \eqref{ei}.

Automorphisms  \eqref{avt} are 
orthogonal transformations, which are Taylor series in $\lambda.$
The functions
$$
\vert {\bf E}_{1}\vert^{2}, \,\quad \vert {\bf E}_{2}\vert^{2},
\quad \,\vert {\bf E}_{3}\vert^{2}, \,\quad ({\bf E}_{1},\,{\bf
E}_{2}), \,\quad ({\bf E}_{1},\,{\bf E}_{3}), \,\quad ({\bf
E}_{2},\,{\bf E}_{3}),
$$
where
$$
\left(\sum_i x_{i} {\bf e_{i}}, \, \sum_j y_{j} {\bf e_{j}}
\right)=\sum_i x_{i} y_{i},
$$
are invariants for the transformations  \eqref{avt}. 

\begin{proposition}\label{prop266}  For 
any factoring subalgebra the following relations hold:
\begin{equation}
\begin{array}{l}
\label{VSmaincond1}
\left(%
\begin{array}{c}
  \left[{\bf E}_1,\,\left[{\bf E}_2,\,{\bf E}_3\right]\right]\\
  \left[{\bf E}_3,\,\left[{\bf E}_1,\,{\bf E}_2\right]\right]\\
  \left[{\bf E}_2,\,\left[{\bf E}_3,\,{\bf E}_1\right]\right]
\end{array}%
\right)={\bf A}\, \left(
\begin{array}{c}
  \left[{\bf E}_{3},\, {\bf E}_{1}\right]\\
  \left[{\bf E}_{1},\, {\bf E}_{2}\right]\\
  \left[{\bf E}_{2},\, {\bf E}_{3}\right]
\end{array}
\right)+{\bf B}\,\left(
\begin{array}{c}
  {\bf E}_{2}\\
  {\bf E}_{3}\\
  {\bf E}_{1}
\end{array}
\right),\\[10mm]
\left(
\begin{array}{c}
  \left[{\bf E}_3,\,\left[{\bf E}_2,\,{\bf E}_3\right]\right]+
  \left[{\bf E}_1,\,\left[{\bf E}_1,\,{\bf E}_2\right]\right] \\
  \left[{\bf E}_1,\,\left[{\bf E}_3,\,{\bf E}_1\right]\right]+
  \left[{\bf E}_2,\,\left[{\bf E}_2,\,{\bf E}_3\right]\right]\\
  \left[{\bf E}_2,\,\left[{\bf E}_1,\,{\bf E}_2\right]\right]+
  \left[{\bf E}_3,\,\left[{\bf E}_3,\,{\bf E}_1\right]\right]
\end{array}
\right)={\bf C}\, \left(
\begin{array}{c}
  \left[{\bf E}_{3},\, {\bf E}_{1}\right] \\
  \left[{\bf E}_{1},\, {\bf E}_{2}\right]\\
  \left[{\bf E}_{2},\, {\bf E}_{3}\right]
\end{array}
\right)+{\bf D}\,\left(
\begin{array}{c}
  {\bf E}_{2} \\
  {\bf E}_{3}\\
  {\bf E}_{1}
\end{array}
\right),
\end{array}
\end{equation}
where 
\begin{equation}
\label{VSmatcond1}
\begin{array}{c}
{\bf A}=\left(%
\begin{array}{ccc}
  -u & w & 0 \\
  u & 0 & -v \\
  0 & -w & v
\end{array}%
\right), \qquad {\bf B}=\left(%
\begin{array}{ccc}
  -\alpha & \beta & 0 \\
  \alpha & 0 & -\gamma \\
  0 & -\beta & \gamma
\end{array}%
\right),\\[8mm]
{\bf C}=\left(%
\begin{array}{ccc}
  x & v & -w \\
  -v & y & u \\
  w & -u & z
\end{array}%
\right), \qquad {\bf D}=\left(%
\begin{array}{ccc}
  \varepsilon & \gamma & -\beta \\
  -\gamma & \tau & \alpha \\
  \beta & -\alpha & \delta
\end{array}%
\right)
\end{array}
\end{equation}
are constant matrices. Moreover  ${\rm tr}\,\, {\bf
C}={\rm tr}\,\, {\bf D}=0$ and 
\begin{equation}
\label{VSlincomb} c_{1} {\bf A}+c_{2} {\bf B}=0, \qquad  c_{1}
{\bf C}+c_{2} {\bf D}=0 
\end{equation}
for some constants $c_1,c_2.$
\end{proposition}
\begin{proof}
The coefficients of $\lambda^{-3}$ in the expressions from the left hand side of   \eqref{VSmaincond1} are equal to zero. Therefore, the expressions  should be linear combinations of 
$${\bf E}_1,\quad {\bf E}_2,\quad {\bf E}_3, \quad  [{\bf E}_1,\,{\bf E}_2], \quad [{\bf E}_3,\,{\bf E}_1], \quad  [{\bf E}_2,\,{\bf E}_3].$$
The relations between the coefficients of these linear combinations follow from  Lemma \ref{prope}.
\end{proof}
\begin{remark} 
Proposition \ref{prop266} means that $d_3=9$ {\rm (}see Lemma \ref{lem255}{\rm )}.
\end{remark}

\begin{example}\label{example24}
For the standard factoring subalgebra  \eqref{polyn} conditions  \eqref{VSmaincond1} are fulfilled with ${\bf A}={\bf B}={\bf
C}={\bf D}=0.$ In this case, 
$$
{\bf E}_{i}=\frac{{\bf e}_{i}}{\lambda}, \qquad i=1,2,3.
$$
It is clear that the algebra of multiplicands is generated by 
$\ds x=\frac{1}{\lambda}$ and therefore the corresponding algebraic curve is a straight line. 
\end{example}

\begin{example}\label{example25}
Suppose that
$$
({\bf E}_{3},{\bf E}_{1})=-\alpha, \qquad ({\bf E}_{1},{\bf
E}_{2})=-\beta, \qquad ({\bf E}_{2},{\bf E}_{3})=-\gamma,
$$
\begin{equation}\label{VSskl2}
\vert {\bf E}_{3}\vert^{2}-\vert {\bf E}_{1}\vert^{2}=\varepsilon,
\qquad \vert {\bf E}_{1}\vert^{2}-\vert {\bf E}_{2}\vert^{2}=\tau,
\qquad \vert {\bf E}_{2}\vert^{2}- \vert {\bf
E}_{3}\vert^{2}=\delta,
\end{equation}
 where $\alpha,\beta,\gamma,\delta, \varepsilon,\tau$ are fixed constants such that $\varepsilon+\tau+\delta=0$. It follows from \eqref{VSskl2} that we may implement the spectral parameter $\lambda$ by formulas
 $$
\vert{\bf E}_{1}\vert=\frac{\sqrt{1-p \lambda^{2}}}{\lambda},
\qquad \vert{\bf E}_{2}\vert=
\frac{\sqrt{1-q\lambda^{2}}}{\lambda}, \qquad \vert{\bf
E}_{3}\vert= \frac{\sqrt{1-r\lambda^{2}}}{\lambda},
$$
where $\varepsilon=p-r$,   $\tau=q-p$,  $\delta=r-q$. The elements ${\bf E}_i$ of the form 
$$
{\bf E}_1=c_1 {\bf e}_1, \qquad {\bf E}_2=c_2 {\bf e}_1+c_3 {\bf e}_2, \qquad 
{\bf E}_3=c_4 {\bf e}_1+c_5 {\bf e}_2+c_6 {\bf e}_3, \qquad c_i\in \C((\lambda)),
$$
 can be easily reconstructed. 
 
 One can verify that such 
elements ${\bf E}_i$ satisfy  \eqref{VSmaincond1},  \eqref{VSmatcond1} with ${\bf A}={\bf C}=0$ and generate a factoring subalgebra. 
This factoring subalgebra is isotropic with respect to the form \eqref{forma}. 

The expressions $X_{i}(\lambda)=\vert{\bf E}_{i}\vert$
are functions 
on the elliptic curve
\begin{equation}\label{VSellip}
X_{1}^{2}+p=X_{2}^{2}+q=X_{3}^{2}+r.
\end{equation}

The functions
$$
x=\frac{1}{\lambda^2}, \qquad y=\frac{\sqrt{(1-p \lambda^2)(1-q
\lambda^2) (1-r \lambda^2)}}{\lambda^3} 
$$
are multiplicands of $\cal U$ of order 2 and 3, respectively. 
For example, in the special case $\alpha=\beta=\gamma=0$ we have 
\begin{equation}\label{EE1}
{\bf E}_{1} =\frac{\sqrt{1-p \lambda^{2}}}{\lambda}\,{\bf e}_{1}, \qquad {\bf E}_{2} =\frac{\sqrt{1-q\lambda^{2}}}{\lambda}\,{\bf e}_{2}, \qquad {\bf E}_{3} =\ \frac{\sqrt{1-r\lambda^{2}}}{\lambda}\,{\bf e}_{3}, 
\end{equation}
and 
$$
x \, {\bf E_2}=[[ {\bf E_1}, {\bf E_2}],\,  {\bf E_1}]+p {\bf E_2}, \qquad 
y \, {\bf E_2}=[[ {\bf E_2}, {\bf E_3}],\,[ {\bf E_1}, {\bf E_2}]] 
$$
and so on. 
The corresponding algebraic curve is elliptic:
$$
y^2=(x-p) (x-q) (x-r).
$$

\end{example}

\begin{example}\label{example26} Let
\begin{equation}\label{VScasec0}
\begin{array}{l}
\vert {\bf E}_{1}\vert^{2}=(\mu-r)(\mu-q)-u^{2}, \qquad
\vert {\bf E}_{2}\vert^{2}=(\mu-r)(\mu-p)-v^{2}, \\[3mm]
\vert {\bf E}_{3}\vert^{2}=(\mu-q)(\mu-p)-w^{2}, \qquad
({\bf E}_{1},\,{\bf E}_{2})=w (\mu-r)+u v, \\[3mm]
({\bf E}_{1},\,{\bf E}_{3})=v (\mu-q)+u w, \qquad \quad ({\bf
E}_{2},\,{\bf E}_{3})=u (\mu-p)+v w,
\end{array}
\end{equation}
 where $\mu=\lambda^{-1}$, $p,q,r,u,v,w$ are arbitrary parameters. The elements ${\bf E}_i,\, i=1,2,3$ satisfy  \eqref{VSmaincond1},  \eqref{VSmatcond1} with $x=p-r, y=q-p$, $z=r-q,$ ${\bf B}={\bf D}=0$ and generate a factoring subalgebra.  
 
If $u=v=w=0$, then ${\bf E}_i$ are given by
\begin{equation}\label{EEE1}
\begin{array}{c}
\ds {\bf E_{1}}= \frac{\sqrt{(1-r
\lambda)(1-q\lambda)}}{\lambda}\,\,{\bf e_{1}},
\qquad
{\bf E_{2}}=
\frac{\sqrt{(1-r\lambda)(1-p\lambda)}}{\lambda}\,\,{\bf e_{2}},
\\[6mm]
\ds {\bf E_{3}}=
\frac{\sqrt{(1-q\lambda)(1-p\lambda)}}{\lambda}\,\,{\bf e_{3}}\,.
\end{array}
\end{equation}
At first glance, we deal with the functions 
$$
X_{1}=\sqrt{\frac{(1-r \lambda)}{\lambda}}, \qquad
X_{2}=\sqrt{\frac{(1-q \lambda)}{\lambda}}, \qquad
X_{3}=\sqrt{\frac{(1-p \lambda)}{\lambda}}
$$
on the elliptic curve  \eqref{VSellip}, but 
in fact ${\bf E}_i$ depend on the products
$$
Z_1=X_2 X_3, \qquad Z_2=X_1 X_3, \qquad Z_3=X_1 X_2
$$
only. The corresponding algebraic curve can be written as 
$$
\frac{Z_1 Z_2}{Z_3}+p=\frac{Z_1 Z_3}{Z_2}+q=\frac{Z_2 Z_3}{Z_1}+r.
$$
This curve is rational. Indeed,
substituting 
$$
Z_{3}=\frac{(q-r)Z_{2} Z_{1}}{Z_{2}^{2}-Z_{1}^{2}}
$$
into the curve, we get
$$
(Z_{2}^{2}-Z_{1}^{2})^{2}+a^{2} Z_{2}^{2}-b^{2}Z_{1}^{2}=0,
$$
where $ a^{2}=(r-q)(q-p), \quad b^{2}=(r-p)(q-p).$ The latter curve admits the rational parameterization
$$
Z_{1}=\frac{a(t^{3}+S t)}{t^{4}+K t^{2}+S^{2}}, \qquad Z_{2}=\frac{b(t^{3}-S t)}{t^{4}+K t^{2}+S^{2}}, 
$$
where 
$$
S=\frac{(a^{2}-b^{2})^{2}}{4 a^{2} b^{2}}, \qquad
K=\frac{a^{4}-b^{4}}{2 a^{2} b^{2}}.
$$
The algebra of multiplicands of the factoring subalgebra is generated by $\ds x=\frac{1}{\lambda}.$ For example,
$$x {\bf E}_1=[{\bf E}_3, \,{\bf E}_2]+p {\bf E}_1. $$

\end{example}
\begin{example}\label{example27} Let
$$
{\bf E}_{i}=\frac{{\bf e}_{i}}{\lambda}+\nu\, [{\bf V},\,{\bf e}_{i}]+
\frac{1}{2}[{\bf V},\,[{\bf V},\,{\bf e}_{i}]],
$$
where
$$
{\bf V}=v_{1} {\bf e}_{1}+v_{2} {\bf e}_{2}+v_{3} {\bf e}_{3},
$$
$\nu$ and $v_i$ are parameters. 
The constants in
 \eqref{VSmaincond1}-- \eqref{VSlincomb} are given by 
$$
u=v_{1} v_{3}, \quad v=v_{2}v_{3}, \quad w=v_{1}v_{2}, \quad
x=v_{1}^{2}-v_{3}^{2}, \quad y=v_{2}^{2}-v_{1}^{2}, \quad
z=v_{3}^{2}-v_{2}^{2},
$$
$$
c_{1}=-1, \qquad
c_{2}=\nu^{2}+\frac{\Delta}{4}, \qquad  \Delta=v_{1}^{2}+v_{2}^{2}+v_{3}^{2}.
$$
The elements ${\bf E}_i$ generate a factoring subalgebra. The multiplicands of second and third order are given by
$$
x=\frac{1}{\lambda^2} - \frac{\Delta}{\lambda}, \qquad y=\frac{1}{\lambda^3} + \frac{\Delta (4\nu^2-3 \Delta)}{4 \lambda}.
$$
\begin{exercise}  
 Verify that $x$ and $y$ are related by a degenerate elliptic curve with canonical form
$$\bar y^2=4 (\bar x - a)(\bar x - b)^2,$$
where
$$
a = \frac{2}{3}\nu^2 \Delta, \qquad  b = -\frac{1}{3}\nu^2 \Delta.
$$
\end{exercise} 
\end{example}

\begin{theorem}
Any factoring subalgebra for ${\cal G}=\mathfrak{so}_3$ is equivalent to one from Examples \ref{example24}--\ref{example27}.  
\end{theorem}
\begin{proof}
For a proof see  \cite{sokol1}. 
\end{proof}

A classification of factoring subalgebras for the semi-simple Lie algebra  ${\cal G}=\mathfrak{so}_4$ is important for applications.  A class of factoring subalgebras was constructed in \cite{efim}.
\begin{op}
Describe all factoring subalgebras for ${\cal G}=\mathfrak{so}_4$.
\end{op}
 
 \subsection{Integrable top-like systems} 
 
As it was mentioned in Remark \ref{rem11}, one may assume that  the $A$-operator in  \eqref{Lax} belongs to  $\cal G$ while $L$ belongs to a module over $\cal G.$ In this section we assume that ${\cal G}$ is semi-simple.
 
For integrable top-like systems the 
 $A$-operator in  \eqref{Lax} belongs to $\cal U$ and $L$ belongs to the orthogonal complement ${\cal U}^{\perp}$ with respect to the scalar product  \eqref{forma}. It follows from its invariance that   ${\cal U}^{\perp}$ is a module over  ${\cal U}$.
\begin{exercise} Prove that  ${\cal U}^{\perp}$ does not contain non-zero Taylor series.
\end{exercise}
 
It was shown in \cite[Theorem 2.3]{sokolgok1} that in this case natural Hamiltonian structures arise.
 
We say that an $L$-operator has order $k$ if $L \in {\cal O}_k \stackrel{def}{=} \lambda^{-k}{\cal G}[[\lambda]]\bigcap {\cal U}^{\perp}.$ To construct $A$-operators we generalize the scheme of Subsections 2.1.2, 2.2.1. Namely, we find elements of ${\cal G}((\lambda))$ that commute with $L$ and project them onto $\cal U.$   

For the sake of simplicity we assume that  ${\cal G}$ is embedded  into a matrix algebra. Suppose that  $B_{ij}=\lambda^i L^j$ belongs to ${\cal G}((\lambda))$.

\begin{proposition} Suppose that $L \in {\cal O}_k$. Then \begin{itemize}
\item[ i)] $[\pi_{+}(B_{ij}),\, L] \in {\cal O}_k,$
\item[ii)] for any $i,j,p,q$ the Lax equations 
$$
L_t=[\pi_{+}(B_{ij}),\, L]
$$
and 
$$
L_{\tau}=[\pi_{+}(B_{pq}),\, L]
$$
are infinitesimal symmetries for each other.
\end{itemize}
Here we denote by  $\pi_{+}$ the projection operator onto  ${\cal U}$ parallel to ${\cal G}[[\lambda]]$.  
\end{proposition}
\begin{exercise}
Prove the proposition (see proof of Theorem  \ref{Hier}).
\end{exercise}

 A general theory of Lax pairs of such kind and of the corresponding  Hamiltonian structures for any semi-simple Lie algebra ${\cal G}$ was presented
in \cite{sokolgok1}.  Below we consider the case ${\cal G}=\mathfrak{so}_3$.

\subsection{$\mathfrak{so}_3$ classical spinning tops} 

In this section we demonstrate  \cite[Section 4]{sokolgok1} that the factoring subalgebras in $\mathfrak{so}_3$ described in Subsection 2.3.1 are in one-to-one correspondence with classical integrable cases for the Kirchhoff problem of the motion of a rigid body in an ideal fluid \cite{bormam}. The equations of motion are given by 
\begin{equation}\label{Kirgh}
\frac{d\, {\bf \Gamma}}{d t}= {\bf \Gamma}\times \frac{\partial H}{\partial {\bf M}}, \qquad \quad \frac{d\, {\bf M}}{d t}= {\bf M}\times \frac{\partial H}{\partial {\bf M}}+{\bf \Gamma}\times \frac{\partial H}{\partial {\bf \Gamma}},
\end{equation}
where ${\bf M}=(M_1,M_2,M_3)$ is the total angular momentum, ${\bf \Gamma}=(\gamma_1,\gamma_2,\gamma_3)$ is the gravitational vector, $\times$ stands for the cross product,
$$
\frac{\partial H}{\partial {\bf M}}=
\Big(\frac{\partial H}{\partial M_1}, 
\frac{\partial H}{\partial M_2},
\frac{\partial H}{\partial M_3}\Big), \qquad \frac{\partial H}{\partial {\bf \Gamma}}=
\Big(\frac{\partial H}{\partial \gamma_1}, 
\frac{\partial H}{\partial \gamma_2},
\frac{\partial H}{\partial \gamma_3}\Big) ,
$$
and the quadratic form $H(\bf M, \bf \Gamma)$ is a Hamiltonian.  
\subsubsection{Structure of the orthogonal complement to $\cal U$}
\begin{proposition}\label{prop28} The orthogonal complement to $\cal U$ can be described as follows:
\begin{itemize}
\item[ i)] There exist unique elements ${\bf R}_i\in {\cal U}^{\perp}$ of the form 
\begin{equation}\label{RRR}
{\bf R}_{i}=\frac{{\bf e}_{i}}{\lambda}+O(1), \qquad i=1,2,3.
\end{equation}
They generate ${\cal U}^{\perp}$ as a ${\cal U}$-module. 
\item[ ii)] The following commutator relations hold:
$$
\begin{array}{l}
\left(%
\begin{array}{c}
  \left[ {\bf E}_{1}, {\bf R}_{1}\right] \\[2pt]
  \left[ {\bf E}_{3}, {\bf R}_{3}\right]\\[2pt]
  \left[ {\bf E}_{2}, {\bf R}_{2}\right]
\end{array}%
\right)={\bf A}\, \left(%
\begin{array}{c}
  {\bf R}_2 \\
  {\bf R}_{3}\\
   {\bf R}_{1}
\end{array}%
\right),\qquad
\left(%
\begin{array}{c}
  \left[{\bf E}_3,\,{\bf R}_1\right]+
  \left[{\bf E}_1,\,{\bf R}_3\right] \\
  \left[{\bf E}_1,\,{\bf R}_2\right]+
  \left[{\bf E}_2,\,{\bf R}_1\right]\\
  \left[{\bf E}_2,\,{\bf R}_3\right]+
  \left[{\bf E}_3,\,{\bf R}_2\right]
\end{array}%
\right)={\bf C}\,\left(%
\begin{array}{c}
  {\bf R}_2 \\
  {\bf R}_{3}\\
   {\bf R}_{1}
\end{array}%
\right),
\end{array}
$$
where ${\bf A}$ and ${\bf C}$ are matrices defined by  \eqref{VSmaincond1},  \eqref{VSmatcond1}.
\end{itemize}
\end{proposition}

\begin{remark}\label{rem23} Elements ${\bf R}_i\in \mathfrak{so}_3((\lambda))$ of the form \eqref{RRR} are determined by the commutator relations up to a summand of the form 
$S(\lambda)\, {\bf e}_i,$ where $S$ is a scalar Taylor series. 
\end{remark}

\begin{op} Prove that for any $S(\lambda)$ the ${\cal U}$-module generated by ${\bf R}_i, \, i=1,2,3$ is the orthogonal complement to ${\cal U}$ with respect to the form 
$$
\langle X(\lambda),\, Y(\lambda)\rangle_P = {\rm res} \,P(\lambda)
\Big(X(\lambda),\,Y(\lambda)\Big), \qquad X(\lambda),Y(\lambda)
\in \mathfrak{so}_3((\lambda))
$$
with a proper scalar Taylor series $P$.
\end{op}

The simplest option $\,L\in {\cal O}_1 ,\quad A=\pi_{+}(L)\,$ corresponds to integrable models of Euler type. In this case we have  
\begin{equation}\label{eule}
L=M_{1}{\bf R_{1}}+M_{2}{\bf R_{2}}+M_{3}{\bf R_{3}}, \qquad 
 A=M_{1}{\bf E_{1}}+M_{2}{\bf E_{2}}+M_{3}{\bf E_{3}},
\end{equation}

The Lax pairs for integrable Kirchhoff type systems have the following form:
\begin{equation}\label{LKir}
\begin{array}{c}
L=\gamma_1 [{\bf R}_3,\,{\bf E}_2 ]+
\gamma_2 [{\bf R}_1,\,{\bf E}_3 ]+
\gamma_3 [{\bf R}_2,\,{\bf E}_1 ]+m_1 {\bf R}_1+m_2 {\bf R}_2+m_3 {\bf R}_3, \\[4mm]
A=\pi_{+}\Big(\lambda L\Big)=\gamma_1 {\bf E}_1+\gamma_2 {\bf E}_2+\gamma_3 {\bf E}_3,
\end{array}
\end{equation}
where $m_i=M_i+c_i \gamma_i$ for some constants $c_i.$
It follows from Proposition \ref{prop28} that $L$ is a generic element of ${\cal O}_2.$ A unique non-trivial higher symmetry for the corresponding  ODE system corresponds to $A=\pi_{+}(L).$

\subsubsection{Clebsch integrable
case}

The factoring subalgebra ${\cal U}$ from Example \ref{example25} generated by elements \eqref{EE1} is isotropic and therefore ${\bf R}_i={\bf E}_i, \,\, i=1,2,3.$ 
The Lax equation  \eqref{Lax}, ({\ref{LKir}), where and $m_i=M_i,$ is equivalent to  \eqref{Kirgh}, where
$$
H=-\frac{1}{2}\Big(M_1^2+M_2^2+M_3^2-(q+r) \gamma_1^2-(p+r) \gamma_2^2-(p+q) \gamma_3^2\Big).
$$
This coincides with the Clebsch integrable
case in the Kirchhoff problem of the motion of a rigid body in an ideal fluid. 

Since the subalgebra ${\cal U}$ is isotropic, the Lax pair \eqref{eule} gives nothing.

\subsubsection{Euler and Steklov--Lyapunov cases}

For the factoring subalgebra from Example {\ref{example26}} generated by \eqref{EEE1} we have 
$$
{\bf R}_{1}={\bf e}_{1}\, \frac{1}{\sqrt{(1-r
\lambda)(1-q\lambda)}\, \lambda},
\qquad
{\bf R}_{2}={\bf e}_{2}\,
\frac{1}{\sqrt{(1-r\lambda)(1-p\lambda)}\,\lambda},
$$
$$
{\bf R}_{3}={\bf e}_{3}\,
\frac{1}{\sqrt{(1-q\lambda)(1-p\lambda)}\, \lambda}.
$$
The Lax pair \eqref{eule} yields the Euler equation (see also Example \ref{manak})
 $$
{\bf M}_{t}={\bf M}\times {\bf V} {\bf M},
$$
where
${\bf V}={\rm diag}\,(p,q,r).$ This Lax pair differs from the one considered in Example \ref{manak}.

The Lax equation  \eqref{Lax}, ({\ref{LKir}) is equivalent to  \eqref{Kirgh}, where
$$
H=-\frac{1}{2}\Big(M_1^2+M_2^2+M_3^2+(r+q) M_1 \gamma_1 +(r+p) M_2 \gamma_2+(q+p) M_3 \gamma_3\Big)-
$$
$$
\frac{1}{8}\Big((r-q)^2 \gamma_1^2 +(p-r)^2 \gamma_2^2 +(q-p)^2 \gamma_3^2\Big),
$$
and
$$
m_1=M_1+\frac{r-q}{2}\, \gamma_1, \qquad 
m_2=M_2+\frac{p-r}{2}\, \gamma_2,\qquad 
m_3=M_3+\frac{q-p}{2}\, \gamma_3.
$$
This is just the integrable Steklov--Lyapunov case. 

\subsubsection{Kirchhoff integrable case}
For the factoring subalgebra described in Example \ref{example27}  the elements
$$
{\bf R}_{i}=\frac{{\bf e}_{i}}{\lambda}+\nu\, [{\bf V},{\bf e}_{i}]-
\frac{1}{2}[{\bf V},\,[{\bf V},{\bf e}_{i}]]+({\bf V},\, {\bf V})\, {\bf e}_{i}, \qquad i=1,2,3
$$
satisfy the commutator relations from Proposition \ref{prop28} and therefore generate a ${\cal U}$-module that does not contain non-zero Taylor series. Any such module can be used to construct Lax pairs.
\begin{remark} This module is not ${\cal U}^{\perp}$ (see Remark \ref{rem23}). 
\end{remark}

\begin{exercise} Find the system of ODEs that corresponds to Lax pair \eqref{eule}.
\end{exercise}

\begin{exercise} Check that the Lax pair \eqref{LKir}  gives rise to the Kirchhoff integrable case (see \cite{sokolgok1}).
\end{exercise}

\begin{op}
Find the elements ${\bf R}_{i}$ for ${\cal U}^{\perp}.$
\end{op}

\subsection{Generalization of Euler and Steklov--Lyapunov cases to the  $\mathfrak{so}_n$-case}

The factorizing subalgebra from Example \ref{example26} can be described by the formula
\begin{equation}\label{Uplus}
{\cal U}=({\bf 1}+\lambda {\bf V})^{1/2}\,{\cal U}^{st}\, ({\bf 1}+\lambda {\bf V})^{1/2},
\end{equation}
where  ${\bf V}={\rm diag}\,(p,q,r),$
$$
({\bf 1}+\lambda {\bf V})^{1/2} = {\bf 1}+\frac{1}{2} {\bf V}\, \lambda - \frac{1}{8} {\bf V}^2\, \lambda^2 + \cdots ,
$$
and ${\cal U}^{st}$ is defined by \eqref{polyn}. 
According to Lemma \ref{lem255} the formula  \eqref{Uplus}, where ${\bf V}$ is arbitrary diagonal matrix,
defines a factorizing subalgebra for
${\cal G}=\mathfrak{so}_n$ as well.
The orthogonal complement to ${\cal U}$ is given by
$$
{\cal U}^{\perp}=(1+\lambda {\bf V})^{-1/2}\,{\cal U}^{st}\, (1+\lambda {\bf V})^{-1/2}.
$$
The simplest possibility $\,L\in {\cal O}_1 ,\quad A=\pi_{+}(L)\,$ corresponds to 
 $$
L=(1+\lambda {\bf V})^{-1/2}\,\frac{{\bf M}}{\lambda}\, (1+\lambda {\bf V})^{-1/2},
\qquad A=(1+\lambda {\bf V})^{1/2}\,\frac{{\bf M}}{\lambda}\, (1+\lambda {\bf V})^{1/2},
$$
where $ {\bf M}\in \mathfrak{so}_n$.
This Lax pair produces the Euler equation on $\mathfrak{so}_n$:
$$
{\bf M}_{t}=[{\bf V}, \, {\bf M}^{2}]. 
$$
 
The system of equations
$$
{\bf M}_{t}=[{\bf V}, \, {\bf M}^{2}]+[{\bf M},\,{\bf \Gamma}],
\qquad {\bf \Gamma}_{t}={\bf V} {\bf M}  {\bf \Gamma} -{\bf \Gamma} {\bf M} {\bf V},
\qquad {\bf M},{\bf \Gamma}\in \mathfrak{so}_n
$$
 possesses the Lax pair
$$
L=(1+\lambda {\bf V})^{-1/2}\,\left(\frac{{\bf M}}{\lambda^{2}}+
\frac{{\bf \Gamma}}{\lambda}\right)\, (1+\lambda {\bf V})^{-1/2},
\qquad A=(1+\lambda {\bf V})^{1/2}\,\frac{{\bf M}}{\lambda}\, (1+\lambda {\bf V})^{1/2}
$$
corresponding to the orbit ${\cal O}_2$.
One can regard this equation as an  $\mathfrak{so}_n$-generalization of the Steklov--Lyapunov top \cite{sokolgok1}.

\subsection{Factoring subalgebras for Kac--Moody algebras}

The Clebsch, Steklov--Lyapunov and Kirchhoff \footnote{For the Kirchhoff case there exists also an integral of first degree.} cases possess
additional integrals of second degree. To get trickier examples
like Kowalevsky top, one can consider a decomposition problem
for Kac-Moody algebras
(see \cite{sokolgok1}).

Let ${\cal G}$ be a semi-simple Lie algebra and $\phi$ be an automorphism of ${\cal G}$ of a finite order $k$. Let
$$
{\cal G}_i = \{a \in {\cal G} \, \vert \quad \phi(a) = \varepsilon^i\,a\}\, \lambda^i,\qquad i\in \Z,
$$
where $\varepsilon$ is a primitive root of $1$ of degree $k$. In this case, the Lie algebra 
$${\cal G}((\lambda, \phi)) = \Big\{ \sum_{i=-n}^{\infty} g_i \quad \vert \quad g_i\in {\cal G}_i, \quad n\in \Z\Big\}$$ 
is $\Z$-graded. It is called an (extended) {\it twisted loop algebra} or a {\it Kac-Moody algebra}.  

Several interesting integrable systems are related to the following Kac-Moody algebra.  
Let
$${\cal G}=\{{\bf A}\in {\rm Mat}_{n+m}\, \vert\, {\bf A}^{t}=-{\bf S} {\bf A} {\bf S} \},$$
where
$$
{\bf S} = \left(\begin{array}{cc}
  {\bf 1}_{n} & 0 \\
  0 & - {\bf 1}_{m} \\
\end{array}\right).
$$
It is clear that the Lie algebra ${\cal G}$ is isomorphic over $\C$ to  $\mathfrak{so}_{n+m}.$

Consider the subalgebra ${\cal A}$ of the loop algebra over ${\cal G}$ consisting of Laurent series such
that the coefficients of even (respectively, odd) powers of $\lambda$  belong to ${\cal G}_{1}$
(respectively, ${\cal G}_{-1}$). Here by ${\cal G}_{\pm 1}$ we denote the eigenspaces of the inner second
order automorphism  $\phi: \, {\cal G}\rightarrow {\bf S}{\cal G} {\bf S}^{-1},$ corresponding to eigenvalues  $\pm 1$.
Actually, this means that the coefficients of even powers of
$\lambda$ have the following block structure
$$
\left(\begin{array}{cc}
  v_{1} & 0 \\
  0 & v_{2} \\
\end{array}\right),
$$ where $v_1\in \mathfrak{so}_n, \, v_2\in \mathfrak{so}_m$, and the coefficients of odd powers
are of the form
$$ \left(\begin{array}{cc}
  0 & w \\
  w^t & 0 \\
\end{array}\right),
$$
where $w\in {\rm Mat}_{n,m}$.

We choose $\hbox{res} (\lambda^{-1}\hbox {tr}(X\,Y))$ for the non-degenerate invariant form
on ${\cal A}$. Note that in this case the form $\hbox{res} (\hbox {tr}(X\,Y))$ is degenerate.

Let ${\cal T}$ be the set of all Taylor series from ${\cal A}$,
$$
{\cal U} =({\bf 1}+\lambda r)^{1/2}\,{\cal U}^{st} \,({\bf 1}+\lambda r)^{1/2},
$$
where ${\cal U}^{st}$ is the set of polynomials in
$\lambda^{-1}$ from ${\cal A}$ and $r$ is arbitrary constant matrix of the form
$$
r=\left(\begin{array}{cc}
  0 & r_1 \\
  -r_1^t & 0 \\
\end{array}\right), \qquad r_1\in  {\rm Mat}_{n,m}.
$$
According to Lemma \ref{lem255}, ${\cal U}$ is a factoring subalgebra and the sum
$
{\cal A}={\cal T} + {\cal U}
$
is direct. The subalgebra $\cal U$ is a natural generalization of  \eqref{Uplus} to the
case when the structure of coefficients of series from ${\cal U}$ are
defined by an additional automorphism of second order.
The orthogonal complement to ${\cal U}$ with respect of the form $$\langle X,\, Y \rangle={\rm res}\,\lambda^{-1} {\rm tr}\,(X Y)$$ is given by 
$$
{\cal U}^{\perp}=({\bf 1}+\lambda r)^{-1/2}{\cal U}^{st} ({\bf 1}+\lambda r)^{-1/2}.
$$

The Lax equation $L_t=[\pi_{+}(L),\, L]$ corresponding to
$$
 L= ({\bf 1}+\lambda
r)^{-1/2}(\lambda^{-1} w+v+\lambda u)({\bf 1}+\lambda r)^{-1/2}
$$
is equivalent to the following system of equations
\begin{equation}\label{uvw}
w_{t}=[w,\, wr+rw-v], \qquad v_{t}=[u,\,w]+vwr-rwv, \qquad
 u_{t}=uwr-rwu.
\end{equation}
It is easy to see that this system admits the reduction
$$
u=\left(\begin{array}{cc}
  0 & r_1 \\
  r_1^t & 0 \\
\end{array}\right),
$$
which leads to the model found in \cite{soktsig}.  In the case
$n=3, m=2$ under further reductions we arrive at the Lax representation for the  
integrable case in the Kirchhoff problem  \cite{sokol2} with the Hamiltonian
$$
H=\frac{1}{2} \vert {\bf u}\vert ^{2} \vert {\bf M}\vert
^{2}+\frac{1}{2} \Big({\bf u},\,{\bf M}\Big)^{2} + \Big({\bf
u}\times {\bf v} ,\, {\bf M}\times {\bf \Gamma}\Big),
$$
where ${\bf u}$ and ${\bf v}$ are arbitrary constant vectors such
that $({\bf u},{\bf v})=0.$ The additional integral of motion in this case is of degree four. 

\subsection{Integrable PDEs of the  Landau-Lifshitz type} 

\subsubsection{Landau-Lifshitz equations related to $\mathfrak{so}_3$}

Any factoring subalgebra $\cal U$ for $\mathfrak{so}_3$ yields the following Lax pair
$$
L=\frac{d}{dx}+U, \qquad U=\sum_{i=1}^3 s_i\, {\bf E}_i, \qquad \qquad s_1^2+s_2^2+s_3^2=1,
$$
$$
A=\sum s_i\, [{\bf E}_j,\,{\bf E}_k]+\sum t_i\, {\bf E}_i
$$
for an integrable PDE of the Landau-Lifshitz type.
In this case the Lax equation has the form
\begin{equation}\label{AU}
U_{t}-A_x+[U,\, A]=0.
\end{equation}
The Laurent expansion of the left hand side of \eqref{AU} contains terms with $\lambda^k,$ 
where $k\ge -2.$ If coefficients at $\lambda^{-2}$ and at $\lambda^{-1}$ vanish, then the left hand side of  \eqref{AU} identically equals zero. Indeed, the subalgebra
$\cal U$ does not contain any non-zero Taylor series. 

To find the corresponding non-linear system of the form $\,\, {\bf s}_{t}=\vec F({\bf s},{\bf s}_x,{\bf
s}_{xx}),$ where ${\bf s}=(s_1,s_2,s_3),$ one can use the following straightforward computation. 
 Comparing the coefficients of $\lambda^{-2}$, we express $t_i$ in terms of ${\bf s}, {\bf s}_{x}$. Equating the
coefficients of $\lambda^{-1},$ we get an evolution system for ${\bf s}$. 
  Using the symmetry approach to integrability \cite{MikShaSok91} these systems were found in \cite{mikshab}. 

A standard way of finding all $A$-operators of the hierarchy defined by a given $L$-operator of the Landau-Lifshitz type based on a diagonalization procedure (cf. with  Theorem \ref{DiaP}) was proposed in \cite{golsok44}. Here we don't discuss it.  

Consider the case of Example \ref{example25}. Equating the coefficients of $\lambda^{-2}$ in   \eqref{AU} to zero, we get ${\bf s}_x={\bf s}\times {\bf t},$ where ${\bf t}=(t_1,t_2,t_3).$ Since ${\bf s}^2=1$ we find ${\bf t}={\bf s}_x\times {\bf s}+\mu\, {\bf s}.$ Comparing the coefficients of $\lambda^{-1},$ we get 
${\bf s}_t={\bf t}_x-{\bf s}\times {\bf V} {\bf s}$ or
$$
{\bf s}_t={\bf s}_{xx}\times {\bf s}+\mu_x\, {\bf s}+\mu\, {\bf s}_x-{\bf s}\times {\bf V} {\bf s}, $$
where ${\bf V}={\rm diag}\,(p,q,r).$ Since the scalar product $({\bf s},\,{\bf s}_t )$
has to be zero, we find that $\mu={\rm const}.$ The resulting equation coincides with  \eqref{ll} up to the involution $t \to -t$, the additional term of the form ${\rm const}\,  {\bf s}_x $ and a change of notation.  

The factoring subalgebra from Example \ref{example26} yields the equation
$$
{\bf s}_{t}={\bf s}\times {\bf s}_{xx}+({\bf s},\, {\bf V}\, {\bf s})\,
{\bf s}_{x}+ 2 {\bf s}\times ({\bf s}\times {\bf V}\, {\bf s}_{x}).
$$

The subalgebra from Example \ref{example27} corresponds to equation
$$
{\bf s}_{t}={\bf s}\times {\bf s}_{xx}+({\bf s},\, {\bf Z}\, {\bf s})\,
{\bf s}_{x}+ 2 {\bf s}\times ({\bf s}\times {\bf Z}\, {\bf s}_{x})+c\,
{\bf s} \times {\bf Z}\, {\bf s},
$$
where
$$ {\bf Z}= \left(
\begin{array}{ccc}
 r_{1}^{2}&r_{1}r_{2}&r_{1}r_{3} \\
r_{1}r_{2}&r_{2}^{2}&r_{2}r_{3}
 \\ r_{1}r_{3}&r_{2}r_{3}&r_{3}^{2}
\end{array}%
\right), \qquad \qquad
 c=\nu^{2}+\frac{r_{1}^{2}+r_{2}^{2}+r_{3}^{2}}{4}.
$$
In this equation ${\bf Z}$ is an arbitrary symmetric matrix of rank one and $c$ is an arbitrary constant.

\subsubsection{Perelomov model and vector Landau-Lifshitz equation}

Consider a special case $n=N, m=1$ of the Kac-Moody algebra from Subsection 2.3.5. Let us take  
$$
{\cal U} =\Big\{\, \sum_{i=-n}^{0} \lambda^{2 i}
\left(\begin{array}{cc}
  \Lambda {\bf A}_{i} \Lambda & \Lambda {\bf u}_i \\
  {\bf u}_i^t \Lambda & 0 \\
\end{array}\right), \quad n\in \N\,
\Big\}
$$
for the factoring subalgebra.
Here
$$
\Lambda=\frac{1}{\lambda}\sqrt{{\bf 1}-\lambda^2 {\bf R}}=\frac{{\bf 1}}{\lambda}
-\frac{{\bf R}}{2}\lambda-\frac{{\bf R}^2}{8 }\lambda^3+\cdots \, ,
$$
${\bf R}={\rm diag}\,(r_1,\dots, r_N), $ ${\bf A}_i$ are skew-symmetric $N\times N$-matrices, and ${\bf u}_i$ are column vectors. The orthogonal complement to ${\cal U}$ with respect to the form ${\rm res}(\lambda^{-1} (X,\,Y))$ has the form
\begin{equation}\label{ortog}
{\cal U}^{\perp}=\Big\{ \sum_{i=-n}^{-1} \lambda^{2 i}
\left(\begin{array}{cc}
  \Lambda^{-1} {\bf A}_{i} \Lambda^{-1} & \Lambda^{-1} {\bf u}_i \\
  {\bf u}_i^t \Lambda^{-1} & 0 \\
\end{array}\right), \quad n\in \N
\Big\}.
\end{equation}
The simplest $L$-operator
$$
L= \frac{1}{\lambda^{2}}\left(\begin{array}{cc}
  \Lambda^{-1} {\bf V} \Lambda^{-1} & \Lambda^{-1} {\bf u} \\
  {\bf u}^t \Lambda^{-1} & 0 \\
\end{array}\right)
$$
corresponds to $n=1$ in \eqref{ortog}. The Lax equation \eqref{Lax} with 
$$
A=\pi_{+}(\lambda^{-2} L)=\frac{1}{\lambda^{2}}\left(\begin{array}{cc}
  0 & \Lambda {\bf u} \\
  {\bf u}^t \Lambda & 0 \\
\end{array}\right)+\left(\begin{array}{cc}
  \Lambda {\bf V} \Lambda & \Lambda {\bf R} {\bf u} \\
  {\bf u}^t {\bf R} \Lambda & 0 \\
\end{array}\right)
$$
gives rise to Perelomov's $\mathfrak{so}_N$ generalization  
$$
{\bf V}_{t}=[{\bf V}^{2},\,{\bf R}]+[ {\bf u} {\bf u}^{t},\,{\bf R}^{2}], \qquad {\bf u}_{t}+({\bf V} {\bf R}+{\bf R} {\bf V})\, {\bf u}=0
$$
of the Clebsch top system.

For the Landau-Lifshitz equation \eqref{lanlif} the $L$-operator is given by \eqref{LAlif} and 
$$
A=\frac{1}{\lambda^2} \left(\begin{array}{cc}
0\,&\Lambda {\bf u} \\ 
{\bf u}^T \Lambda\,&0  
\end{array} \right)+\left(\begin{array}{cc}
  \Lambda {\bf V} \Lambda & \Lambda {\bf y} \\
  {\bf y}^t \Lambda & 0 \\
\end{array}\right),
$$
where the entries of ${\bf V}$ are given by $v_{i,j}=u_i (u_j)_x -u_j (u_i)_x $ and $${\bf y}={\bf u}_{xx}+\Big(\frac{3}{2}\langle {\bf u}_x,{\bf u}_x  \rangle+\frac{1}{2}\langle {\bf u},{\bf R}{\bf u}\rangle\Big)\, {\bf u}.$$

\subsection{Hyperbolic models of chiral type}

A class of factoring subalgebras for ${\cal G}=\mathfrak{so}_4$ and their relations with integrable $\mathfrak{so}_4$ spinning tops were investigated in \cite{efim}.  These subalgebras also generate \cite{sokolgok2, efsok} integrable 
hyperbolic PDEs of the form
$$
{\bf u}_{\xi}= {\bf A} {\bf v} \times {\bf u}, \qquad {\bf
v}_{\eta}=\bar {\bf A}\, {\bf u} \times{\bf v} ,
$$
where
$$
\qquad {\bf A}=\hbox {diag}(a_{1},a_{2},a_{3}),
\qquad  \bar {\bf A}=\hbox {diag}(\bar a_{1},\bar
a_{2},\bar a_{3}),
$$
${\bf u}, {\bf v}$ are three-dimensional 
s,  and the constants $a_i, \bar a_j$ obey the following relations. 
$$
a_{1} \bar a_{1}
(a_{3}^{2}-a_{2}^{2})+a_{2} \bar a_{2}
(a_{1}^{2}-a_{3}^{2})+a_{3} \bar a_{3}
(a_{2}^{2}-a_{1}^{2})=0,
$$$$
a_{1} \bar a_{1} (\bar a_{3}^{2}-\bar
a_{2}^{2})+ a_{2} \bar a_{2} (\bar
a_{1}^{2}-\bar a_{3}^{2})+a_{3} \bar a_{3}
(\bar a_{2}^{2}-\bar a_{1}^{2})=0.
$$
The Cherednik model \cite{cher} corresponds to $a_i=\bar a_i, \quad i=1,2,3.$ For the Golubchik-Sokolov case \cite{sokolgok2} we have $a_i=\bar a_i^{-1}, \quad i=1,2,3.$

The case when $a_{1}=a_{2}$
and then $\bar a_{1}=\bar a_{2},$ 
$a_{3},\,\bar a_{3}$ are arbitrary, was found in \cite{efsok}.

\subsubsection{Resume}

The description of factoring subalgebras is a fundamental problem of the theory. Each factoring subalgebra generates several different integrable PDEs and ODEs.

\section{Factorization method and non-associative algebras} 

Let ${\cal A}$ be an $N$-dimensional algebra  with a multiplication operation  $\circ$ defined by the
structural constants $C^i_{jk}.$ We associate to $\cal A$
a top-like ODE-system of the form
\begin{equation}\label{utop}
u^i_t=\sum_{j,k}C^i_{jk}\, u^j u^k, \qquad 
  i,j,k = 1, . . . , N.  
\end{equation}
Let ${\bf e}_1,\dots,\bf{e}_N$ be a basis in ${\cal A}$  and 
$$
U=\sum^N_{i=1} u_i {\bf e}_i.
$$
The system \eqref{utop} can be written in a short form
\begin{equation}\label{algtop}
U_t = U\circ U. 
\end{equation}
The system \eqref{algtop} is called {\it the ${\cal A}$-top}.

\begin{definition}
Algebras with the identity 
$[X, Y, Z]=0$ are called {\it left-symmetric} \cite{vinb}.
\end{definition}
Hereinafter we use the notation \eqref{as} and \eqref{br}. 

\begin{definition} Algebras with  the identity
\begin{equation}\label{SS}
[V, X, Y \circ Z] - [V, X, Y] \circ Z - Y \circ [V, X, Z] = 0.
\end{equation}
are called {\it SS-algebras} \cite{ss, gss}. 
\end{definition}
\begin{remark}
It follows from \eqref{SS} that for any $SS$-algebra ${\cal A}$ the operator
$$
K_{YZ} = [L_Y , L_Z] - L_{Y\circ Z} + L_{Z\circ Y}
$$
is a derivation of ${\cal A}$ for any $Y,Z$. As usual, $L_X$ denotes the operator of left multiplication by $X$.
\end{remark}

\begin{definition} An algebra with identities
\begin{equation}\label{t1}
[X, Y, Z] + [Y, Z, X] + [Z, X, Y] = 0, 
\end{equation}
and
\begin{equation}\label{t2}
V\circ [X, Y, Z] = [V\circ X, Y,Z] + [X, V\circ Y,Z] + [X, Y, V \circ Z]
\end{equation}
is called $G$-{\it algebra} \cite{golsokNon}. 
\end{definition}
\begin{remark}
Identity \eqref{t1} means that the operation $X \circ Y - Y \circ X$ is a Lie bracket.
\end{remark}

\subsection{Factorization method}

The factorization method (or, the same, AKS-scheme \cite{AKS}) is a finite--dimensional analog of the Riemann-Hilbert problem \cite{ZakShab79}, which can be used as a basis for the inverse scattering method. 

Similarly to Section 2.3, we deal with a vector space decomposition of a Lie algebra into a direct sum of its subalgebras.

Let ${\cal G}$ be a finite-dimensional Lie algebra, ${\cal G}_{+}$ and ${\cal G}_{-}$ be subalgebras in  ${\cal G}$ such that
\begin{equation}\label{decomp}
{\cal G}={\cal G}_{+} \oplus {\cal G}_{-}.
\end{equation}
The simplest example is the Gauss  decomposition of the matrix algebra into the sum of upper and law triangular matrices.  

The standard factorization method is used to integrate the following very special systems of the form \eqref{utop}:
\begin{equation}\label{decsys}
X_t = [\pi_{+}(X) ,\, X], \qquad X(0) = x_0. 
\end{equation}
Here $X(t) \in {\cal G}$, $\pi_{+}$ is the projector onto ${\cal G}_{+}$ parallel to ${\cal G}_{-}$.   Very often we denote by $X_{+}$ and $X_{-}$ the projections of $X$ onto ${\cal G}_{+}$ and ${\cal G}_{-}$, 
respectively. For
simplicity we assume that ${\cal G}$ is embedded into a matrix algebra. 
\begin{remark}
It follows from Lemma \ref{tr} that for any $k$ the function ${\rm tr}\, X^k$ is an integral of motion for \eqref{decsys}.
\end{remark}

\begin{proposition}
The solution of  Cauchy problem \eqref{decsys} is given by the formula
\begin{equation}\label{ans}
X(t) = A(t)\,x_0\,A^{-1}(t), 
\end{equation}
where function $A(t)$ is defined as a solution of the following factorization problem
\begin{equation}\label{fac}
 A^{-1}\,B = {\rm exp}\,(- x_0 \,t), \qquad A \in G_{+}, \quad B \in G_{-}, 
\end{equation}
where $G_{+}$ and $G_{−}$ are the Lie groups of  ${\cal G}_{+}$ and ${\cal G}_{-}$, respectively. 

\end{proposition}
\begin{proof}
Differentiating \eqref{ans},  we obtain
$$
X_t=A_t x_0 A^{-1}-A x_0 A^{-1} A_t A^{-1} = [A_t A^{-1},\, X].
$$
It follows from \eqref{fac} that
$$
-A^{-1} A_t A^{-1} B + A^{-1} B_t = -x_0 A^{-1} B. 
$$
The latter relation is equivalent to
$$
-A_t A^{-1} + B_t B^{-1}=- A x_0 A^{-1}. 
$$
Projecting it onto ${\cal G}_{+}$, we get $A_t A^{-1}=X_{+}$ which proves \eqref{decsys}.

\end{proof}

If the groups  $G_{+}$ and $G_{−}$  are
algebraic, then the conditions
$$
A \in G_{+},\qquad  A\, {\rm exp}(-x_0\,t) \in G_{-}
$$
are equivalent to a system of algebraic equations from which (for small $t$) the matrix
$A(t)$ is uniquely determined. 

The factorization problem \eqref{fac} can also be reduced to a system of linear differential equations with variable coefficients for $A(t)$. Define a linear operator $L(t): \,{\cal G}_{+} \to {\cal G}_{+}$ by the formula 
$$L(t)(v)=\Big({\rm exp}(x_0\,t)\, v\, {\rm exp}(-x_0\,t)  \Big)_{+}.
$$
Since $L(0)$ is the identity operator, $L(t)$ is invertible for small $t$. 

\begin{proposition}\label{prlin}
Let $A(t)$ be the solution of the initial problem
$$
A_t= A\, L(t)^{-1}\Big( (x_0)_{+}\Big), \qquad A(0)=I.
$$
Define $B$ by the formula $B=A\, {\rm exp}(-x_{0}\,t)$. Then the pair $(A,\,B)$ is the solution of the factorization problem \eqref{fac}. 
\end{proposition}
\begin{proof} Since $A^{-1}\,A_t \in {\cal G}_{+}$ and $A(0)=I$, we have $A\in G_{+}.$ It suffices to verify that $B^{-1} B_t \in {\cal G}_{-}$. We have 
$$
B^{-1} B_t= {\rm exp}(x_{0}\,t) \,A^{-1}\,\Big(A_t  {\rm exp}(-x_{0}\,t) -A x_{0}\,  {\rm exp}(-x_{0}\,t)  \Big)=
$$
$$
 {\rm exp}(x_{0}\,t)\,\Big(L(t)^{-1} (x_0)_{+}   \Big)\,{\rm exp}(-x_{0}\,t) - x_0.
$$
Projecting this identity onto ${\cal G}_{+}$ and using the definition of the operator $L(t)$, we obtain $(B^{-1} B_t)_{+}=0.$
\end{proof}

\subsection{Reductions}
It follows from \eqref{ans} that if the initial data $x_0$ for \eqref{decsys} belongs to a ${\cal G}_{+}$-module ${\cal M}$,
then $X(t)\in {\cal M}$ for any $t$. Such a specialization of \eqref{decsys} can be written as
\begin{equation}\label{mm}
M_t=[\pi_{+}(M),\,M], \qquad  M\in {\cal M}. 
\end{equation} 
Introducing the product 
\begin{equation}\label{mprod}
M_1\circ M_2=[\pi_{+}(M_1),\,M_2], \qquad  M_i\in {\cal M}, 
\end{equation} 
we equip ${\cal M}$ with a structure of algebra. The system \eqref{mm} is called ${\cal M}$-{\it reduction} and the operation \eqref{mprod} is called ${\cal M}$-{\it product}. 

Some classes of modules ${\cal M}$ correspond to interesting non-associative  algebras defined by \eqref{mprod}.

\subsubsection{Reductions for $\Z_2$--graded Lie algebras}
Let \begin{equation}\label{Zgrad}
{\cal G}={\cal G}_0\oplus {\cal G}_1
\end{equation}
 be a $\Z_2$-graded Lie algebra:
$$
[{\cal G}_0, {\cal G}_0] \subset {\cal G}_0, \qquad  [{\cal G}_0, {\cal G}_1] \subset {\cal G}_1, \qquad  [{\cal G}_1, {\cal G}_1] \subset {\cal G}_0.
$$

Suppose that we have a decomposition \eqref{decomp}, where ${\cal G}_{+} = {\cal G}_0$.
Let us consider the ${\cal G}_1$-reduction. 

\begin{example}
Let ${\cal G}_0={\cal G}_{+}, \,\, {\cal G}_1={\cal M},\,\, {\cal G}_{-},$ and ${\cal G}$ be the sets of skew-symmetric, symmetric,
upper-triangular and all matrices, respectively. Then the formula \eqref{mprod} defines the structure of a $G$-algebra on the set of symmetric matrices.  
\end{example}

It is clear that
\begin{equation}\label{Gmin}
{\cal G}_{-} = \{m - R(m)\,\vert\, m \in {\cal G}_1\},
\end{equation}
where $R=\pi_{+}$ is the projection onto ${\cal G}_{+} = {\cal G}_0$ parallel to ${\cal G}_{-}$.
\begin{theorem} {\rm \cite{golsokNon}}
 The vector space \eqref{Gmin} is a Lie subalgebra in ${\cal G}$ iff $R:\, {\cal G}_1\to {\cal G}_0 $ satisfies the modified Yang-Baxter equation
$$
R\Big([R(X),\,Y] - [R(Y),\,X]\Big) - [R(X),\, R(Y)] - [X,\, Y] = 0, \qquad X, Y \in {\cal G}_1.
$$
\end{theorem}
\begin{remark}
It is important to note that in our case $R$ is an operator defined on ${\cal G}_1$ and acting from
${\cal G}_1$ to ${\cal G}_0$, whereas usually {\rm (see \cite{semen})} $R$ is assumed to be an operator on ${\cal G}$.
\end{remark}

\begin{proposition}\label{lefts} If $[{\cal G}_1,\, {\cal G}_1]=\{ 0 \},$ then ${\cal G}_1$ is a
left-symmetric algebra with respect to the product \eqref{mprod}.
\end{proposition}
\begin{proof} Let $X,Y,Z \in {\cal G}_1$. Let us verify that
\begin{equation}\label{iden}
[X,\, Y, \,Z] = [[X, Y ], Z],
\end{equation}
where the left hand side is defined by \eqref{br}. According to \eqref{mprod}, it has the form
$$
[[X_{+},Y]_{+},\, Z] - [[Y_{+},X]_{+},\, Z] + [Y_{+}, \,[X_{+}, Z]]-  [X_{+}, \,[Y_{+}, Z]].
$$ 
 We have
$$
[[X_{+},Y]_{+},\, Z] - [[Y_{+},X]_{+},\, Z]=[[X,\,Y]_{+},\, Z]-[[X_{-},\,Y]_{+},\, Z]+ [[X, Y_{+}]_{+},\, Z]=
$$
$$
[[X,\,Y],\, Z]+[[X_{+}, Y_{+}]_{+},\, Z]=[[X,\,Y],\, Z]+[[X_{+}, Y_{+}],\, Z].
$$
Now \eqref{iden} follows from the Jacobi identity for $X_{+}, Y_{+}, Z$.  Since $[{\cal G}_1,\, {\cal G}_1]=\{ 0 \}$ the proposition statement is a consequence of \eqref{iden}.
\end{proof}

In the general case we arrive at $G$-algebras.

\begin{theorem} \label{TG} \begin{itemize}
\item[i)]  The vector space ${\cal G}_1$ is a
$G$-algebra with respect to the operation \eqref{mprod}.
\item[ ii)] Any $G$-algebra can be obtained from a suitable $\Z_2$-graded Lie algebra by the above construction.
\end{itemize}
\end{theorem}
\begin{proof}
To prove identity \eqref{t1} it suffices to project the Jacobi identity for $X_{-}, Y_{-}, Z_{-}$ onto ${\cal G}_1.$ Rewriting \eqref{t2} in terms of the ${\cal G}$-bracket with the help of \eqref{iden}, we see that \eqref{t2} follows from
the Jacobi identity for ${\cal G}$.

It remains to prove the second part of the theorem. Let ${\cal G}_1$ be a $G$-algebra. Define ${\cal G}$ by formula \eqref{Zgrad}, where ${\cal} G_0$ is the Lie algebra generated by all operators of left multiplication of ${\cal G}_1$.
Recall
that the left multiplication operator $L_X$ is defined as follows: $L_X(Y)=X\circ Y$. The bracket on ${\cal G}$ is defined by
\begin{equation}\label{313}
[(A, X), (B, Y )] =\Big([A,B] - [L_X, L_Y ] + L_{X\circ Y} - L_{Y\circ X}, \,\, A(Y ) - B(X)\Big).
\end{equation}
The skew-symmetry is obvious. One can easily verify that the identities
\eqref{t1}, \eqref{t2} are equivalent to the Jacobi identity for \eqref{313}. It follows from \eqref{313} that the
decomposition \eqref{Zgrad} defines a $\Z_2$-gradation. To define a decomposition \eqref{decomp} we take for ${\cal G}_{-}$
the set $ \{(-L_X, \,X) \}$ and ${\cal G}_0$ for ${\cal G}_{+}$. Formula \eqref{313} implies that ${\cal G}_{-}$ is a subalgebra in ${\cal G}$.
For ${\cal G}_{\pm}$ thus defined, \eqref{mprod} has the form $(0,X)\circ (0,Y) = [(L_X,\,0),\, (0,\,Y)]$. This
relation is fulfilled according to \eqref{313}.
\end{proof}
The part i) of the theorem means that any $G$-top \eqref{utop} is integrable by the factorization method. 

\begin{example} Putting 
$$  {\cal G}_+ = \left\{\begin{pmatrix}
a\,&c\,&0 \\ 
d\,&b\,&0 \\ 
0\,&0\,&-a-b \end{pmatrix} \right\}, \qquad {\cal G}_1 =
\left\{ \begin{pmatrix}
0\,&0\,&P \\ 
0\,&0\,&Q \\ 
R\,&S\,&0 \end{pmatrix} \right\},
$$
we take $\mathfrak{sl}_3$ for $\cal G$.   Let us choose a complementary subalgebra ${\cal G}_-$ as follows: 
$$ {\cal G}_- = \left\{  \begin{pmatrix}
-Y+X+\alpha U\,&X\,&Y-X \\ 
-Z+(2-3\alpha )U\,&(1-2\alpha )U\,&Z+(3\alpha -2)U \\ 
-Y\,&X\,&Y-X+(\alpha -1)U \end{pmatrix} \right\},$$
where $\alpha$ is a  parameter.

The operation \eqref{mprod} turns the vector space ${\cal M}={\cal G}_1$ into  a 
$G$-algebra. The corresponding ${\cal M}$-top is the following system  of  
differential equations 
\begin{equation}\label{syst}\begin{cases} 
P_t = P^2 - R P - Q S, \\
Q_t = (\beta -2) R Q + \beta P Q, \\
R_t = R^2 - R P - Q S, \\
S_t = (3-\beta )R S + (1-\beta )P S, 
\end{cases}
\end{equation}
where $\beta = 3\alpha$, for the entries of the matrix 
$${\bf M} =   \begin{pmatrix}
0\,&0\,&P \\ 
0\,&0\,&Q \\ 
R\,&S\,&0 \end{pmatrix} .$$  
From \eqref{mm} it follows that $I_1 = {\rm tr} \,{\bf M}^2 = R P+Q S$ is a first integral 
for the system. Other integrals of the form ${\rm tr}\, {\bf M}^k$ are trivial. 
Nevertheless it is not hard to integrate \eqref{syst} by quadratures. 
The auxiliary two first integrals are of the form 
$$I_2 = \frac{P-R}{QS}, \qquad I_3= Q^{1-\beta }S^{-\beta }(R^2-RP-QS).$$
For generic $\beta$ the integral $I_3$ is a multi-valued function. It shows  that \eqref{syst} is not integrable 
from the view-point of the Painlev\'e approach (see for example, \cite{Newell}). 
\end{example}

\subsection{Generalized factorization method} 
Suppose that 
\begin{equation}\label{VV}
{\cal G}=V_1 \oplus V_2,
\end{equation} 
where the $V_i$ are vector spaces. Let 
\begin{equation}\label{calH}
{\cal H}= [V_1,\,V_1]_{-} +  [V_2,\,V_2]_{+} .
\end{equation}
Here $+$ and $-$ symbolize the projections onto $V_1$ and $V_2$, respectively. 
 Assume that $V_1$ and $V_2$ satisfy the following conditions 
\begin{equation}\label{HHH}
[{\cal H}, \, V_1] \subset V_1, \qquad [{\cal H}, \, V_2] \subset V_2.
\end{equation}
If $V_i$ are subalgebras, then ${\cal H}=\{0\}$ and conditions \eqref{HHH} are trivial.

It turns out \cite{golsok4} that equation \eqref{decsys}, where 
  $\pi_{+}$ is the projection onto $V_{1}$ parallel to $V_{2}$, can be reduced to solving a system of linear equations with variable coefficients (cf. Proposition \ref{prlin}).
  
\begin{remark}
 If conditions \eqref{HHH} hold, then 
 ${\cal H}$, ${\cal G}_{+}=V_1 + [V_1,\, V_1]_{-}$ and 
 ${\cal G}_{-}= V_2 + 
 [V_2,\, V_2]_{+}$ are Lie subalgebras in ${\cal G}.$ Moreover, ${\cal H}={\cal G}_{+}\bigcap {\cal G}_{-}$
\end{remark} 
  
\begin{theorem}\begin{itemize}
\item[ i)] Let ${\cal G} = {\cal G}_0 \oplus {\cal G}_1$ be a $\Z_2$-graded Lie algebra, such that $[{\cal G}_1, \,{\cal G}_1] = 0$. Given a
vector space decomposition \eqref{VV} with $V_1 = {\cal G}_0$ and a vector space $V_2$ satisfying conditions \eqref{HHH}, we equip $V_2$ with an algebraic structure by
formula \eqref{mprod}. Then $V_2$ is a $SS$-algebra with respect to the operation $\circ$.
\item[ ii)] Any $SS$-algebra ${\cal A}$ can be obtained from a suitable $\Z_2$-graded Lie algebra by the above construction
\end{itemize}
\end{theorem}
\begin{proof} The first part can be proved in the
same manner as the first part of Theorem \ref{TG}. We explain only how to construct ${\cal G}, {\cal G}_+, V_2.$ for a given $SS$-algebra. We take for  ${\cal G}_+$ the Lie algebra ${\rm End}\, {\cal A}$ of all linear endomorphisms
of ${\cal A}$ . The vector space
$${\cal G} = ({\rm End}\, {\cal A}) \oplus {\cal A}$$
becomes a $\Z_2$-graded Lie algebra if we define the bracket by
$$
[(A, X), (B, Y)] = \Big([A,\,B],\, A(Y) - B(X)\Big).
$$
It is not difficult to show that \eqref{SS} implies that
a) the vector space ${\cal H}$ generated by all elements of the form
$$
\Big([L_Y,\, L_Z] - L_{Y\circ Z} + L_{Z\circ Y},\, 0\Big)
$$
is a Lie subalgebra in ${\cal G}$, and
b) the vector spaces $V_2 = \{(-   L_X,\, X)\},$  $V_1={\cal G}_{+}$ and the subalgebra ${\cal H}$ satisfy conditions \eqref{calH} and \eqref{HHH}. 

\end{proof}

\begin{example}
 Let us take 
$$ {\cal G} = \left\{ \begin{pmatrix}
*\,&*\,&*\,&* \\ 
*\,&*\,&*\,&* \\ 
*\,&*\,&*\,&* \\
0\,&0\,&0\,&0 \end{pmatrix}   \right\}$$
for $\Bbb Z_2$-graded Lie algebra. It is clear that ${\cal G} = {\cal G}_{0} 
\oplus {\cal G}_1$, where 
$$ {\cal G}_0 = \left\{   \begin{pmatrix}
*\,&*\,&*\,&0 \\ 
*\,&*\,&*\,&0 \\ 
*\,&*\,&*\,&0 \\
0\,&0\,&0\,&0 \end{pmatrix}   \right\}, \qquad {\cal G}_{1} =
\left\{  \begin{pmatrix}
0\,&0\,&0\,&P \\ 
0\,&0\,&0\,&Q \\ 
0\,&0\,&0\,&R \\
0\,&0\,&0\,&0 
\end{pmatrix} \right\}. 
$$
Let ${\cal G}_{+}={\cal G}_0$
and
$$ {\cal G}_{-} = \left\{ \begin{pmatrix}
c\,&\lambda c\,&a\,&a \\ 
-\lambda c\,&c\,&b\,&b \\ 
a\,&b\,&c\,&c \\
0\,&0\,&0\,&0 \end{pmatrix} \right\}, 
$$
where $\lambda$ is a parameter. 
Since ${\cal G}_{-}$ is not a subalgebra, we have to find the vector space $\cal H$ using \eqref{calH}. A simple
calculation shows that
$$
\qquad {\cal H} =
\left\{ \begin{pmatrix}
0\,&-d\,&0\,&0 \\ 
d\,&0\,&0\,&0 \\ 
0\,&0\,&0\,&0 \\
0\,&0\,&0\,&0 
\end{pmatrix} \right\}
$$
and that the conditions \eqref{HHH} are fulfilled.  The corresponding $SS$-top (up to a 
scaling) is given by
$$\begin{cases} 
P_t = 2 P R + \lambda Q R, \\
Q_t = 2 Q R - \lambda P R, \\
R_t = P^2+Q^2+R^2.
\end{cases} $$ 
\end{example}

\chapter{Algebraic structures in bi-Hamiltonian approach}

\section{Polynomial forms for elliptic Calogero-Moser systems}

\subsection{Calogero-Moser Hamiltonians}

Consider quantum integrable Hamiltonians of the form
\begin{equation}\label{genHam}
H=-\Delta+U(x_1,...,x_n), \qquad {\rm where} \qquad  \Delta =\sum_{i=1}^n \frac{\partial^2}{\partial x_i^2}
\end{equation}
related to simple Lie algebras \cite{Olshanetsky:1983}.  For such Hamiltonians the potential $U$ is a rational, trigonometric or elliptic function. 
\begin{observation} \rm{(}A.Turbiner{\rm )}. For many of these Hamiltonians there exists a change of variables and a  gauge transformation that bring the Hamiltonian to a differential operator with polynomial coefficients. 
\end{observation}

The elliptic Calogero-Moser Hamiltonian is given by
\begin{equation}\label{CAL}
H_N=- \Delta+ \beta(\beta-1)\sum_{i\neq j}^{N+1}\wp(x_i-x_j). 
\end{equation}
Here  $\beta$ is a parameter, and  $\wp(x)$ is the Weierstrass $\wp$-function with the invariants $g_2,\ g_3$, i.e., a solution of the ODE  $\,\,\wp'(x)^2=4 \wp(x)^3-g_2 \wp(x)-g_3.\,$ In the coordinates 
$$
X=\frac{1}{N+1}\sum_{i=1}^{N+1} x_i,\qquad y_i=x_i-X
$$
the operator  \eqref{CAL} takes the form
$$
H_{N}=-\frac{1}{N+1}\frac{\partial^2}{\partial X^2}+{\cal H}_N\left(y_1,y_2,\dots y_{N}\right),
$$
where
\begin{equation}\label{CM}
{\cal H}_N=-\frac{N}{N+1}\sum_{i=1}^{N}\frac{\partial^2}{\partial y_i^2}+\frac{1}{N+1}\sum_{i\neq j}^{N}\frac{\partial^2}{\partial y_i\partial y_j}+\beta(\beta-1)\sum_{i\neq j}^{N+1}\wp(y_i-y_j).
\end{equation}
In the last term we have to substitute $- \sum_{i=1}^N y_i$ for $y_{N+1}$. 

In \cite{matsok} the following transformation 
  $(y_1,\dots ,y_N)\to (u_1,\dots ,u_N)$ defined by
 \begin{equation}\label{tran}
  \begin{pmatrix}\wp(y_1)&\wp'(y_1)&\dots&\wp^{(N-2)}(y_1)&\wp^{(N-1)}(y_1)\\\wp(y_2)&\wp'(y_2)&\dots&\wp^{(N-2)}(y_2)&\wp^{(N-1)}(y_2)\\ \vdots&\vdots&\dots&\vdots&\vdots\\\wp(y_{N})&\wp'(y_{N})&\dots&\wp^{(N-2)}(y_{N})&\wp^{(N-1)}(y_{N})\\\end{pmatrix}
\begin{pmatrix} u_1\\u_2\\\vdots\\u_{N}\\\end{pmatrix}=\begin{pmatrix} 1\\1\\\vdots\\1\\\end{pmatrix}
 \end{equation}
 was considered. Denote by 
 $D_N(y_1,\dots, y_N)$ the Jacobian of the transformation  \eqref{tran}. 
 
 \begin{conjecture}\label{Conjecture1} The gauge transform  ${\cal H}_N \rightarrow D_N^{-\frac{\beta}{2}}{\cal H}_N D_N^{\frac{\beta}{2}}$ and subsequent change of variables  \eqref{tran} bring   \eqref{CM} to a differential operator  $P_N$ with polynomial coefficients.
\end{conjecture}

In the case $N=2$ the transformation  \eqref{tran} coincides with the transformation 
$$
u_1=\frac{\wp'(y_2)-\wp'(y_1)}{\wp(y_1)\wp'(y_2)-\wp(y_2)\wp'(y_1)}, \qquad \, 
u_2=\frac{\wp(y_1)- \wp(y_2)}{\wp(y_1)\wp'(y_2)-\wp(y_2)\wp'(y_1)},
$$
found in \cite{sokturb}. In addition to explicit form of $P_2$, in this paper a polynomial form for the elliptic $G_2$-model was found. Polynomial forms for rational and trigonometric Calogero-Moser Hamiltonians in the case of arbitrary $N$ were described in  \cite{RT:1995}. 

\begin{remark}\label{rem31} Obviously, for any polynomial form $P$ of Hamiltonian  \eqref{genHam}  
\begin{itemize}
\item[ 1:] the contravariant metric g defining by the symbol of P is flat;
\item[ 2:] $P$ can be reduced to a self-adjoint operator by a gauge transformation $P\to f P f^{-1},$ where $f$ is a function. 
\end{itemize}
\end{remark}
  Besides evident properties 1,2  we have in mind the following  
\begin{observation} {\rm (}A. Turbiner{\rm )}. For all known integrable cases the polynomial form
 $P$ preserves some nontrivial finite-dimensional vector space of polynomials. 
\end{observation}

For Hamiltonians \eqref{CM} the situation can be described as follows.
Consider the differential operators 
 $e_{i,j} = E_{i-1,j-1}$, where
\begin{equation}
\begin{array}{c}\label{opera}
\displaystyle E_{ij}=y_i\frac{\partial}{\partial y_j},\qquad E_{0i}=\frac{\partial}{\partial y_i},\\[5mm] 
\displaystyle E_{00}=-\sum_{j=1}^{N}  y_j \frac{\partial}{\partial y_j}+\beta\,(N+1), \qquad E_{i0}=y_i E_{00}.
\end{array}
\end{equation}
It is easy to verify that they
satisfy the commutator relations  
\begin{equation}\label{univen}
e_{ij} e_{kl}-e_{kl} e_{ij}=\delta_{j,k}e_{il}-\delta_{i,l}e_{kj}, \qquad i,j=1,\dots,N+1,
\end{equation}
and, therefore, define representations of the Lie algebra $\mathfrak{gl}_{N+1}$ and of the universal enveloping algebra $U(\mathfrak{gl}_{N+1}).$ 
The latter representation is not exact.  
\begin{conjecture}\label{Conjecture2} The differential operator $P_N$ from Conjecture \ref{Conjecture1} can be written as a linear combination of anti-commutators
of the operators $E_{ij}$.
\end{conjecture}
The conjectures \ref{Conjecture1} and \ref{Conjecture2}  have been verified  in \cite{matsok} for $N=2,3$. Moreover, differential operators with polynomial 
coefficients that commute with $P_2$ and with $P_3$ were found. These operators can also be written as non-commutative polynomials in the $E_{ij}$. 

\begin{remark} If $k=-\beta(N+1)$ is a positive integer, then  
the operators  \eqref{opera} preserve the vector space of all polynomials in $y_1,\dots, y_{N},$ whose degrees are not greater than $k$ {\rm \cite{RT:1995}}.
\end{remark}

 \subsection{Quasi-solvable differential operators}
 
 \begin{definition} A linear differential operator 
\begin{equation}\label{gendifop} 
Q=\sum_{i_1+\cdots+i_N\le m} a_{i_1,...,i_N}  \p_{y_1}^{i_1} \cdots \p_{y_N}^{i_N} 
\end{equation}
of order $m$ with polynomial coefficients
is called {\it quasi-solvable} if it preserves the vector space of all polynomials in $y_1,...,y_N$ of degree $\le k$ for some $k\ge m.$

 \end{definition}
 
\begin{theorem} {\rm \cite{sok}} For any quasi-solvable differential operator \eqref{gendifop} $${\rm deg}\, (a_{i_1,...,i_N})\le m+i_1+\cdots+i_N.$$
\end{theorem}

\begin{op}
Prove that any quasi-solvable operator can be represented as a {\rm (}non-commutative{\rm )} polynomial in the variables \eqref{opera}, where $k=-\beta (N+1)$.  
\end{op}
\begin{remark}
Such a representation is not unique.
\end{remark} 
 
\subsubsection{ODE case}
Consider the case $N=1$. 
\begin{lemma} Any quasi-solvable operator $P$  of second order  has the following structure:
$$
P= (a_4 x^4+a_3 x^3+a_2 x^2+a_1 x +a_0)  \frac{ d^2}{dx ^2} +(b_3 x^3+b_2 x^2+b_1 x +b_0)  \frac{ d}{dx}+c_2 x^2+c_1 x+c_0, 
$$
where the coefficients are related by the following identities
$$
b_3=2 (1-k)\, a_4 ,\qquad c_2= k (k-1) \, a_4, \qquad c_1=k (a_3-k a_3-b_2). 
$$
\end{lemma}  
The transformation group ${\rm GL}_2$
acts on the nine-dimensional vector space of such operators as
\begin{equation}\label{flin}
x\to \frac{s_1 x+s_2}{s_3 x+s_4}, \qquad P\to  (s_3 x+s_4)^{-k} P (s_3 x+s_4)^k.
\end{equation}
The coefficient $a(x)$  of the second derivative is a fourth order 
polynomial, which transforms as follows
$$
a(x)\to (s_3 x+s_4)^4 a\Big(\frac{s_1 x+s_2}{s_3 x+s_4}\Big).
$$
If $a(x)$ has four distinct roots, we call the operator $P$ {\it elliptic}. In the elliptic case using transformations  \eqref{flin}, 
we may reduce $a$ to
$$
a(x)=4\, x(x-1)(x-\kappa),
$$ 
where $\kappa$ is the elliptic parameter.

Define parameters $n_1,...,n_5$ by the following identities:
$$
b_0=2 (1+2 n_1), \quad b_1=-4 \Big((\kappa+1)(n_1+1)+\kappa n_2+n_3 \Big),$$$$ b_2=-2\,(3+2 n_1+2 n_2+2 n_3),
$$
$$
k=- \frac{1}{2}(n_1+n_2+n_3+n_4), $$$$ n_5=c_0+ n_2(1-n_2)+\kappa n_3 (1-n_3)+(n_1+n_3)^2+\kappa (n_1+n_2)^2.
$$
Then the operator $H=h P h^{-1},$ where 
$$\displaystyle  h=x^{\frac{n_{1}}{2}}(x-1)^{\frac{n_{2}}{2}}(x-\kappa)^{\frac{n_{3}}{2}},
$$
has the form
$$
H=a(x) \frac{ d^2}{dx ^2}+\frac{a'(x)}{2} \frac{ d}{dx}+n_5+n_4 (1-n_4)\,x+
\frac{n_1 (1-n_1) \kappa}{x}+$$$$\frac{n_2 (1-n_2) (1-\kappa)}{x-1}+\frac{n_3 (1-n_3) \kappa (\kappa-1)}{x-\kappa}.
$$
Now after the transformation $y=f(x)$, where 
$$ f'^2=4 f (f-1) (f-\kappa)$$ 
we arrive at
$$
H= \frac{ d^2}{dy ^2}+n_5+n_4 (1-n_4)\,f+
\frac{n_1 (1-n_1) \kappa}{f}+\frac{n_2 (1-n_2) (1-\kappa)}{f-1}+\frac{n_3 (1-n_3) \kappa (\kappa-1)}{f-\kappa}.
$$ 
In general here  $n_i$ are arbitrary parameters.

Another form of this Hamiltonian (up to a constant) is given by
$$
H= \frac{ d^2}{dy ^2}+n_4 (1-n_4)\,\wp(y)+
 n_1 (1-n_1) \,\wp(y+\omega_1)+
 n_2 (1-n_2) \,\wp(y+\omega_2)+
 n_3 (1-n_3) \,\wp(y+\omega_1+\omega_2),
$$ 
where $\omega_i$ are half-periods of the Weierstrass function $\wp(x)$. If $n_1=n_2=n_3=0,$ we get the Lame operator. In general, it is the  Darboux-Treibich-Verdier operator \cite{veselov}.

When $$k=-\frac{1}{2}(n_1+n_2+n_3+n_4)$$  is a natural number, this operator $H$ preserves the finite-dimensional vector space of elliptic functions, which corresponds to polynomials for the initial operator $P.$ 

\subsubsection{Two-dimensional operators}

Consider second order differential operators of the form
\begin{equation}
\label{oper}
  P\ =\ a(x,y) \frac{\pa^2}{\pa x^2}+2 b(x,y) \frac{\pa^2}{\pa x \pa y}+c(x,y) \frac{\pa^2}{\pa y^2}
 + d(x,y) \frac{\pa}{\pa x}\ +
   e(x,y) \frac{\pa}{\pa y}+f(x,y)
\end{equation}
with polynomial coefficients. Denote by $D(x,y)$ the determinant $a(x,y) c(x,y)-b(x,y)^2.$ We assume that $D\ne 0.$

\begin{lemma}\label{lem32} The operator \eqref{oper} is quasi-solvable iff the coefficients have the following structure
$$
a=q_1 x^4+q_2 x^3 y+q_3 x^2 y^2+z_1 x^3+z_2 x^2 y+ z_3 x y^2+a_1 x^2+a_2 x y+ a_3 y^2+a_4 x+a_5 y+a_6;
$$
$$
b=q_1 x^3 y+q_2 x^2 y^2+q_3 x y^3+\frac{1}{2} \Big(z_4 x^3+(z_1+z_5) x^2 y+(z_2+z_6) x y^2+z_3 y^3  \Big)
$$$$
+b_1 x^2+ b_2 x y+ b_3 y^2+b_4 x+b_5 y+b_6;
$$
$$
c=q_1 x^2 y^2+q_2 x y^3+q_3 y^4+z_4 x^2 y+z_5 x y^2+ z_6 y^3+c_1 x^2+c_2 x y+ c_3 y^2+c_4 x+c_5 y+c_6;
$$ 
$$
d=(1-k) \Big(2 (q_1 x^3+q_2 x^2 y+q_3 x y^2)+z_7 x^2+ (z_2+z_8-z_6) x y+z_3 y^2\Big)+d_1 x+d_2 y+d_3;
$$
$$
e=(1-k) \Big(2 (q_1 x^2 y+q_2 x y^2+q_3 y^3)+z_4 x^2+(z_5+z_7-z_1) x y+
z_8 y^2\Big)+e_1 x+e_2 y+e_3;
$$
$$
f=k (k-1) \Big(q_1 x^2+q_2 x y+q_3 y^2+(z_7-z_1) x +(z_8-z_6) y\Big)+f_1.
$$
\end{lemma}
The dimension of the vector space of such operators equals 36. The group ${\rm GL}_3$ acts on this vector space in a projective way by transformations  
\begin{equation}\begin{array}{c}
\ds \tilde x=\frac{a_1 x+a_2 y+a_3}{c_1 x+c_2 y+c_3}, \qquad \quad \bar y=\frac{b_1 x+b_2 y+b_3}{c_1 x+c_2 y+c_3}, \\[6mm] \tilde P=(c_1 x+c_2 y+c_3)^{-k} P \circ (c_1 x+c_2 y+c_3)^{k}. 
\end{array}\label{GL3}
\end{equation}
This transformation corresponds to the matrix 
$$
  \begin{pmatrix}
a_1\,&a_2\,&a_3 \\ 
b_1\,&b_2\,&b_3 \\ 
c_1\,&c_2\,&c_3 \end{pmatrix} \in {\rm GL}_3.  
$$
The representation is a sum of irreducible representations $W_1$, $W_2$ and $W_3$ of dimensions 27, 8 and 1, correspondingly. A basis of $W_2$ is given by  
$$
x_1=5 z_7-z_5-7 z_1, \qquad x_2=5 z_8-z_2-7 z_6,  \qquad  x_3= 5 d_1+2  (k-1) (2 a_1 +b_2),  $$$$
x_4= 5 e_1+2 (k-1) (2 b_1 +c_2), \qquad 
x_5=5 d_2+2 (k-1) (2 b_3+a_2), \qquad x_6=5 e_2+2 (k-1) (2 c_3+b_2),
$$
$$
x_7=5 d_3+2 (k-1) (a_4+b_5), \qquad x_8=5 e_3+2 (k-1) (b_4+c_5).
$$
The generic orbit of the group action on $W_2$ has dimension 6. There are two polynomial invariants of the action:
$$
I_1=x_3^2-x_3 x_6+x_6^2+3 x_4 x_5+3 (k-1) (x_1 x_7+x_2 x_8),
$$
and
$$
I_2=2 x_3^3-3 x_3^2 x_6-3 x_3 x_6^2+2 x_6^3+9 x_4 x_5 (x_3+x_6)+
$$
$$
9 (k-1) (x_1 x_3 x_7+x_2 x_6 x_8-2 x_1 x_6 x_7-2 x_2 x_3 x_8+3 x_2 x_4 x_7+3 x_1 x_5 x_8).
$$

\subsubsection{Flat polynomial metrics}
According to Remark \ref{rem31}, the contravariant metric
\[
g^{1,1}=a\ ,\qquad  g^{1,2}=g^{2,1}=b\ ,\qquad  g^{2,2}=c\ ,
\]
defined by the coefficients $a,b,c$ of the operator \eqref{oper} is flat (i.e. $R_{1,2,1,2}=0$) for any polynomial form $P$. 
\begin{op} Describe all flat contravariant metrics defined by the polynomials  $a,b,c$ from Lemma \ref{lem32} up to transformations 
\eqref{GL3}.
\end{op} 
Some particular results were obtained in \cite{sok}.
\begin{example} 
For any constant $\kappa$ the metric $g$ with 
$$
a=(x^2-1) (x^2-\kappa) +(x^2+\kappa)\,  y^2,
\qquad b=x y\, (x^2+y^2+1-2 \kappa),
$$
\[
c=(\kappa-1)(x^2-1)+(x^2+2-\kappa)\,y^2+y^4
\]
is flat. Moreover, this is a linear pencil of polynomial contravariant flat metrics with respect to the parameter $\kappa$ \cite{dub}.  The metric is related to a polynomial form \cite{Takemura} for the so called Inozemtsev $BC_2$ Hamiltonian
$$
H=\Delta+ 2 m (m-1) (\wp(x+y)+\wp(x-y))+\sum_{i=0}^3 n_i(n_i-1) (\wp(x+\omega_i)+\wp(y+\omega_i)),
$$
where $\omega_0=0, \omega_3=\omega_1+\omega_2$ and $\omega_1,\omega_2$ are the half-periods of 
the Weierstrass function $\wp(x)$. 
\end{example}

\subsection{Commutative subalgebras in  $U(\mathfrak{gl}_{N+1})$ and \\ quantum Calogero-Moser Hamiltonians}

A class of commutative subalgebras in   $U(\mathfrak{gl}_{n})$ was constructed in \cite{Vinberg}. These subalgebras are quantizations of commutative Poisson subalgebras generated by compatible constant and linear $\mathfrak{gl}_n$-Poisson brackets. The quantization recipe is very simple: 
any product $\prod_1^k x_i$ of commuting generators should be replaced by $\ds \frac{1}{k!}\sum_{\sigma\in S_k} \prod y_{\sigma(i)},$ where $y_i$ are non-commutative generators. 

The universal enveloping algebra $U(\mathfrak{gl}_{N+1})$ is an associative algebra generated by elements $e_{ij}$ and relations \eqref{univen}. 
 Consider the case $N=2$.   It turns out that the element  of the universal enveloping algebra  $U(\mathfrak{gl}_{3})$  
$$
H=H_0+H_1 g_2+H_2 g_2^2+H_3 g_3, 
$$
where 
$$
\begin{array}{c}
H_0=12 e_{12} e_{11}-12 e_{32} e_{13}-12 e_{33} e_{12} - e_{23}^2,\\[3mm]  H_1=-e_{21} +2 e_{21} e_{11}-e_{22} e_{21}-e_{31} e_{23}-12 e_{32}^2 - e_{33} e_{21},\\[3mm]
H_2=- e_{31}^2, \qquad \quad H_3=36 e_{32} e_{31}+3 e_{21}^2, 
\end{array}
$$
commutes with two third order elements of the form  
$$
\begin{array}{c}
K=K_0+K_1 g_2+K_2 g_3, \\[3mm] M=M_0 +M_1 g_2+M_2 g_3+M_3 g_2^2+M_4 g_2 g_3 +M_5 g_3^2+M_6 g_2^3.
\end{array}
$$
Here $g_2$ and $g_3$ are arbitrary parameters and 
$$
\begin{array}{c}
K_0=-e_{23}+ 2 e_{21} e_{13} - e_{23} e_{22} - 36 e_{32} e_{12} + e_{33} e_{23} - e_{21} e_{13} e_{11} - e_{22} e_{21} e_{13} +
e_{23} e_{11}^2 + \\[3mm] 2 e_{23} e_{21} e_{12} - e_{23} e_{22} e_{11} + 12 e_{31} e_{12}^2 - e_{31} e_{23} e_{13} - 
12 e_{32} e_{12} e_{11} - e_{32} e_{23}^2 -\\[3mm] 12 e_{32}^2 e_{13} + 2 e_{33} e_{21} e_{13} -e_{33} e_{23} e_{11} + 
e_{33} e_{23} e_{22 }+ 12 e_{33} e_{32} e_{12},\\[3mm]
\end{array}
$$

$$
\begin{array}{c}
K_1=3 e_{31} e_{11}- 3 e_{31} e_{22} - 2 e_{32} e_{21} + e_{31} e_{21} e_{12} + e_{31} e_{22} e_{11} - e_{31} e_{22}^2 + 
e_{31}^2 e_{13} - \\[3mm] 2 e_{32} e_{21} e_{11} + e_{32} e_{22} e_{21} - 2 e_{32} e_{31} e_{23} - e_{33} e_{31} e_{11} + 
e_{33} e_{31} e_{22}+e_{33} e_{32} e_{21}, 
\end{array}
$$

$$
\begin{array}{c}
K_2=3\, (2 e_{31} e_{21} + e_{31} e_{22} e_{21} + e_{31}^2 e_{23} - e_{32} e_{21}^2 - e_{33} e_{31} e_{21});
\end{array}
$$

$$
\begin{array}{c}
M_0= 2 \Big(12 e_{13} e_{11} - 6 e_{22} e_{13} - 6 e_{33} e_{13} - 12 e_{13} e_{11}^2 - 6 e_{22} e_{13} e_{11} + 
   6 e_{22}^2 e_{13} +  18 e_{23} e_{12} e_{11} -  \\[3mm] 18 e_{23} e_{22} e_{12} + e_{23}^3 -
   216 e_{32} e_{12}^2 + 18 e_{32} e_{23} e_{13} + 30 e_{33} e_{13} e_{11} - 6 e_{33} e_{22} e_{13} - 
   12 e_{33}^2 e_{13}\Big),
\end{array}
$$

$$
\begin{array}{c}
M_1= -3 \Big(2 e_{23} e_{21} - 36 e_{31} e_{12} + 20 e_{32} e_{11} - 28 e_{32} e_{22} + 8 e_{33} e_{32} - 4 e_{23} e_{21} e_{11} + 
   2 e_{23} e_{22} e_{21} - \\[3mm]  12 e_{31} e_{12} e_{11} -
   e_{31} e_{23}^2  + 8 e_{32} e_{11}^2 + 36 e_{32} e_{21} e_{12} + 4 e_{32} e_{22} e_{11} - 4 e_{32} e_{22}^2 -
   24 e_{32} e_{31} e_{13} - \\[3mm] 12 e_{32}^2 e_{23} + 2 e_{33} e_{23} e_{21} + 12 e_{33} e_{31} e_{12} - 
   20 e_{33} e_{32} e_{11} + 4 e_{33} e_{32} e_{22} + 8 e_{33}^2 e_{32}\Big),
\end{array}
$$

$$
\begin{array}{c}
M_2= -18 \Big(4 e_{31} e_{11}  - 2 e_{31} e_{22} - 2 e_{33} e_{31} - e_{23} e_{21}^2 - 2 e_{31} e_{11}^2 - 
   6 e_{31} e_{21} e_{12} + 2 e_{31} e_{22} e_{11} - 
   2 e_{31} e_{22}^2 + \\[3mm] 6 e_{31}^2 e_{13} + 6 e_{32} e_{22} e_{21} + 24 e_{32}^3 + 
   2 e_{33} e_{31} e_{11} + 2 e_{33} e_{31} e_{22} - 
   6 e_{33} e_{32} e_{21} - 2 e_{33}^2 e_{31}\Big),\\[3mm]
\end{array}
$$

$$
\begin{array}{c}
M_3=  -3 \Big(2 e_{31} e_{21} - 2 e_{31} e_{21} e_{11} + e_{31} e_{22} e_{21} + 
   e_{31}^2 e_{23} - 24 e_{32}^2 e_{31} + e_{33} e_{31} e_{21}\Big),
\end{array}
$$
$$
\begin{array}{c}
\qquad \qquad \qquad \quad M_4= 9 \Big(e_{31} e_{21}^2  - 12 e_{32} e_{31}^2\Big), \qquad M_5= 108 e_{31}^3, \qquad M_6=-2  e_{31}^3.
\end{array}
$$
One can verify that  $[K,M]=0$.  Thus, we get a  commutative subalgebra in $U(\mathfrak{gl}_{3})$  generated by the elements $H,K,M$ and by the three central elements of  $U(\mathfrak{gl}_{3})$ of order 1,2, and 3.

This subalgebra generates ``integrable'' operators\footnote{This means that there exist ``rather many'' operators commuting with them.} by different representations of $U(\mathfrak{gl}_3)$ by differential, difference and  $q$-difference operators.

In particular, the substitution of differential operators  \eqref{opera} with two independent variables for $e_{ij}$ 
 maps the element $H$ to a polynomial form $P_2$ for the elliptic Calogero-Moser Hamiltonian \eqref{CM} with $N=2$, the element $M$ to a third order differential operator that commutes with $P_2$, and the element $K$ to zero.  The parameters $g_2$ and $g_3$ coincide with the invariants of the Weierstrass function $\wp(x)$ from \eqref{CM}.    
 
\begin{remark} The representation of $U(\mathfrak{gl}_3)$ by the matrix unities in ${\rm Mat}_3$ maps $H,K$ and $M$ to zero. 
\end{remark}

The representation defined by
$$
e_{ij}\rightarrow z_i \frac{\partial}{\partial z_j}
$$
maps $H$ to a homogeneous differential operator with 3 independent variables of the form 
${\cal H}=\sum_{i\ge j} a_{ij}\,\frac{\partial^2}{\partial z_i \partial z_j},$ where 
$$
\begin{array}{c}
a_{11}=-2 g_2 z_1 z_2-3 g_3 z_2^2+g_2^2 z_3^2, \qquad a_{22}=12 g_2 z_3^2, \qquad a_{33}=z_2^2, \\[3mm]
a_{21}=-12 z_1^2+g_2 z_2^2-36 g_3 z_3^2, \qquad a_{31}=2\, g_2 z_2 z_3, 
\qquad a_{32}=24\, z_1 z_3,
\end{array}
$$
and $M$ to an operator $\ds {\cal M}=\sum_{i\ge j\ge k} b_{ijk}\,\frac{\partial^3}{\partial z_i \partial z_j  \partial z_k}$ that commutes with $\cal H$. It is interesting 
that  the lower order terms are absent in both ${\cal H}$ and in ${\cal M}$.

A similar commutative subalgebra in $U(\mathfrak{gl}_4)$ \cite{matsok} yields a polynomial form for the elliptic Calogero-Moser Hamiltonian \eqref{CM} with $N=3$

\subsection{Bi-Hamiltonian origin of classical\\ elliptic Calogero-Moser models}

Consider the following limit procedure. Any element  $f\in U(\mathfrak{gl}_{n})$ is a polynomial in the non-commutative variables $e_{ij}$, which 
satisfy the commutator relation  \eqref{univen}. Taking all the terms of highest degree in $f$ and replacing there $e_{ij}$ by commutative variables $x_{ij}$, we get a polynomial that we call $symbol(f).$

It is known that for any elements $f$ and $g$ of $U(\mathfrak{gl}_{n})$
$$
symbol([f,g])=\{symbol(f), symbol(g)\},
$$
where  $\{ , \}$ is the linear Poisson bracket defined by
\begin{equation}\label{lin}
\{x_{ij}, x_{kl}\}=\delta_{j,k}\, x_{il}-\delta_{i,l} \, x_{kj},  \qquad i,j=1,\dots, n, 
\end{equation}
which corresponds to the Lie algebra $\mathfrak{gl}_{n}.$
In particular, if $[f,g]=0,$ then $$\{symbol(f),\,\, symbol(g)\}=0.$$

Consider polynomials in the commutative variables $x_{ij}$. We will regard $x_{ij}$ as the entries of a matrix $X$. 
Applying the limit procedure to the generators of the commutative subalgebra in $U(\mathfrak{gl}_{3})$ described in Subsection 3.1.3, we get the polynomials 
\begin{equation}\label{coma}
\begin{array}{c}
c_1=\mbox{tr}\,X, \qquad c_2=\mbox{tr}\,X^2, \qquad c_3=\mbox{tr}\,X^3, \\ [3mm]
 h=h_0+h_1 g_2+h_2 g_2^2+h_3 g_3, \qquad k=k_0+k_1 g_2+k_2 g_3,    \\ [3mm]
 m=m_0 +m_1 g_2+m_2 g_3+m_3 g_2^2+m_4 g_2 g_3 +m_5 g_3^2+m_6 g_2^3,
\end{array}
\end{equation}
where 
$$
\begin{array}{c}
h_0=12 x_{12} x_{11}-12 x_{32} x_{13}-12 x_{33} x_{12} - x_{23}^2,\\[3mm]  \qquad h_1= 2 x_{21} x_{11}-x_{22} x_{21}-x_{31} x_{23}-12 x_{32}^2 - x_{33} x_{21},\\[3mm]
h_2=- x_{31}^2, \qquad \quad h_3=36 x_{32} x_{31}+3 x_{21}^2, 
\end{array}
$$
 $$
\begin{array}{c}
k_0=  - x_{21} x_{13} x_{11} - x_{22} x_{21} x_{13} +
x_{23} x_{11}^2 +  2 x_{23} x_{21} x_{12} - \\[3mm] \qquad \qquad x_{23} x_{22} x_{11} + 12 x_{31} x_{12}^2 - x_{31} x_{23} x_{13} - 
12 x_{32} x_{12} x_{11} - x_{32} x_{23}^2 -\\[3mm] \qquad \qquad 12 x_{32}^2 x_{13} + 2 x_{33} x_{21} x_{13} -x_{33} x_{23} x_{11} + 
x_{33} x_{23} x_{22 }+ 12 x_{33} x_{32} x_{12},\\[3mm]
\end{array}
$$

$$
\begin{array}{c}
k_1=  x_{31} x_{21} x_{12} + x_{31} x_{22} x_{11} - x_{31} x_{22}^2 + 
x_{31}^2 x_{13} -  2 x_{32} x_{21} x_{11} +\\[3mm] \qquad  x_{32} x_{22} x_{21} - 2 x_{32} x_{31} x_{23} - x_{33} x_{31} x_{11} + 
x_{33} x_{31} x_{22}+x_{33} x_{32} x_{21}, 
\end{array}
$$

$$
\begin{array}{c}
k_2=3\, (  x_{31} x_{22} x_{21} + x_{31}^2 x_{23} - x_{32} x_{21}^2 - x_{33} x_{31} x_{21});
\end{array}
$$

$$
\begin{array}{c}
m_0= 2 \Big( - 12 x_{13} x_{11}^2 - 6 x_{22} x_{13} x_{11} + 
   6 x_{22}^2 x_{13} +  18 x_{23} x_{12} x_{11} -  18 x_{23} x_{22} x_{12} + \\[3mm] \qquad \qquad x_{23}^3 -
   216 x_{32} x_{12}^2 + 18 x_{32} x_{23} x_{13} + 30 x_{33} x_{13} x_{11} - 6 x_{33} x_{22} x_{13} - 
   12 x_{33}^2 x_{13}\Big),
\end{array}
$$

$$
\begin{array}{c}
m_1= -3 \Big(  - 4 x_{23} x_{21} x_{11} + 
   2 x_{23} x_{22} x_{21} -  12 x_{31} x_{12} x_{11} - 
   x_{31} x_{23}^2  + 8 x_{32} x_{11}^2 +\\[3mm]  \qquad 36 x_{32} x_{21} x_{12} + 4 x_{32} x_{22} x_{11} - 4 x_{32} x_{22}^2 -
   24 x_{32} x_{31} x_{13} -  12 x_{32}^2 x_{23} +\\[3mm] \qquad 2 x_{33} x_{23} x_{21} + 12 x_{33} x_{31} x_{12} - 
   20 x_{33} x_{32} x_{11} + 4 x_{33} x_{32} x_{22} + 8 x_{33}^2 x_{32}\Big),
\end{array}
$$

$$
\begin{array}{c}
m_2= -18 \Big( - x_{23} x_{21}^2 - 2 x_{31} x_{11}^2 - 
   6 x_{31} x_{21} x_{12} + 2 x_{31} x_{22} x_{11} - 
   2 x_{31} x_{22}^2 +  6 x_{31}^2 x_{13} + \\[3mm] \qquad \qquad 6 x_{32} x_{22} x_{21} + 24 x_{32}^3 + 
   2 x_{33} x_{31} x_{11} + 2 x_{33} x_{31} x_{22} - 
   6 x_{33} x_{32} x_{21} - 2 x_{33}^2 x_{31}\Big),\\[3mm]
\end{array}
$$

$$
\begin{array}{c}
m_3=  -3 \Big( - 2 x_{31} x_{21} x_{11} + x_{31} x_{22} x_{21} + 
   x_{31}^2 x_{23} - 24 x_{32}^2 x_{31} + x_{33} x_{31} x_{21}\Big),
\end{array}
$$
$$
\begin{array}{c}
\qquad \qquad \qquad \quad m_4= 9 \Big(x_{31} x_{21}^2  - 12 x_{32} x_{31}^2\Big), \qquad m_5= 108 x_{31}^3, \qquad m_6=-2  x_{31}^3.
\end{array}
$$
These six polynomials commute with each other with respect to the linear $\mathfrak{gl}_{3}$-Poisson bracket  \eqref{lin}.

It can be verified that elements of the universal enveloping algebra can be reconstructed from the polynomials  \eqref{coma} by the quantization procedure described at the beginning of Subsection 3.1.3. 

\subsubsection{Quadratic Poisson bracket}

Consider the following quadratic bracket  
\begin{equation}\label{qvpu}
\{f,g\}_2=\{f,g\}_a+\kappa \{f,g\}_b+\kappa^2 \{f,g\}_c,
\end{equation}
where $\kappa$ is an arbitrary parameter, 
$$
\{f,g\}_a=-3\, \mbox{tr}(X)\, \{f,g\}_1, \qquad 
\{f,g\}_c=Z_1(f) Z_2(g)-Z_1(g) Z_2(f).
$$ 
Here the bracket $\{ , \}_1$ is defined by \eqref{lin}.
The above vector fields $Z_i$ are defined as follows: 
$$
Z_1(f)=\sum_{i=1}^{3} \frac{\partial f}{\partial x_{ii}}, \qquad Z_2(f)=\{h,f\}_1,
$$
and
$$
\{f,g\}_b=Z_3(\{f,g\}_1) -\{Z_3(f),g\}_1-\{f,Z_3(g)\}_1,
$$
where
$$
Z_3(f)=\sum_{i,j=1}^{3} G_{i,j}\,\frac{\partial f}{\partial x_{ij}}.
$$
Here $\{\cdot,\cdot\}_1$ denotes the linear bracket   \eqref{lin}.  The coefficients of the vector field $Z_3$ are given by
$$
G_{1,1}=(-2 x_{11} x_{23}+x_{22} x_{23}+36 x_{12} x_{32}+x_{23} x_{33})+x_{31}(x_{11}-2 x_{22}+x_{33})\,g_2+9\, x_{21} x_{31}\, g_3,
$$

$$
G_{2,2}=-G_{1,1}, \qquad G_{3,3}=0,
$$

$$
G_{1,2}=(x_{11} x_{13} + x_{13} x_{22} - 3\, x_{12} x_{23} - 2 x_{13} x_{33})+(3 x_{12} x_{31} + 5 x_{11} x_{32} - 4 x_{22} x_{32} - 
 x_{32} x_{33})\,g_2-
$$
$$
-3 (2 x_{11} x_{31} - x_{22} x_{31} - 3 x_{21} x_{32} - x_{31} x_{33})\,g_3,
$$

$$
G_{1,3}=3 x_{13} x_{23} -(x_{11} - x_{22}) (x_{11} + x_{22} - 2 x_{33})\,g_2-3 x_{21} (x_{11} + x_{22} - 2 x_{33})\,g_3,
$$

$$
G_{2,1}=-3 (x_{21} x_{23} + 12 x_{12} x_{31} + 4 x_{11} x_{32} - 
   8 x_{22} x_{32} + 4 x_{32} x_{33})-6 x_{21} x_{31}\,g_2,
$$

$$
G_{2,3}=3 (4 x_{11} x_{12} + 4 x_{12} x_{22} + x_{23}^2 - 8 x_{12} x_{33})+x_{21} (x_{11} + x_{22} - 2 x_{33})\,g_2,
$$

$$
G_{3,1}=2 (x_{11} x_{21} + x_{21} x_{22} + 18 x_{32}^2 - 2 x_{21} x_{33})-3 x_{31}^2\,g_2,
$$

$$
G_{3,2}=-(x_{11} - x_{22}) (x_{11} + x_{22} - 2 x_{33})+6 x_{31} x_{32}\,g_2-9 x_{31}^2\, g_3.
$$
\begin{theorem} \begin{itemize}
\item[i)] Formula \eqref{qvpu} defines a Poisson bracket;
\item[ii)] This  quadratic bracket is compatible {\rm (see Section \ref{sec14})} with the linear $\mathfrak{gl}_{3}$-Poisson bracket  \eqref{lin};
\item[iii)] The Casimir function of the pencil of these two brackets generates {\rm (see Theorem \ref{biH})} the commutative Poisson subalgebra described in Subsection 3.1.4;
\end{itemize}
\end{theorem}
\begin{conjecture} The Poisson bracket \eqref{qvpu} is the elliptic Poisson bracket of the type $q_{9,2}$ {\rm(see~\cite{feyod})} written in an unusual basis.
\end{conjecture}

To get the classical elliptic Calogero-Moser Hamiltonian, one should use the following classical limit of formulas  \eqref{opera}:
\begin{equation}\label{clastr}
 x_{i+1,j+1}=q_i \, p_j,\qquad x_{1,i+1}=p_i ,\qquad  x_{1,1}=-\sum_{j=1}^{N}  q_j p_j +\beta (N+1), \qquad x_{i+1,1}=q_i\, x_{1,1},
\end{equation}
where $p_i$ and $q_i$ are canonical variables for the standard constant Poisson bracket \eqref{stand}. 
One can verify that $p_i,\,q_i$  are Darboux coordinates on the minimal symplectic leaf of the  $\mathfrak{gl}_{N+1}$-Poisson bracket.  This leaf 
is the orbit of the diagonal matrix $\rm{diag}(\beta (N+1),0,0,...,0).$  

In the case $N=2$ after substitution  \eqref{clastr} into  \eqref{coma} we get  polynomials in the canonical variables $p_i, q_i$ commuting with respect to \eqref{stand}. The element $h$ becomes the Calogero-Moser Hamiltonian written in unusual coordinates, the element $k$ vanishes, the element $m$  converts to the integral of third degree in momenta that commutes with the Hamiltonian $h$, and the Casimir functions $c_i$ become constants. 

To bring the Hamiltonian and the cubic integral to the standard Calogero-Moser form 
\begin{equation}\label{CALst}
h_N=- \sum_{i=1}^{N+1} p_i^2+ \beta(\beta-1)\sum_{i\neq j}^{N+1}\wp(q_i-q_j). 
\end{equation}
one has to apply a canonical transformation, where the transformation rule for the coordinates is given by  \eqref{tran}.

In the case $N=3$ a quadratic bracket  exists. This bracket is compatible with the linear $\mathfrak{gl}_{4}$-Poisson bracket and generates the corresponding classical elliptic Calogero-Moser Hamiltonian in the same way as for $N=2$.

\begin{conjecture} For any $N$ the classical elliptic Calogero-Moser Hamiltonian \eqref{CALst}
can be obtained from the elliptic Poisson bracket of the $q_{(N+1)^2,N}$-type by the above procedure.
\end{conjecture}

\begin{op} For the elliptic bracket $\{ , \}$ of the $q_{(N+1)^2,N}$-type find a basis such that the linear bracket $\{ , \}_1$ 
compatible with $\{ , \}$ has the canonical form \eqref{lin}.

\end{op}

\section{Bi-Hamiltonian formalism and compatible algebras}

\subsection{Compatible Lie algebras}

Suppose that two linear finite-dimensional Poisson brackets are compatible. As it was mentioned in the introduction, each of these brackets corresponds to a Lie algebra. Denote by $[\cdot,\cdot]_1$ and $[\cdot,\cdot]_2$ the operations of these algebras. It is clear that the Poisson brackets are compatible iff 
the operation 
$ \lambda_1 [\cdot,\cdot]_1 + \lambda_2 [\cdot,\cdot]_2 
$
is a Lie bracket for any $\lambda_i$. Without loss of generality we may put $\lambda_1=1.$

\begin{definition} Two Lie brackets $[\cdot,\cdot]$ and $[\cdot,\cdot]_1$ defined on the same vector space ${\bf V}$ are called {\it compatible} if the operation 
\begin{equation}\label{pensil}
[\cdot, \cdot]_{\lambda} = [\cdot,\cdot] + \lambda [\cdot,\cdot]_1 
\end{equation}
is a Lie bracket for any $\lambda.$
\end{definition}

Suppose that  $[\cdot, \cdot]$ corresponds to a semi-simple Lie algebra ${\cal G}.$ 
The following classification problem arises: 
\begin{op}\label{op4} Describe all possible Lie brackets $[\cdot, \cdot]_1$  compatible with a given semi-simple Lie bracket $[\cdot, \cdot].$ 
\end{op}

The Lie algebra with bracket \eqref{pensil} can be regarded as a linear deformation of the algebra ${\cal G}$. 
Since any semi-simple Lie
algebra is rigid (i.e., the cohomology $H^2[{\cal G},{\cal G}]$  of the Lie algebra ${\cal G}$ vanishes), the bracket  \eqref{pensil} is isomorphic to
$[\cdot,\cdot].$ This means that there exists a formal series of
the form
$$
A_{\lambda}=I+ R \ \lambda+ S \ \lambda^2 + \cdots \,\, , $$
where the
coefficients $R, S, \dots$ are constant linear operators on ${\cal G}$ and $I$ is the identity operator, such that
\begin{equation}
A_{\lambda}^{-1}\,[A_{\lambda}(X), \, A_{\lambda}(Y)]=[X, \, Y]+
\lambda \, [X, \, Y]_1. \label{sog}
\end{equation}
It follows from  \eqref{sog} that
\begin{equation}\label{br2}
[X, \, Y]_1=[R(X), \, Y]+[X, \, R(Y)]-R([X, \, Y]),
\end{equation}
where $R$ is the corresponding coefficient of $A_{\lambda}$.

\begin{lemma} The bracket $[\cdot,\cdot]_1$ is a Lie bracket iff
there exists a linear operator $S:\, {\cal G}\to {\cal G}$ such that
$$
R\big([R(X), \ Y] - [R(Y), \ X]\big) - [R(X), \ R(Y)] -
R^2([X, \ Y]) =[S(X), \ Y] - [S(Y), \ X] - S([X, \ Y]).
$$
\end{lemma}

In the special case $S=0$ the relation from Lemma 3.3  takes the form
\begin{equation}\label{RR}
R\big([R(X), \ Y] - [R(Y), \ X]\big) - [R(X), \ R(Y)] -
R^2([X, \ Y])=0.
\end{equation}

We present below examples \cite{sokolgok2} of compatible Lie brackets with the
corresponding operators $A_{\lambda}$ and $R$.

\begin{example}\label{rman} Let ${\cal G}$ be the Lie algebra associated with
an associative algebra ${\cal A}$, then we can take for $R$ the
operator of left multiplication by any element $r$. In this case
$$
[X,\,Y]_1= X r Y - Y r X,  \qquad A_{\lambda}: \, g \rightarrow
g+\lambda r g.
$$
\end{example}
\begin{example} Let ${\cal A}$ be an associative algebra with an
involution $*$, ${\cal G}$ be the Lie algebra of all
skew-symmetric elements of ${\cal A}$, $r$ be an element symmetric
with respect to $*$. In this case the operator
$$
A_{\lambda}: \, g \rightarrow \sqrt{1+r \lambda} \,\, g \,\,
\sqrt{1+r \lambda},
$$
can be taken as $A_{\lambda}$, $\ds R(X)=\frac{1}{2}(r X+X r)$ and $[X,\,Y]_1= X r Y - Y r X$.
\end{example}

\begin{example} Let $\varphi$ be an automorphism of order $n$ of
a Lie algebra ${\cal G}$, $\,\,{\cal G}_i$ be the eigenspace of
the operator $\varphi$ corresponding to eigenvalue
$\varepsilon^i,$ where $\varepsilon^n=1$. Then the vector space
${\cal G}_0$ is a Lie subalgebra.
 Suppose that  $${\cal G}_0={\cal G}_{+}\oplus {\cal G}_{-}$$ with
vector spaces ${\cal G}_{+}$ and ${\cal G}_{-}$ being subalgebras
of ${\cal G}_0$. Consider the operator $A_{\lambda}$ acting on
${\cal G}_{+}$, ${\cal G}_{-},$ and ${\cal G}_{i}, \,\,\, i>0,$ by
means of multiplication by
$$
1+\alpha \lambda, \quad  1+\beta \lambda, \quad \sqrt[n]{(1+\alpha
\lambda)^i(1+\beta \lambda)^{n-i}},
$$
respectively.  The operator $R$ is given by
$$
R(g)=\alpha \,g_{+}+\beta\,g_{-}+\sum_{i=1}^{n-1}
\left(\frac{i}{n} \alpha+\frac{n-i}{n}\beta\right)\, g_i,
$$
where $g_{\pm}$  mean the projections of $g$ onto ${\cal G}_{\pm}.$
\end{example}

\begin{example}\label{Ex4} This class of compatible brackets is related to
finite dimensional $\Z_2 \times \Z_2$-graded Lie algebras.
 Recall the definition. Let $\varphi$ and $\psi$ be two
automorphisms of ${\cal G}$ commuting with each other and such that
$\varphi^{2}=\psi^{2}=Id.$  The decomposition ${\cal G}=\oplus
{\cal G}_{ij}$, $i,j=\pm 1$ into a direct sum of the following
four invariant vector spaces
$$
{\cal G}_{ij}=\{g\in {\cal G} \, \vert \, \varphi(g)=i \, g, \quad
\psi(g)=j \, g \}
$$
is called the $\Z_2 \times \Z_2$-gradation.

Define an operator $A_{\lambda}$ on the homogeneous components by
the formulas
$$
A_{\lambda}(g_{1,1})=(1+\gamma \lambda) \, g_{1,1}, \qquad 
A_{\lambda}(g_{-1,1})=\sqrt{(1+\beta \lambda)(1+\gamma \lambda)}
\, g_{-1,1}, 
$$
$$
A_{\lambda}(g_{1,-1})=\sqrt{(1+\alpha \lambda)(1+\gamma \lambda)}
\, g_{1,-1},
$$
$$
A_{\lambda}(g_{-1,-1})=\sqrt{(1+\alpha \lambda)(1+\beta \lambda)}
\, g_{-1,-1}+ \lambda \sqrt{(\gamma-\alpha)(\gamma-\beta)} \,
\rho (g_{-1,-1}).
$$
 Here $\alpha$, $\beta$, and $\gamma$ are arbitrary constants and
the operator $\rho: \ {\cal G}_{-1,-1} \rightarrow {\cal G}_{1,1}$
is any solution of the modified Yang-Baxter equation
$$
\rho \big([\rho(X), \ Y] - [\rho(Y), \ X]\big) - [\rho(X), \
\rho(Y)] - [X, \ Y]=0. $$ In other words, the Lie algebra ${\cal
G}_{1,1}\oplus {\cal G}_{-1,-1}$ is assumed to be decomposed into
a direct sum of ${\cal G}_{1,1}$ and some complementary subalgebra
${\cal B}$ and $r$ denotes the projection onto ${\cal G}_{1,1}$
parallel to ${\cal B}$.

The operator $R$ is defined by
$$
R(g)=\gamma \ g_{1,1}+\frac{1}{2} (\alpha+\gamma) \, g_{1,-1}
+\frac{1}{2} (\beta+\gamma) \, g_{-1,1}+\frac{1}{2}
(\alpha+\beta) \, g_{-1,-1} +\sqrt{(\gamma-\alpha)(\gamma-\beta)}
\, \rho(g_{-1,-1}).
$$
\end{example}
\begin{remark} The operator $A_{\lambda}$ in Example \ref{Ex4} can be parametrized by points of the elliptic curve 
$$
\lambda_1^2 - \alpha = \lambda_2^2 - \beta = \lambda_3^2 - \gamma = \frac{1}{\lambda}.
$$
\end{remark}
\begin{remark}
A wide class of $\Z_2 \times  \Z_2$-graded Lie algebras can be constructed as follows.
Let ${\cal G} =\oplus {\cal G}_i$ be $\Z$-graded Lie algebra that possesses an involution
$\psi$ such that $\psi({\cal G}_i)={\cal G}_{-i}$. In particular, such an involution exists for any standard gradation of a simple Lie algebra.   
Taking for $\varphi$ the involution $\varphi(X) = (-1)^i\,X$, where $X \in {\cal G}_i$,  we define on 
$\cal G$ a structure of  $\Z_2 \times  \Z_2$-graded algebra.
\end{remark}

\subsubsection{Applications}

For applications the operator $A_{\lambda}$ has to be written in a closed form, i.e. as an analytic operator-valued function in $\lambda$. In known examples the $\lambda$-dependence is rational, trigonometric or elliptic. 

Very often a Lax pair for the corresponding bi-Hamiltonian model can be written in terms of  $A_{\lambda}$ \cite{sokolgok1, golsokNL}.

\begin{application}
Consider the following system of ODEs:
$$
w_{t}=[w,\,v]+w*w, \qquad v_{t}=[w,\,u]+w*v, \qquad u_{t}=w*u,
$$
where  
$$
X*Y=[R(X),\, Y]-[X,\,R^{*}(Y)]+R^{*}([X,Y]),
$$
and $R^{*}$ stands for the operator
adjoint to $R$ with respect to the Killing form.
Then the operators
\begin{equation}\label{calAL}
{\cal L}=(A_{\lambda}^{-1})^{*}(\lambda u+v+\lambda^{-1} w), \qquad
{\cal A}=\lambda^{-1}\,A_{\lambda}(w)
\end{equation}
form a Lax pair for this system. As
usual, the integrals of motion are given by $\mbox{tr}\,{\cal
L}^k, \,$ $\, k=1,2...$

In the case of Example \ref{rman}  the first bracket $[\cdot,\cdot]$ is a
standard matrix commutator,   the second bracket is given by
$[x,\, y]_1=x r y-y r x$  and $X*Y=r X Y - Y X r,$ where $r$ is an arbitrary matrix. We have  
$$
R(w)=r w, \qquad A_{\lambda} (w)=(I+\lambda r)\, w, \qquad (A_{\lambda}^{-1})^{*}(w)=w \,(I+\lambda r)^{-1}\,  $$
If $u=v=0$, then
$$ \qquad {\cal L}=w (I+\lambda r)^{-1}, \qquad {\cal A}=\lambda^{-1} (I+\lambda r)\,w.
$$
The Lax equation is equivalent (up to $t\to -t$) to  equation \eqref{euler}, where ${\bf U}=x,\,a=r$. For the Lax 
pair \eqref{calAL} we arrive at \eqref{uvw}, where $u,v$ and $w$ are generic matrices. 
 
\end{application}

\begin{application} Consider the system of equations
\begin{equation}
u_{x}=[u,\ v], \qquad v_{y}=[v, \ u]_{1}, \label{twolie}
\end{equation}
where $u$ and $v$ belong to a vector space ${\bf V}$ equipped
with two Lie brackets $[\cdot,\cdot]$ and $[\cdot,\cdot]_1$. For
the well-known integrable principle chiral model
$$
u_x  = [u, \ v], \qquad v_y  = [u, \ v]
$$
the brackets $[\cdot,\cdot]$ and $[\cdot,\cdot]_1$ are identical
up to sign.
\end{application}

\begin{theorem} If the Lie brackets $[\cdot,\cdot]$ and
$[\cdot,\cdot]_1$ are compatible, then the hyperbolic system
 \eqref{twolie} possesses the following Lax pair
$$
{\cal L}=\frac{d}{dy}+\frac{1}{\lambda}\,A_{\lambda}(u), \qquad
{\cal A}=\frac{d}{dx}+A_{\lambda}(v).$$
\end{theorem}

Compatible Lie algebras and the corresponding operator $A_{\lambda}$ are closely related to different kinds of the Yang-Baxter equation \cite{golsokTMF, odsok1}.

\begin{application}  Consider the classical Yang-Baxter equation
$$
[r^{1,2}(\lambda, \mu),\, r^{1,3}(\lambda, \nu)]+[r^{1,2}(\lambda,
\mu),\, r^{2,3}(\mu, \nu)] + [r^{1,3}(\lambda, \nu),\,
r^{2,3}(\mu, \nu)]=0,
$$
where $$r(x,y)=\sum_i a_i(x,y)\otimes b_i(x,y)$$ is a function of
two complex variables with values in $\mathfrak{gl}_N\otimes \mathfrak{gl}_N$ and  $r^{i,j}$, where $1\le i,j \le 3, \quad i\ne j,$ are the following
functions
$$
r^{1,2}(\lambda,\mu)=\sum_i a_i(\lambda,\mu)\otimes
b_i(\lambda,\mu)\otimes 1,
$$
$$
r^{1,3}(\lambda,\nu)=\sum_i a_i(\lambda,\nu)\otimes 1\otimes
b_i(\lambda,\nu),
$$
$$
r^{2,3}(\mu,\nu)=\sum_i  1\otimes a_i(\mu,\nu)\otimes b_i(\mu,\nu)
$$
with values in $\mathfrak{gl}_N\otimes \mathfrak{gl}_N\otimes \mathfrak{gl}_N$.
We suppose, as usual, that the unitary condition
$$
r^{1,2}(\lambda, \mu)=-r^{2,1}(\mu,\lambda)
$$
holds, where
$$
r^{2,1}(\mu,\lambda)=\sum_i b_i(\mu,\lambda)\otimes
a_i(\mu,\lambda)\otimes 1.
$$
\end{application}

\begin{theorem}  Let $[\cdot, \cdot]_1$ and $[\cdot,\cdot]_2$ be
two compatible Lie brackets on a $N$-dimensional vector space
${\bf V}.$ Suppose that there exists a non-degenerate symmetric
form $\omega(X,Y)$ on ${\bf V}$ invariant with respect to both
brackets $[\cdot, \cdot]_{1,2}.$ Let ${\bf e_1},\dots, {\bf e_N}$
be a basis orthonormalized  with respect to $\omega.$ Then
$$
r(x, y)=\sum_{i=1}^N \frac{ (x\, ad_1 {\bf e_i}+ad_2 {\bf
e_i})\otimes (y\, ad_1 {\bf e_i}+ad_2 {\bf e_i})}{x-y}
$$
satisfies the classical Yang-Baxter equation. Here $ad_i q$ are
linear operators defined by
$$
ad_i q (p)=[q,\,p]_i, \qquad i=1,2.
$$
\end{theorem}

\begin{application}
The operator form of the classical Yang-Baxter equation \cite{semen} is given by
$$
[r(u,w)x,r(u,v)y]=r(u,v)[r(v,w)x,y]+r(u,w)[x,r(w,v)y].
$$
Here $r(u,v)\in {\rm End}({\cal G})$. The solution is called unitary if $\langle x,r(u,v)y \rangle=-\langle r(v,u)x,y\rangle$ for the Killing form of ${\cal G}.$
\end{application}

\begin{theorem} If $A_{\lambda}$  
satisfies  \eqref{sog};  then
\begin{equation}\label{ruv}
r(u,v)=\frac{1}{u-v}A_u A_v^{-1}
\end{equation}
satisfies  the Yang-Baxter equation.
\end{theorem}
\begin{remark}  The $r$-matrix \eqref{ruv} is unitary with respect to the form 
$\langle \cdot, \cdot\rangle$ if the operator $A_{\lambda}$ is orthogonal. In this case the formula \eqref{sog} implies that the form $\langle \cdot, \cdot\rangle$
is invariant also with respect to the bracket $[\cdot,\cdot]_1$.
\end{remark}

\begin{application} It is known that  decompositions \eqref{decompgen} of the loop algebra over a Lie algebra ${\cal G}$
into a sum of the Lie algebra of all Taylor series and a
factoring subalgebra ${\cal U}$   gives rise to Lax representations for diverse
integrable models (see Section 2.3).

A factoring subalgebra ${\cal U}$ is said to be {\it
homogeneous} if it has the multiplicand $\lambda^{-1}$. This means that 
$$
\frac{1}{\lambda} \ {\cal U} \subset {\cal U}. $$ 
It turns out \cite{golsokNL}
that for any semi-simple Lie algebra ${\cal G}$ with a bracket
$[\cdot,\cdot]$ there exists a one-to-one correspondence between
the homogeneous subalgebras and brackets $[\cdot,\cdot]_1$
compatible with $[\cdot,\cdot].$
\end{application}

\subsection{Compatible associative algebras}

While Problem \ref{op4} for Lie algebras seems to be very difficult, a similar problem for 
associative algebras is more treatable \cite{odsok1,odsok2,odsok3}. 

\begin{definition} Two associative algebras with multiplications $\star$ and $\circ$
 defined on
the same finite dimensional vector space $\bf V$ are said to be
compatible if the multiplication
$$
a \bullet b= a\star b+ \lambda \,a\circ b
$$
is associative for any constant $\lambda$.
\end{definition}
\begin{remark}
For compatible associative algebras with multiplications $\star$ and $\circ$ the Lie algebras with the brackets $[X, \, Y]_1= X \star Y - Y\star X$ 
and $[X, \, Y]_2= X \circ Y - Y\circ X$ are compatible. 
\end{remark}

\begin{example}\label{Ex1}  Let $\bf V$ be the vector space of polynomials of degree $\le k-1$
in one variable, $\mu_1$ and
$\mu_2$ be polynomials of degree $k$ without common roots. Any polynomial $Z$, where $\hbox{deg} Z\le 2 k-1,$
can be uniquely represented in the form $ Z=\mu_1 P+
\mu_2 Q, $ where $P, Q\in {\bf V}$.
Define multiplications $\circ$ and $\star$ on ${\bf V}$ by the formula
$$
X\,Y=\mu_1 (X\circ Y) +\mu_2 (X\star  Y), \qquad  X,Y\in {\bf V}.
$$
It can be verified that associative algebras with products $\circ$ and $\star$  are compatible.  
\end{example}

\begin{example} Let ${\bf e}_{1},\dots,{\bf e}_{m}$ be a basis of
$\bf V$ and let the multiplication $\star$ be given by
$$
{\bf e}_{i}\star {\bf e}_{j}=\delta^{i}_{j} {\bf e}_{i}.
$$
 Let
$$
r_{ii}=q_{0}- \sum_{k\ne i} r_{ki}, \qquad    r_{ij}=\frac{q_{i}
\lambda_{i}}{\lambda_{j}-\lambda_{i}}, \qquad i\ne j,
$$
where $\lambda_{i},  q_{j}$ are arbitrary constants. Then the product defined by the formula
$$
{\bf e}_{i}\circ {\bf e}_{j}=r_{ij} {\bf e}_{j}+r_{ji} {\bf
e}_{i}-\delta^{i}_{j}\sum_{k=1}^{m} r_{ik} {\bf e}_{k}
$$
is associative and compatible with $\star$. Since this product is linear with
respect to the parameters $q_{i},$ we have constructed a family of $m+1$
pairwise compatible associative multiplications. 
\end{example}

Suppose that the associative algebra $A$ with
multiplication $\star$ is semi-simple. Since such algebras are rigid, the associative algebra with the multiplication $\bullet$ is isomorphic to the algebra $A$ for almost all values of the parameter $\lambda$. Hence there exists a linear operator $S_{\lambda}$ on $\bf V$ such that
$$
S_{\lambda}(X) \star S_{\lambda}(Y)= S_{\lambda} \Big( X\star Y+\lambda \
X \circ Y\Big).  
$$
If
$$
S_{\lambda}={\bf 1}+ R \ \lambda+ O(\lambda^2),    \label{RS}
$$
then the multiplication $\circ $ is given by
\begin{equation} \label{mult2}
X \circ Y =R(X) \star Y+X \star R(Y)-R(X \star Y).
\end{equation}

Consider the case when the associative algebra  with
multiplication $\star$ coincides with ${\rm Mat}_m$. Then 
$R:\, {\rm Mat}_m \rightarrow {\rm Mat}_m$ is a linear operator on the space of $m \times m$-matrices. We will omit the sign  $\star$. In other words, we investigate {\it associative linear deformations of the matrix product}. 

\begin{example}\label{Ex2} (see Example \ref{rman}).  Let $c$ be an element of ${\rm Mat}_m$ and $R: X
\rightarrow c X$ be the operator of left multiplication by $c$ .
Then the corresponding multiplication $X\circ Y=X\, c \, Y$
is associative and compatible with the standard matrix product in ${\rm Mat}_m$.
\end{example}

\begin{example}\label{Ex3}  Suppose that $a,b\in {\rm Mat}_2;$ then the product
\begin{equation}\label{abprod}
X\circ Y=(a X-X a)\,(b Y-Y b)
\end{equation}
is compatible with the standard  product in ${\rm Mat}_2\,$. The
corresponding operator $R$ is given by $$R(X)=a\, (X b-b X).
$$
\end{example}

\begin{proposition} In the case of ${\rm Mat}_2$ any linear deformation of the matrix product is given
by one of these two examples.
\end{proposition}
\begin{remark}\label{rem39}
If $a,b\in {\rm Mat}_m, \,\, m >2,$ we need the additional assumption
$\,\,a^2=b^2={\bf 1}\,\,$ for \eqref{abprod} to be compatible with the product in ${\rm Mat}_m$.
\end{remark}

In the matrix case the operator $R$ is defined up to the transformation
\begin{equation}\label{Req}
R \to R + {\rm ad}_s,
\end{equation}
where $s\in  {\rm Mat}_m.$ For any $s$ this transformation does not change the multiplication \eqref{mult2}.

It is easy to see that any operator $R: {\rm Mat}_m\to {\rm Mat}_m$ can be written in the form 
$$
R(x)={\bf a}_1 \,x \,{\bf b}^1+...+{\bf a}_{p+1} \,x\, {\bf b}^{p+1},
$$
where $ {\bf a}_i,{\bf b}^i \in {\rm Mat}_m.$
We will assume  that $p$ is as small as possible. In this case the matrices ${\bf a}_1,\dots, {\bf a}_{p+1}$ as well as  ${\bf b}^1,\dots, {\bf b}^{p+1}$ are linearly independent. Using \eqref{Req}, we can represent $R(x)$ in the following form:
\begin{equation}\label{Rmat}
R(x)={\bf a}_1 \,x \,{\bf b}^1+...+{\bf a}_p \,x\, {\bf b}^p + {\bf c}\,x.
\end{equation}

\subsubsection{Integrable matrix ODEs related to $R$-operator}

Let $R:\, {\rm Mat}_m \rightarrow {\rm Mat}_m$  be a linear operator such that the product \eqref{mult2} is associative. Consider \cite{sokolgok1,odsok2}
  the following   matrix differential
equation:
\begin{equation}\label{difR}
\frac{d x}{dt}=[x,\,R(x)+R^*(x)], \qquad x(t) \in {\rm Mat}_m, 
\end{equation}
  and $R^{*}$ stands for the operator
adjoint to $R$ with respect to the bi-linear form $\langle x,\,y \rangle=\mbox{tr}\,(x \,y)$. For an operator written in the form \eqref{Rmat} we have 
$$
R^*(x)={\bf b}^1 \,x \,{\bf a}_1+...+ {\bf b}^p\,x\,{\bf a}_p + x\,{\bf c}.
$$
\begin{theorem}
Equation \eqref{difR} possesses the following Lax pair:
$$
L=\Big(S_{\lambda}^{-1}\Big)^{*}(x), \qquad A=\frac{1}{\lambda}\,S_{\lambda}(x).
$$
\end{theorem}
\begin{op} Show that equation \eqref{difR} is bi-Hamiltonian with the Hamiltonian operators
$$
{\cal H}_1 = ad_x, \qquad {\cal H}_2 = ad_x^{\, 1}, 
$$
where $ad_x^{\,1}$ is defined by the multiplication \eqref{mult2}.
\end{op}
\begin{example}\label{EEx2}
In the case of Example \ref{Ex2} we get
$$
\frac{d x}{dt}=[x,\, x c + c x]=x^2\,c-c\, x^2
$$
for $m\times m$-matrix $x$ and any constant matrix
$c$.
Under the reductions $x^T=-x, \, c^T=c$ the equation   is a
generalized symmetry for the $n$-dimensional Euler top.
\end{example}
 
\begin{example}\label{EEx3}
In the case of Remark \ref{rem39} we have \cite{odsok2}
\begin{equation}\label{abx}
x_{t}=[x,\,\, b x a+a x b+x b a+b a x], \qquad x,a,b \in {\rm Mat}_m,
\end{equation}
where $\,\,a^2=b^2={\bf 1}_m$.
Equation  \eqref{abx} admits the following skew-symmetric
reduction
$$
x^T=-x, \qquad b=a^T.
$$
Different integrable $\mathfrak{so}_m$-models provided by this reduction
are in one-to-one correspondence with equivalence
classes  of $m\times m$
matrices $a$ such that $a^2={\bf 1}$ with respect to the $SO_m$ gauge action.
For the real matrix $a$, a canonical form for such equivalence
class can be chosen as
$$
a=\left(%
\begin{array}{cc}
  {\bf 1}_p &T \\[2mm]
  0 \, &  {\bf 1}_{m-p} 
\end{array}\right)%
$$
Here ${\bf 1}_s$ stands for the unity $s\times s$-matrix and
$T=\{t_{ij}\},$ where $t_{ij}=\delta_{ij} \alpha_i$.
This canonical form is defined by the discrete natural parameter
$p$ and continuous parameters $\alpha_1,\dots, \alpha_r$, where
$p\le m/2,\,\,$ $r=\min(p,m-p)$.

For example, in the case $m=4$ the equivalence classes with $p=2$
and $p=1$ give rise to the $\mathfrak{so}_4$ Steklov  and   Poincare
integrable models, respectively.
\end{example}

Thus, whereas Example \ref{EEx2} leads to the matrix version of the $\mathfrak{so}_4$
Schottky-Manakov top, the tops corresponding to Example \ref{EEx3}  with $p=[m/2]$ and
$p=1$ can be regarded as generalizations for the $\mathfrak{so}_4$ Steklov
and Poincar\'e models.

For Example \ref{Ex2} the operator $R$ is given by  $R(x)=c\, x$. In the case of Remark \ref{rem39} we have $R(x)=a x b+ b a x. $ 
Both of the $R$-operators are written in the form \eqref{Rmat}. 

Several classes of integrable equations \eqref{difR} can be constructed using results of the next section. 

\subsubsection{Algebraic structure associated with $R$-operator}

Consider an operator $R$ of the form \eqref{Rmat} that defines an associative product \eqref{mult2}. 
Our aim is to study \cite{odsok1,odsok3} 
the structure of the associative algebra ${\cal
M}$ generated by ${\bf a}_i, {\bf b}^i, {\bf c}$ and ${\bf 1}$.  Denote by $\cal L$ the $2 p+2$ dimensional vector space spanned by these matrices.

\begin{lemma} If the product  \eqref{mult2} is associative, then $$ 
{\bf a}_i {\bf a}_j=\sum_k \phi_{i,j}^k {\bf a}_k+\mu_{i,j} {\bf 1},\qquad
{\bf b}^i {\bf b}^j=\sum_k \psi_k^{i,j} {\bf b}^k+\lambda^{i,j} {\bf 1}
$$
for some tensors $\phi^k_{i,j},\, \mu_{i,j},\, \psi^{i,j}_k,\,
\lambda^{i,j}.$
\end{lemma}

This means that the vector spaces spanned by ${\bf 1}, {\bf a}_1,\dots {\bf a}_p$
and ${\bf 1}, {\bf b}^1,\dots {\bf b}^p$ are associative algebras. We denote them by
${\cal A}$ and ${\cal B},$ respectively. These algebras should be in some sense compatible with each other. The simplest example of such compatibility  can be described as follows. 

\begin{example}\label{ex310}
Let ${\cal A}$ and ${\cal B}$ be associative algebras
with basis $A_{1},\dots, A_{p}$ and $B^{1},\dots, B^{p}$ and
structural constants $\phi^i_{j,k}$ and
$\psi^{\alpha,\beta}_{\gamma}.$  Suppose that the structural
constants satisfy the following identities:
$$
 \phi^s_{j,k}\psi_s^{l,i}=\phi^l_{s,k}\psi_j^{s,i}+\phi^i_{j,s}\psi_k^{l,s},
 \qquad 1\le i,j,k,l \le p.
$$
Here and below we assume that the {\it summation is carried out over repeated indices}. 
Then the algebra ${\cal M}$ of dimension $2p+p^{2}$ with the basis
$A_{i}, B^{j},A_{i} B^{j}$ and relations
$$
B^i A_j=\psi_j^{k,i} A_k+\phi^i_{j,k} B^k
$$
is associative. This structure is called an {\it associative bi-algebra} \cite{Agui}.
\end{example}
In general case the compatibility of  ${\cal A}$ and ${\cal B}$ is described by
\begin{lemma} If  \eqref{mult2} is associative, then
$$\phi^s_{j,k}\psi_s^{l,i}=\phi^l_{s,k}\psi_j^{s,i}+\phi^i_{j,s}\psi_k^{l,s}+\delta^l_kt^i_j-
 \delta^i_jt_k^l-\delta^i_j\phi^l_{s,r}\psi^{r,s}_k,$$
 and
\begin{equation}\label{bi1}
{\bf b}^i {\bf a}_j=\psi_j^{k,i} \,{\bf a}_k+\phi^i_{j,k} \,
{\bf b}^k+t^i_j {\bf 1}+\delta^i_j\,{\bf c},
\end{equation}
for some tensor $t^i_j$. 
\end{lemma}
The remaining matrix ${\bf c}$ obeys the following relations:
\begin{lemma} If  \eqref{mult2} is associative, then
\begin{equation}\label{bi2}
{\bf b}^i\,{\bf c}=\lambda^{k,i}{\bf a}_k-t^i_k\,{\bf b}^k-\phi^i_{k,l}\psi^{l,k}_s\,{\bf b}^s-\phi^i_{k,l}\lambda^{l,k}\,{\bf 1},
\end{equation}
\begin{equation}\label{bi3}
{\bf c}\,{\bf a}_j=\mu_{j,k}\,{\bf b}^k-t^k_j\,{\bf a}_k-\phi^s_{k,l}\psi^{l,k}_j\,{\bf a}_s-\mu_{k,l}\psi^{l,k}_j\,{\bf 1},
\end{equation}
where
$$
\phi^s_{j,k}t^i_s=\psi_j^{s,i}\mu_{s,k}+\phi^i_{j,s}t_k^s-\delta^i_j\psi_k^{s,r}\mu_{r,s},
 \qquad
 \psi^{k,i}_st_j^s=\phi^i_{j,s}\lambda^{k,s}+\psi_j^{s,i}t_s^k-\delta^i_j\phi^k_{s,r}\lambda^{r,s}.
 $$
\end{lemma}
Relations \eqref{bi1}-\eqref{bi3} mean that the vector space ${\cal L}$ is a left ${\cal B}$-module and a right  ${\cal A}$-module. 
\subsubsection{Invariant description}
In this section we forget that the generators of ${\cal L}$ are matrices and give a purely algebraic description of the structure appeared  above. 
\begin{definition} By weak ${\cal M}$-structure on a linear space
$\cal L$ we mean a collection of the following data:
\begin{itemize}
\item  Two subspaces $\cal A\subset \cal L$ and $\cal B\subset \cal L$ and a distinguished
element $1\in\cal A\cap\cal B$. 

\item  A non-degenerate symmetric scalar product $(\cdot, \cdot)$ on the
space $\cal L$.  

\item  Two associative products $\cal A\times\cal A\to\cal A$ and
$\cal B\times\cal B\to\cal B$ with a unity $1$. 

\item  A left action $\cal B\times\cal L\to\cal L$ of the algebra
$\cal B$ and a right action $\cal L\times\cal A\to\cal L$ of the
algebra $\cal A$ on the space $\cal L,$ which  commute with each other.
\end{itemize}

This data should satisfy the following properties:
\begin{itemize}

\item[ {\bf 1.}] $\dim{\cal A\cap\cal B=\dim\cal L/(\cal A+\cal B)}= 1$.
The intersection of $\cal A$ and $\cal B$ is a one dimensional space
spanned by the unity $1$.

\item[ {\bf 2.}] The restriction of the action $\cal B\times\cal L\to\cal L$
to the subspace $\cal B\subset \cal L$ is the product in $\cal B$.
The restriction of the action $\cal L\times\cal A\to\cal L$ to
the subspace $\cal A\subset \cal L$ is the product in $\cal A$.

\item[ {\bf 3.}] $(A_1,A_2)=(B^1,B^2)=0$ and $$(B^1 B^2,\,v)=(B^1,\,B^2 v),
\qquad (v,\,A_1 A_2)=(v A_1,\,A_2)$$ for any $A_1,A_2\in\cal A$,
$B^1,B^2\in\cal B$ and $v \in\cal L$.
\end{itemize}
\end{definition}

\begin{remark} It follows from these properties that $(\cdot,\cdot)$ gives a non-degenerate pairing 
between the quotient spaces ${\cal A}/\C_1$ and ${\cal B}/\C_1$, so
$\dim\cal A=\dim\cal B$ and $\dim{\cal L}=2\dim{\cal A}$.
\end{remark}

For a given weak ${\cal M}$-structure  we can define an algebra
generated by $\cal L$ with natural compatibility and universality
conditions.

\begin{definition} By weak ${\cal M}$-algebra associated with a weak
${\cal M}$-structure on $\cal L$ we mean an associative algebra
$U(\cal L)$ with the following properties:

\begin{itemize}

\item[ {\bf 1.}] $\cal L\subset U(\cal L)$ and the actions $\cal
B\times\cal L\to\cal L$,\quad  $\cal L\times\cal A\to\cal L$ are
restrictions of the product in $U(\cal L)$.

\item[ {\bf 2.}] For any algebra $X$ with the property 1 there exists
a unique homomorphism of algebras $X\to U(\cal L)$ that is
identical on $\cal L$.
\end{itemize}
\end{definition}

\subsubsection{Explicit formulas for $U(\cal L)$}

Let
$\{\bar A_1,...,\bar A_p\}$ be a basis of ${\cal A}/\C_1$ and $\{\bar B^1,...,\bar B^p\}$
be a dual basis of ${\cal B}/\C_1$. This means that
$(\bar A_i, \bar B^j)=\delta_i^j$. It is clear, that $1,A_1,...,A_p$ and $1,B^1,...B^p,$ where $A_i\in \bar A_i, \, B_i\in \bar B_i,$ are bases in 
${\cal A}$ and ${\cal B},$ respectively. 
The element $C\in\cal L$ does not belong to the sum of $\cal A$ and $\cal
B$. Since $(\cdot , \cdot)$ is non-degenerate, we have $(1,C)\ne
0$. Without loss of generality we may assume that $(1,C)=1$,
 $(C,C)=(C,A_i)=(C,B^j)=0.$ Given basis of $\cal A$ and $\cal B$, such an element $C$ is uniquely determined.
 
\begin{proposition} The algebra $U(\cal L)$ is defined by the
following relations
$$
A_iA_j=\phi_{i,j}^kA_k+\mu_{i,j}\,1, \qquad
B^iB^j=\psi_k^{i,j}B^k+\lambda^{i,j}\,1
$$$$
B^iA_j=\psi_j^{k,i}A_k+\phi^i_{j,k}B^k+t^i_j\,1 +\delta^i_jC,
$$$$
B^iC=\lambda^{k,i}A_k+u_k^iB^k+p^i\,1, \qquad
CA_j=\mu_{j,k}B^k+u^k_jA_k+q_i\,1
$$
for certain tensors
$\phi_{i,j}^k,\psi_k^{i,j},\mu_{i,j},\lambda^{i,j},u_k^i,p^i,q_i$.
\end{proposition}

Let us define an element $K\in U(\cal L)$ by the formula $$K=A_i
B^i+C.$$
\begin{definition} A weak ${\cal M}$-structure on $\cal L$ is called
${\cal M}$-structure if $K\in U(\cal L)$ is a central element of
the algebra $U(\cal L)$.
\end{definition}

\begin{theorem} {\rm (cf. Example \ref{ex310})}. For any ${\cal M}$-structure the algebra $U(\cal L)$ is spanned by the
elements $K^s,\,A_i K^s,\,$ $ B_j K^s,\, A_i B^j K^s,$ where
$i,j=1,...,p,\,$ and $\, s=0,1,2,...$
\end{theorem}
\begin{theorem} For any representation ${U(\cal L)} \to {\rm Mat}_m$
given by 
$$A_1\to {\bf a}_1,...,  A_p\to {\bf a}_p, \quad  B^1\to {\bf b}^1,...,B^p\to
{\bf b}^p, \quad C\to {\bf c}$$   
the formula
$$
X \circ Y =R(X) Y+X R(Y)-R(X Y),
$$
where
$$
R(x)={\bf a}_1 \,x \,{\bf b}^1+...+{\bf a}_p \,x\, {\bf b}^p+{\bf c}\, x,
$$
defines an associative product on
${\rm Mat}_m$ compatible with the usual matrix product.
\end{theorem}
\begin{example}\label{SerieA}
Suppose that $\cal A$ and $\cal B$ are generated by
 elements $A\in \cal A$ and  $B\in
\cal B$ such that $A^{p+1}=B^{p+1}=1$.  Assume that
$(B^i,A^{-i})=\epsilon^i-1$, $(1,C)=1$ and other scalar products
are equal to zero. Here $\epsilon$ is a primitive root of unity of
order $p$. Let
$$B^iA^j=\frac{\epsilon^{-j}-1}{\epsilon^{-i-j}-1}A^{i+j}+\frac{\epsilon^i-1}{\epsilon^{i+j}-1}B^{i+j}$$
for $i+j\ne0$ modulo $p$ and
$$B^iA^{-i}=1+(\epsilon^i-1)C,\qquad
CA^i=\frac{1}{1-\epsilon^i}A^i+\frac{1}{\epsilon^i-1}B^i, \qquad
B^iC=\frac{1}{\epsilon^{-i}-1}A^i+\frac{1}{1-\epsilon^{-i}}B^i$$
for $i\ne0$ modulo $p$. These formulas define an ${\cal
M}$-structure.
 
The central element has the following form
$$K=C+\sum_{0<i<p}\frac{1}{\epsilon^i-1}A^{-i}B^i.$$

Let $a$, $t$ be linear operators on some vector space. Assume that
$a^{p+1}=1$, $a t=\epsilon \,t a$ and that the operator $t-1$ is
invertible. It is easy to check that the formulas $$A\to a,\qquad
B\to \frac{\epsilon t-1}{t-1}a,\qquad C\to\frac{t}{t-1}$$ define a
representation of the algebra $U(\cal L)$.
\end{example}

\subsubsection{Case of semi-simple ${\cal A}$ and ${\cal B}$}

\begin{proposition}\label{prdim} Suppose that for a weak ${\cal M}$-structure the algebra ${\cal A}$ is semi-simple:
$$
{\cal A} = \oplus_{1\le i \le r} \, {\rm End}(V_i),  \qquad {\rm dim}\, V_i = m_i.
$$
Then ${\cal L}$ as a right ${\cal A}$-module is isomorphic to $ \oplus_{1\le i \le r} \, (V^{*}_
{i})^{2m_i}$ .
\end{proposition}
\begin{proof}
Since any right ${\cal A}$-module has the form $ \oplus_{1\le i \le r} \, (V^{*}_
{i})^{l_i}$ for some $l_1,\dots, \quad l_r \ge 0$, we have
${\cal L} = \oplus_{1\le i \le r} \,{\cal L}_i$, where ${\cal L}_i= (V^*_i)^{l_i}$. Note that ${\cal A} \subset {\cal L}$ and, moreover, $ {\rm End}(V_i) \subset {\cal L}_i$ for $i = 1,\dots, r$.
Besides, ${\rm End}(V_i)\perp {\cal L}_j$ for $i \ne j$. Indeed, we have $(v, a) = (v, \,Id_i\,a) = (v\,Id_i,\, a) = 0$ for $v \in {\cal L}_j$ and
$a \in {\rm End}(V_i)$, where $Id_i$ is the unity of the subalgebra ${\rm End}(V_i)$. Since $(\cdot, \cdot)$ is non-degenerate and
${\rm End}(V_i) \perp {\rm End}(V_i)$ by Property 3 of weak ${\cal M}$-structure, we have ${\rm dim}\,{\cal L}_i \ge 2 {\rm dim}\,{\rm End}(V_i)$.
But $\sum_i {\rm dim}\,{\cal L}_i = {\rm dim} \,{\cal L} = 2\,{\rm dim}\,{\cal A} = \sum_i 2 {\rm dim}\,{\rm End}(V_i)$ and we obtain ${\rm dim}\,{\cal L}_i = 2 {\rm dim}\,{\rm End}(V_i)$
for each $i = 1,\dots, r$ which is equivalent to the statement of the proposition. 
\end{proof}

\begin{theorem}\label{aij}
Suppose that a vector space $\cal L$ is equipped with a weak ${\cal
M}$-structure such that the associative algebras ${\cal A}$ and ${\cal B}$  are semi-simple: 
\begin{equation} \label{AABB}{\cal A}=\oplus_{1\le i\le
r}{\rm End}(V_i),\qquad {\cal B}=\oplus_{1\le j\le s}{\rm End}(W_j), \end{equation}
$$ {\rm dim}\,
V_i=m_i, \qquad \, {\rm dim}\, W_j=n_j.$$ 
Then ${\cal L}$ as an  ${\cal A} \otimes {\cal B}$-module is given by the formula
\begin{equation}\label{stL}
{\cal L} = \oplus_{1\le i\le r,1\le
j\le s}(V^{\star}_i\otimes W_j)^{a_{i,j}} 
\end{equation}
for some integers $a_{i,j}\ge
0,$
and
\begin{equation}\label{adm}   \sum_{j=1}^s
a_{i,j}\,n_j=2\,m_i,\qquad \sum_{i=1}^r a_{i,j}\,m_i=2\,n_j.
\end{equation}
\end{theorem}
\begin{proof}
It is known that any ${\cal A} \otimes {\cal B}$-module has the form  \eqref{stL}. Applying Proposition \ref{prdim}, we obtain ${\rm dim}\,{\cal L}_i = 2\,m^2_i,$
where ${\cal L}_i  = \oplus_{1\le j \le s} (V^{*}_i \otimes W_j)^{a_{i,j}}$ . This
gives the first equation from \eqref{adm}. The second equation can be obtained similarly.
\end{proof}
\begin{remark}
Since the dimensions  of ${\cal A}$ and ${\cal B}$ coincide we have
$$
\sum_{i=1}^r m_i^2=\sum_{i=1}^s n_i^2.
$$
\end{remark}

\begin{definition} The $r\times s$-matrix $A=\{a_{i,j} \}$ from Theorem \ref{aij} is called the {\it matrix of multiplicities}
of the weak ${\cal M}$-structure.
\end{definition}

\begin{definition} The $r \times s$-matrix $A$ is called {\it decomposable} if there exist partitions
$\{1,\dots, r\} = I\cup I'$  and $\{1,\dots, s\} = J \cup J'$  such that $a_{i, j} = 0$ for $(i, j ) \in I\times J'$ or for $(i, j ) \in I'\times J$ .
\end{definition}
\begin{lemma} The matrix of multiplicities $A$ is indecomposable.
\end{lemma}
Consider  \eqref{adm} as a system of linear equations for the vector $(m_1,\dots , m_r, n_1,\dots, n_s)$.
The matrix of the system  has the form
$Q = \begin{pmatrix} 2&-A\\-A^t&2\end{pmatrix}$.
According to the result by E. Vinberg \cite{Vinb}, if the kernel of
an indecomposable matrix $Q$ contains an integer positive vector,
them $Q$ is the Cartan matrix of an affine Dynkin diagram.
Moreover, it follows from the structure of $Q$ that this is
a simple laced affine Dynkin diagram with a partition of the
set of vertices into two subsets (white and black vertices on the pictures below) such that vertices of the same
subset are not connected.

\begin{theorem} Let $A$ be an $r\times s$-matrix of multiplicities
for an indecomposable weak \linebreak ${\cal M}$-structure. Then, after a
permutation of rows and columns and up to transposition, the matrix
$A$ is equal to one from the following list:

{\bf 1.} $A=(2)$. Here $r=s=1,$ $n_1=m_1=m$. The corresponding
Dynkin diagram is of the type $\tilde A_1.$

{\bf 2.} $a_{i,i}=a_{i,i+1}=1$ and $a_{i,j}=0$ for other pairs
$i,j$. Here $r=s=k\ge 2,$ the indices are taken modulo $k$,  and
$n_i=m_i=m$. The corresponding Dynkin diagram is $\tilde A_{2
k-1}.$

\bigskip
\bigskip

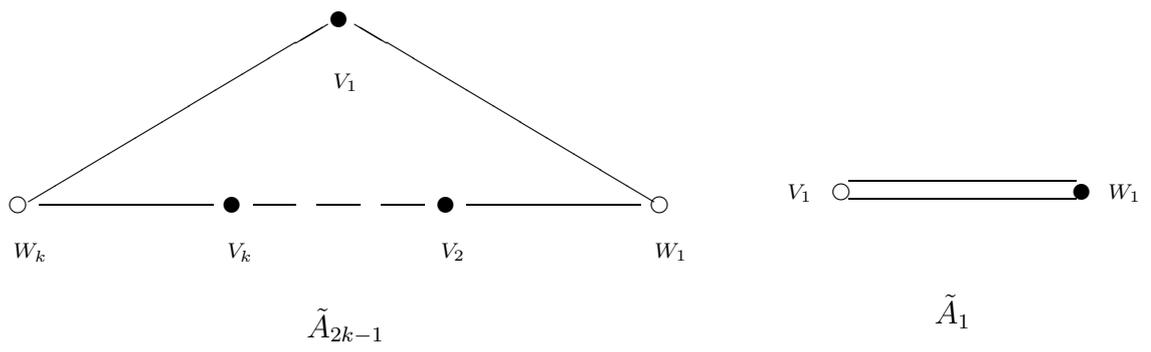
\begin{picture}(450,100)

\put(242,30){\circle{6.0}}

\put(162,30){\circle*{6.2}}

\put(82,30){\circle*{6.2}}

\put(2,30){\circle{6.0}}

\put(240.3,10){{\scriptsize $W_1$}}

\put(160.3,10){{\scriptsize $V_2$}}

\put(80.3,10){{\scriptsize $V_{k}$}}

\put(0.3,10){{\scriptsize $W_{k}$}}

\put(90,30){\line(1,0){16}} \put(114,30){\line(1,0){16}}
\put(138,30){\line(1,0){16}}

\put(10,30){\line(1,0){65}}

\put(170,30){\line(1,0){65}}



\put(6,31){\line(5,3){112}}

\put(240,31){\line(-5,3){112}}

\put(122,100){\circle*{6.0}}

\put(120,74){{\scriptsize $V_1$}}

\put(110,-20){{\bf $\tilde A_{2 k-1}$}}


\put(310,35){\circle{6.0}}

\put(400,35){\circle*{6.0}}


\put(313,39){\line(1,0){85}}

\put(313,32){\line(1,0){85}}

\put(290,32){{\scriptsize $V_1$}}

\put(410,32){{\scriptsize $W_1$}}

\put(345,-15){{\bf $\tilde A_{1}$}}

\end{picture}

\vspace{1cm}

{\bf 3.}   $A=$\begin{math} \bordermatrix{&\cr &1&1&0&0\cr
&1&0&1&0\cr &1&0&0&1}\end{math}. Here $r=3, \, s=4$ and $n_1=3
m,\,\, n_2=n_3=n_4=m, \,$ $\,m_1=m_2=m_3=2 m$.
The Dynkin diagram is $\tilde E_{6}:$ 
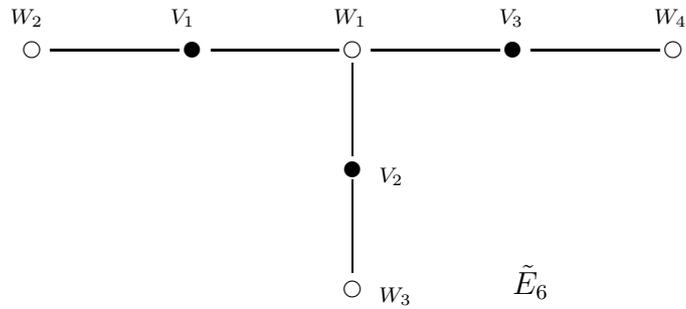
\vspace{1cm}

\begin{picture}(650,100)
\put(310,100){\circle{6.0}} \put(250,100){\circle*{6.0}}
\put(190,100){\circle{6.0}} \put(130,100){\circle*{6.0}}
\put(70,100){\circle{6.0}} \put(190,55){\circle*{6.0}}
\put(190,10){\circle{6.0}}

\put(303,110){{\scriptsize $W_4$}}

\put(257,100){\line(1,0){48}} \put(197,100){\line(1,0){48}}
\put(137,100){\line(1,0){48}} \put(77,100){\line(1,0){48}}
\put(190,16){\line(0,1){35}} \put(190,60){\line(0,1){35}}

\put(245,110){{\scriptsize $V_3$}}

\put(183,110){{\scriptsize$W_1$}}

\put(122,110){{\scriptsize $V_1$}}

\put(62,110){{\scriptsize $W_2$}}

\put(200,50){{\scriptsize $V_2$}}

\put(200,5){{\scriptsize $W_3$}}

\put(250,8){{\bf $\tilde E_{6}$}}

\end{picture}

\vspace{1cm}
{\bf 4.} $A=$\begin{math} \bordermatrix{&\cr &1&1&0&0&0\cr
&0&1&1&1&0\cr &0&0&0&1&1}\end{math}. \quad
Here $r=3$, $s=5$ \quad and \quad
$n_1=m,\quad n_2=3m, \quad n_3=2m,\quad n_4=3m,$  $\quad
n_5=m$, \qquad $m_1=2m,\quad m_2=4m, \quad m_3=2m$.
The Dynkin diagram is $\tilde E_{7}:$

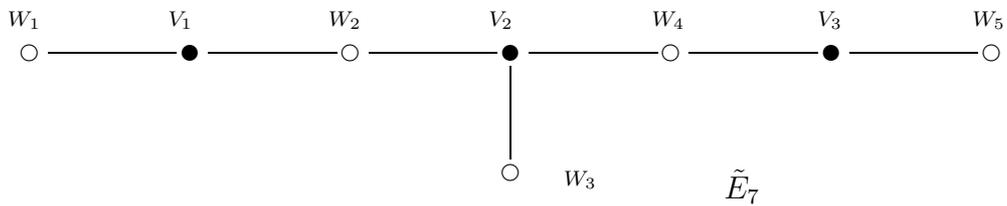
\vspace{1cm}

\begin{picture}(650,80)
\put(370,70){\circle{6.0}} \put(310,70){\circle*{6.0}}
\put(250,70){\circle{6.0}} \put(190,70){\circle*{6.0}}
\put(70,70){\circle*{6.0}} \put(130,70){\circle{6.0}}
\put(10,70){\circle{6.0}} \put(190,25){\circle{6.0}}

\put(363,80){{\scriptsize $W_5$}}

\put(317,70){\line(1,0){48}} \put(257,70){\line(1,0){48}}
\put(197,70){\line(1,0){48}} \put(137,70){\line(1,0){48}}
\put(77,70){\line(1,0){48}} \put(17,70){\line(1,0){48}}
 \put(190,30){\line(0,1){35}}

\put(305,80){{\scriptsize $V_3$}}

\put(243,80){{\scriptsize$W_4$}}

\put(182,80){{\scriptsize $V_2$}}

\put(122,80){{\scriptsize $W_2$}}

\put(62,80){{\scriptsize $V_1$}}

\put(2,80){{\scriptsize $W_1$}}

\put(210,20){{\scriptsize $W_3$}}

\put(270,15){{\bf $\tilde E_{7}$}}

\end{picture}

{\bf 5.} $A=$\begin{math} \bordermatrix{&\cr &1&0&0&0&0\cr
&1&1&1&0&0\cr &0&0&1&1&0\cr &0&0&0&1&1}\end{math}. \qquad 
Here $r=4,s=5$ \quad and \qquad $n_1=4m,\quad n_2=3m,\quad
n_3=5m,\quad
 n_4=3m, $ $\quad n_5=m,\qquad m_1=2m,\quad m_2=6m,\quad
m_3=4m,$ \quad $ m_4=2m$.
The Dynkin diagram is $\tilde E_{8}:$

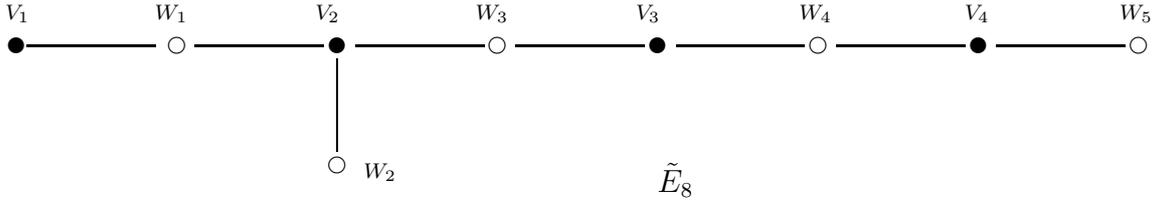
\vspace{1cm}
\begin{picture}(650,80)
\put(400,70){\circle{6.0}} \put(340,70){\circle*{6.0}}
\put(280,70){\circle{6.0}} \put(220,70){\circle*{6.0}}
\put(100,70){\circle*{6.0}} \put(160,70){\circle{6.0}}
\put(40,70){\circle{6.0}} \put(100,25){\circle{6.0}}
\put(-20,70){\circle*{6.0}}

\put(393,80){{\scriptsize $W_5$}}

\put(347,70){\line(1,0){48}} \put(287,70){\line(1,0){48}}
\put(227,70){\line(1,0){48}} \put(167,70){\line(1,0){48}}
\put(107,70){\line(1,0){48}} \put(47,70){\line(1,0){48}}
\put(-16,70){\line(1,0){48}}
 \put(100,30){\line(0,1){35}}

\put(335,80){{\scriptsize $V_4$}}

\put(273,80){{\scriptsize$W_4$}}

\put(212,80){{\scriptsize $V_3$}}

\put(152,80){{\scriptsize $W_3$}}

\put(92,80){{\scriptsize $V_2$}}

\put(32,80){{\scriptsize $W_1$}}

\put(110,20){{\scriptsize $W_2$}}

\put(-24,80){{\scriptsize $V_1$}}

\put(220,15){{\bf $\tilde E_{8}$}}

\end{picture}

{\bf 6.} $A=(1,1,1,1)$. Here $r=1,s=4$ and $n_1=n_2=n_3=n_4=m,\,\,
m_1=2 m$. The corresponding Dynkin diagram is $\tilde D_4.$

{\bf 7.} $a_{1,1}=a_{1,2}=a_{1,3}=1,\,\,$
$a_{2,3}=a_{2,4}=a_{3,4}=a_{3,5}=\cdots=a_{k-2,k-1}=a_{k-2,k}=1,\,\,$
$a_{k-1,k}=a_{k-1,k+1}=a_{k-1,k+2}=1,\,\,$
 and $a_{i,j}=0$ for other $(i,j)$.
 
Here we have $r=k-1,$ $s=k+2$ and $n_1=n_2=n_{k+1}=n_{k+2}=m,\,\,
n_3=\cdots=n_{k}=2 m,\,\, m_1=\cdots=m_{l}=2 m$. The corresponding
Dynkin diagram is $\tilde D_{2 k}$,\, where $k \ge 3.$

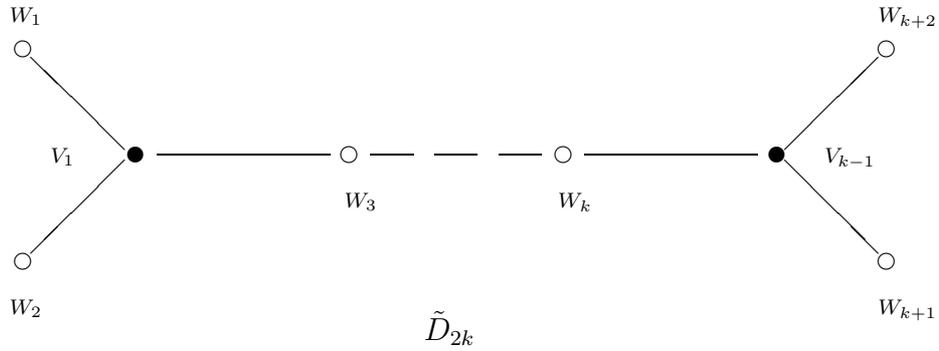
\begin{picture}(450,100)

\put(302,30){\circle*{6.0}}

\put(222,30){\circle{6.2}}

\put(142,30){\circle{6.2}}

\put(62,30){\circle*{6.0}}

\put(320,27){{\scriptsize $V_{k-1}$}}

\put(220.3,10){{\scriptsize $W_{k}$}}
\
\put(140.3,10){{\scriptsize $W_{3}$}}

\put(30.3,27){{\scriptsize $V_{1}$}}

\put(150,30){\line(1,0){16}} \put(174,30){\line(1,0){16}}
\put(198,30){\line(1,0){16}}

\put(70,30){\line(1,0){65}}

\put(230,30){\line(1,0){65}}



\put(58,28){\line(-1,-1){35}}

\put(58,32){\line(-1,1){35}}

\put(20,70){\circle{6.0}} \put(20,-10){\circle{6.0}}

\put(15,80){{\scriptsize $W_1$}} \put(15,-30){{\scriptsize $W_2$}}

\put(305,28){\line(1,-1){35}}

\put(305,32){\line(1,1){35}}

\put(343,70){\circle{6.0}} \put(343,-10){\circle{6.0}}
\put(340,80){{\scriptsize $W_{k+2}$}} \put(340,-30){{\scriptsize
$W_{k+1}$}}

\put(170,-40){{\bf $\tilde D_{2 k}$}}

\end{picture}
\vspace{2cm}

{\bf 8.} $a_{1,1}=a_{1,2}=a_{1,3}=1,\,\,$
$a_{2,3}=a_{2,4}=a_{3,4}=a_{3,5}=\cdots=a_{k-2,k-1}=a_{k-2,k}=1,\,\,$
$a_{k-1,k}=a_{k,k}=1,\,\,$
 and $a_{i,j}=0$ for other $(i,j)$.

Here we have $r=s=k\ge 3,\,\,$ $n_1=n_2=m,\,$ $ n_3=\cdots=n_{k}=2
m, \,\,m_1=\cdots=m_{k-2}=2 m, \,\,m_{k-1}=m_k=m$. The
corresponding Dynkin diagram is $\tilde D_{2 k-1}:$

\begin{picture}(450,100)

\put(302,30){\circle{6.0}}

\put(222,30){\circle*{6.2}}

\put(142,30){\circle{6.2}}

\put(62,30){\circle*{6.0}}

\put(320,27){{\scriptsize $W_k$}}

\put(220.3,10){{\scriptsize $V_{k-2}$}}

\put(140.3,10){{\scriptsize $W_{3}$}}

\put(30.3,27){{\scriptsize $V_{1}$}}

\put(150,30){\line(1,0){16}} \put(174,30){\line(1,0){16}}
\put(198,30){\line(1,0){16}}

\put(70,30){\line(1,0){65}}

\put(230,30){\line(1,0){65}}



\put(58,28){\line(-1,-1){35}}

\put(58,32){\line(-1,1){35}}

\put(20,70){\circle{6.0}} \put(20,-10){\circle{6.0}}

\put(15,80){{\scriptsize $W_1$}} \put(15,-30){{\scriptsize $W_2$}}

\put(305,28){\line(1,-1){35}}

\put(305,32){\line(1,1){35}}

\put(343,70){\circle*{6.0}} \put(343,-10){\circle*{6.0}}
\put(340,80){{\scriptsize $V_k$}} \put(340,-30){{\scriptsize
$V_{k-1}$}}

\put(170,-40){{\bf $\tilde D_{2 k-1}$}}

\end{picture}

\vspace{2cm}

 Note that if $k=3$, then
$a_{1,1}=a_{1,2}=a_{1,3}=1,\,\,$ $a_{2,3}=a_{3,3}=1$.

\end{theorem}

\subsubsection{Resume}

Suppose that $\cal L$ is an indecomposable ${\cal
M}$-structure with semi-simple algebras  \eqref{AABB}; then there exists an affine Dynkin diagram of the
type $A$, $D$, or $E$ such that:

{\bf 1.} There is a one-to-one correspondence between the set of
vertices and the set of vector spaces $\{V_1,...,V_r$,
$W_1,...,W_s\}$.

{\bf 2.} For any $i,j$ the spaces $V_i$, $V_j$ are not connected
by edges as well as $W_i$, $W_j$.

{\bf 3.} The vector $$(\dim V_1,...,\dim V_r,\dim W_1,...,\dim
W_s)$$ is equal to $m J,$ where $J$ is the minimal imaginary
positive root of the Dynkin diagram.
 
\begin{remark} It can be proved that for indecomposable ${\cal M}$-structures $m=1$.
\end{remark}

Given  an affine Dynkin diagram of the
$A$, $D$, or $E$-type,  to define the corresponding  ${\cal M}$-structure it remains to construct an
embedding $\cal A\to\cal L$, $\cal B\to\cal L$ and a scalar
product $(\cdot,\cdot)$ on the vector space $\cal L$.

If we fix an element $1\in\cal L$, then we can define the
embedding $\cal A\to\cal L$, $\cal B\to\cal L$ by the formula
$a\to 1a$, $b\to b1$ for $a\in\cal A$, $b\in\cal B$. We may assume
that $1$ is a generic element of $\cal L$.

Thus to study ${\cal M}$-structures corresponding to a Dynkin
diagram, one should take a generic element in ${\cal
L}=\oplus_{1\le i\le r,1\le j\le s}(V^{\star}_i\otimes
W_j)^{a_{i,j}},$ find its simplest canonical form by choosing
basis in the vector spaces $V_1,...,V_r,W_1,...,W_s$, calculate
the embedding $\cal A\to\cal L$, $\cal B\to\cal L$ and the scalar
product $(\cdot,\cdot)$ on the vector space $\cal L$.

The classification of generic elements $1\in\cal L$ up to choice
of basis in the vector spaces $V_1,...,V_r,W_1,...,W_s$ is
equivalent to the classification of irreducible representations of the
quivers corresponding to our affine Dynkin diagrams  and we can
apply known results about these representations.

\section{Non-abelian Hamiltonian formalism and trace Poisson brackets}

 A Poisson structure on a commutative algebra $A$ is a Lie algebra structure on $A$ given
by a Lie bracket
$$
\{\cdot,\,\cdot\}:A\times A \mapsto A,
$$
which satisfies the Leibniz rule
$$\{a ,\, b\,c\} = \{a,\, b\}\,c+ b\,\{a,\, c\},\qquad  a,b,c\in A,$$
with the right (and, hence, also with the left) argument.

It is well-known  that a naive translation of this definition
to the case of a non-commutative associative algebra $A$ is not very interesting because of lack
of examples different from the usual commutator (for prime rings it was shown in \cite{FL}).

\subsection{Non-abelian Poisson brackets on free associative algebras}

Here we consider a version of the Hamiltonian formalism for free associative algebra proposed in \cite{miksokcmp}.

Let ${\cal A}$ be free associative algebra $\C[x_1,\ldots,x_N\,]$ with the product $\circ$. 
For any $a \in {\cal A}$ we denote by $L_a$ (resp. $R_a$) the operators of left (resp. right) multiplication by $a$:
$$
L_a(X)=a\,X, \qquad R_a(X)=X\,a, \qquad X\in {\cal A}.
$$
The associativity of ${\cal A}$ is equivalent to the identity $[L_a, R_b]=0$
for any $a$ and $b$. Moreover,
$$
L_{ab}=L_{a}\,L_{b}, \quad R_{ab}=R_{b}\,R_{a},
\quad L_{a+b}=L_{a}+L_{b},   \quad R_{a+b}=R_{a}+R_{b}.
$$
\begin{definition}\label{localr}
We denote by ${\cal O}$ the associative algebra generated by
all operators of left and right multiplication by elements $x_i$.
This algebra is called the {\it algebra of local operators}.
\end{definition}

To introduce the concept of first integrals, we need an analog of
trace, which is not yet defined in the algebra ${\cal A}$. As a
matter of fact, in our calculations we use only two properties of
the trace, namely linearity and the possibility to perform cyclic
permutations in monomials. Let us define an equivalence relation
for elements of ${\cal A}$ in a standard way.

\begin{definition}\label{def311} Two elements $f_1$ and $f_2$ of ${\cal A}$ are said to be equivalent
and denote $f_1 \sim f_2$ iff $f_1$ can be obtained from $f_2$ by
cyclic permutations of generators in its monomials. We denote by ${\rm tr}\,f$ the equivalence class of the element $f$. 
\end{definition}

We are going to define a class of Poisson brackets on the functionals ${\rm tr}\,f$ (see Definition \ref{def311}). It is easy to see that  the vector space of such functionals can be identified with the quotient vector space ${\cal T}={\cal A}/[{\cal A},\,{\cal A}]$ and  
 the Poisson brackets have to be defined on  ${\cal T}$ \cite{CB}.

Let $a({\bf x})\in {\cal A}$, where ${\bf x}=(x_{1},\dots,x_{N})$ and $\delta {\bf x}=(\delta x_1,\dots,\delta x_N)$, $\delta x_i\in {\cal A}$. Then
$ \mbox{grad}_{\bf x}(a)\in {\cal A}^N $ is a vector 
$$
{\rm grad}_{\bf x}\,(a)=\Big({\rm grad}_{x_1}(a),\dots , {\rm grad}_{x_N}(a)\Big)
$$
uniquely defined by the formula
\[ \frac{d}{d\epsilon}a({\bf x}+\epsilon \, \delta \! {\bf x})|_{\epsilon
=0}\sim
\,\langle\delta \! {\bf x},\,\,  \mbox{grad}_{\bf x} \Big(\,a({\bf x})\Big)\rangle\, ,\]
where $\langle (p_{1},\dots,p_{N}), (q_{1},\dots,q_{N})\rangle = p_{1}\circ
q_{1}+\cdots+p_{N}\circ q_{N}.$
\begin{lemma}
 If $f \sim g,$ then ${\rm grad}_{\bf x}(f)={\rm grad}_{\bf x}(g).$
\end{lemma}
It follows from the lemma that the map 
$ \mbox{grad}_{\bf x}: {\cal T} \to {\cal A}^N  $  is well-defined.

The  Poisson brackets on ${\cal T}$ are defined by the formula
\begin{equation}\label{nabr}  \{f,\,g\}=\langle\mbox{grad}_{\bf x}\, f, \,\,  \Theta ( \mbox{grad}_{\bf x}\, g)\rangle  ,\qquad
f,g\in {\cal T}, 
\end{equation}
where
\begin{equation}\label{skewsym}
\{f,g\}+\{g,f\}\sim 0, \qquad \{f,\{g,h\}\}+\{g,\{h,f\}\}+\{h,\{f,g\}\}\sim 0
\end{equation}
 for any elements $f,g,h\in {\cal T}$.  
Here a skew-symmetric Hamiltonian operator $\Theta$ is supposed to be an element of  ${\cal O}\otimes \mathfrak{gl}_{N}.$  
\begin{remark}
Actually, the right hand side of \eqref{nabr} is a well-defined element of ${\cal A}$ and we take its equivalence class for the left hand side.  The left hand sides of \eqref{skewsym} are regarded as elements of  ${\cal A}.$
\end{remark}
It is easy to show that
 \eqref{skewsym} is equivalent to
$$
\Theta ^{\star}=-\Theta
$$
i.e. $\Theta $ is a skew--symmetric operator with respect to the involution defined by
$$
L^\star _a =R_{a},\qquad  R^\star _a =L_{a}.
$$

\begin{definition} Brackets \eqref{nabr} are called {\it non-abelian Poisson brackets}.
\end{definition}

Hamiltonian system of equations on ${\cal A}$ has the form
\begin{equation}\label{hform}
\frac{d\, {\bf x}}{d\, t}=\Theta \Big(\mbox{grad}_{\bf x}  H\Big),
\end{equation}
where $H({\bf x}) \in {\cal A}/[{\cal A},\,{\cal A}] $ is a Hamiltonian and $\Theta$ is a
Hamiltonian operator. The ODE system \eqref{hform} has the form
\begin{equation}\label{geneq}
\frac{d x_{\alpha}}{d t}=F_{\alpha}({\bf x}), \qquad {\bf x}=(x_1,...,x_N),
\end{equation}
where $F_{\alpha}$ are (non-commutative) polynomials with scalar  constant coefficients. Formula \eqref{geneq} does not mean that the generators $x_i$ of the algebra $\cal A$ 
evolve in $t$. This formula defines a derivation $D_F$ of ${\cal A}$ such that $D_F(x_i)=F_i.$ However, if we replace the generators $x_i$ by $m\times m$ matrices ${\bf x}_i$, then \eqref{geneq} becomes a usual system of ODEs for the entries of the matrices ${\bf x}_i$.

\subsubsection{Non-abelian Hamiltonian operators}

Consider linear Hamiltonian operators. It means that the entries of the
Hamiltonian operator $\Theta$ are given by
\begin{equation}\label{thetaij}
\Theta_{ij}=b_{ij}^{k} \, R_{x_{k}}+\bar b_{ij}^{k}\, L_{x_{k}}.  
\end{equation}
The skew-symmetry of $\Theta$ implies
\begin{equation}\label{theta1}
\bar b_{ij}^{k}=- b_{ji}^{k}.
\end{equation}

\begin{proposition} An operator $\Theta$ given by  \eqref{thetaij},
 \eqref{theta1} is Hamiltonian iff $b_{ij}^{k}$ are the structural constant of an associative algebra.
\end{proposition}
\begin{corollary} Any pair of compatible associative algebras {\rm (}see Subsection 3.2.2 {\rm )} generates a pair of compatible linear non-abelian Poisson brackets.
\end{corollary}

\begin{example}  Let $N=2$. Consider the following compatible associative products:
$$
x_{1}\star x_{1}=x_{1}, \quad x_{1}\star x_{2}=x_{2}\star x_{1}=x_{2}\star
x_{2}=0
$$
and
$$
x_{1}\circ x_{1}=x_{2}, \quad x_{1}\circ x_{2}=x_{2}\circ x_{1}=x_{2}\circ
x_{2}=0.
$$
The corresponding Poisson brackets $\{\cdot,\cdot\}_i$ have the Hamiltonian operators 
$$
\Theta_i=\begin{pmatrix} R_{x_{i}}-L_{x_{i}} & 0 \\ 0 & 0\end{pmatrix}, \qquad  i=1,2.
$$
The pencil  $\{\cdot,\cdot\}_1 + \lambda \{\cdot,\cdot\}_2$ has a Casimir function $C={\rm tr}\,(x_1+\lambda x_2)^3$, which produces the 
Hamiltonians $H_1 = -{\rm tr}\, (x_1^2 x_2)$ and  $\ds H_2=\frac{1}{3}{\rm tr}\, (x_1^3)$ commuting with respect to both brackets (see Theorem \ref{biH}). The formula
$$\frac{d\, {\bf x}}{d\, t}=\Theta_1 \Big(\mbox{grad}_{\bf x}  H_1\Big) = \Theta_2 \Big(\mbox{grad}_{\bf x}  H_2\Big)
$$
gives us a bi-Hamiltonian representation for the system
$$
\frac{d\, x_1}{d\, t} = x_1^2 x_2-x_2 x_1^2, \qquad \frac{d\, x_2}{d\, t} = 0
$$
already mentioned in Examples \ref{manak},\,\,\ref{EEx2} and in Application 3.1.
\end{example}

For  quadratic Poisson brackets the general form of the Hamiltonian operator is given by
\begin{equation}\label{tet}
\Theta_{i,j}=a^{pq}_{ij} L_{x_p} L_{x_q} - a^{qp}_{ji} R_{x_p}
R_{x_q} + r^{pq}_{ij} L_{x_p} R_{x_q},
\end{equation}
where $a^{pq}_{ij}$ and $r^{pq}_{ij}$ are some (complex) constants,
$r^{pq}_{ij}=-r^{qp}_{ji}$,
$p,q,i,j=1,...,N,$ and the summation  over repeated indices is assumed.

\begin{proposition} Formula  \eqref{tet} define a Poisson bracket iff the following relations hold:
\begin{equation}\label{r1}
r^{\sigma \epsilon}_{\alpha\beta}=-r^{\epsilon\sigma}_{\beta\alpha},
\end{equation}
\begin{equation}\label{r2}
r^{\lambda\sigma}_{\alpha\beta}
r^{\mu\nu}_{\sigma\tau}+r^{\mu\sigma}_{\beta\tau} r^{\nu\lambda}_{\sigma\alpha}+r^{\nu\sigma}_{\tau\alpha} r^{\lambda\mu}_{\sigma\beta}=0,
\end{equation}
\begin{equation}\label{r3}
a^{\sigma\lambda}_{\alpha\beta} a^{\mu\nu}_{\tau\sigma}=a^{\mu\sigma}_{\tau\alpha} a^{\nu\lambda}_{\sigma\beta},
\end{equation}
\begin{equation}\label{r4}
a^{\sigma\lambda}_{\alpha\beta} a^{\mu\nu}_{\sigma\tau}=a^{\mu\sigma}_{\alpha\beta} r^{\lambda\nu}_{\tau\sigma}+a^{\mu\nu}_{\alpha\sigma}
r^{\sigma\lambda}_{\beta\tau}.
\end{equation}
and
\begin{equation}\label{r5}
a^{\lambda\sigma}_{\alpha\beta} a^{\mu\nu}_{\tau\sigma}=a^{\sigma\nu}_{\alpha\beta} r^{\lambda\mu}_{\sigma\tau}+a^{\mu\nu}_{\sigma\beta}
r^{\sigma\lambda}_{\tau\alpha}.
\end{equation}
\end{proposition}
\begin{remark}
Conditions \eqref{r1} and \eqref{r2} mean that the tensor ${\bf r}$ satisfies the associative Yang-Baxter {\rm(}or Rota-Baxter{\rm)} equation {\rm \cite{Rota,Agui}}. 
\end{remark}

\subsection{Trace Poisson brackets}

The non-abelian brackets are defined between traces only. But if $x_1,...,x_N$ are $m\times m$-matrices, we can extend these brackets to the matrix entries.

We consider $N m^2$-dimensional Poisson brackets defined on functions in entries $x^j_{i,\alpha}$ of 
$m\times m$-matrices $x_{\alpha}$. 
Here and in the sequel, we use Latin indices ranging from 1 to $m$ for the matrix entries and Greek
indices ranging from 1 to $N$ to label the matrices.

\begin{definition} Such a bracket is called {\it trace Poisson bracket} iff
\begin{itemize}
\item the bracket is ${\rm GL}(m)$-invariant;
\item for  any two matrix polynomials $P_i(x_1,...,x_N), \quad i=1,2\,\,$ with coefficients in $\C$ the bracket between its traces  is the trace of some matrix 
polynomial $P_3$.
\end{itemize}
\end{definition}
 
\begin{theorem}\label{th51} Any constant trace Poisson bracket has the form 
\begin{equation}\label{Poissoncon}
\{x^j_{i,\alpha},x^{j^{\prime}}_{i^{\prime},\beta}\}=
 \delta^j_{i^{\prime}} \delta^{j^{\prime}}_i c_{\alpha \beta};
\end{equation}
 
 Any linear trace Poisson bracket has the form 
\begin{equation}\label{Poissonlin}
\{x^j_{i,\alpha},x^{j^{\prime}}_{i^{\prime},\beta}\}=
b_{\alpha,\beta}^{\gamma}x_{i,\gamma}^{j^{\prime}}\delta^j_{i^{\prime}}-
b_{\beta,\alpha}^{\gamma}x_{i^{\prime},\gamma}^j\delta^{j^{\prime}}_i;
\end{equation}

Any quadratic trace Poisson bracket is given by 
\begin{equation}\label{Poisson}
\{x^j_{i,\alpha},x^{j^{\prime}}_{i^{\prime},\beta}\}=
r^{\gamma\epsilon}_{\alpha\beta}x^{j^{\prime}}_{i,\gamma}x^j_{i^{\prime},\epsilon}+
a^{\gamma\epsilon}_{\alpha\beta}x^k_{i,\gamma}x^{j^{\prime}}_{k,\epsilon}\delta^j_{i^{\prime}}-
a^{\gamma\epsilon}_{\beta\alpha}x^k_{i^{\prime},\gamma}x^{j}_{k,\epsilon}\delta^{j^{\prime}}_i.
\end{equation}
 Moreover

1)\,  Formula \eqref{Poissoncon} defines a Poisson
bracket iff
$$c_{\alpha \beta}=- c_{\beta \alpha};$$

2)\, Formula \eqref{Poissonlin} defines a Poisson
bracket iff
\begin{equation}\label{r0}
b^{\mu}_{\alpha \beta} b^{\sigma}_{\mu \gamma}=b^{\sigma}_{\alpha \mu} b^{\mu}_{\beta \gamma};
\end{equation}

3)\, Formula \eqref{Poisson} defines a Poisson brackets iff conditions \eqref{r1}--\eqref{r5} hold.
\end{theorem}

\begin{remark} Formula  \eqref{r0} means that  $b^{\sigma}_{\alpha \beta}$
are the structure constants of an associative algebra $\cal A$. A straightforward verification shows that  \eqref{Poissonlin} is nothing but the Lie-Kirillov-Kostant bracket defined by the Lie algebra corresponding to the associative  algebra ${\rm Mat}_m\otimes {\cal A}.$
\end{remark} 

\begin{lemma} For any Hamiltonian of the form $H=\mbox{tr}\,P,$ where $P$ is a matrix polynomial, the equations of 
motion can be written in the matrix form \eqref{geneq}.
\end{lemma}

Under a linear transformation of the matrices $x_i \to g^j_i x_j$ the constants in \eqref{Poisson} are
transformed in a standard way:
\begin{equation}\label{basis}
r_{ij}^{kl} \to g^{\alpha}_i g^{\beta}_j h_{\gamma}^k h_{\epsilon}^l
\,r_{\alpha,\beta}^{\gamma,\epsilon}, \qquad
a_{ij}^{kl} \to g^{\alpha}_i g^{\beta}_j h_{\gamma}^k h_{\epsilon}^l
\,a_{\alpha,\beta}^{\gamma,\epsilon}.
\end{equation}
Here $g_i^jh_j^k=\delta_i^k$.
\begin{definition}
Two brackets of the form \eqref{Poisson}  related by  \eqref{basis} are called {\it equivalent}.
\end{definition}

A relation between non-abelian and trace Poisson brackets can be established via the formula
$$
x^{j}_{i,\alpha}={\rm tr}(e^{i}_{j} x_{\alpha}), \qquad
x^{j^{\prime}}_{i^{\prime},\beta}={\rm
tr}(e^{i^{\prime}}_{j^{\prime}}
x_{\beta}),
$$
where $e^{i}_{j}$ stand for the matrix unities. For instance, consider the Hamiltonian operator \eqref{tet}. Applying formula \eqref{nabr} and using the definition of the gradient, we arrive at \eqref{Poisson}.

Identities \eqref{r1}--\eqref{r5} can be rewritten in a tensor form. 
Let $\bf V$ be a linear space with a basis ${\bf e}_i,~i=1,...,N$. Define linear
operators $R$ and $A$ on the space ${\bf V}\otimes {\bf V}$ by $R\, {\bf e}_i\otimes
{\bf e}_j=r^{pq}_{ij} {\bf e}_p\otimes {\bf e}_q,~A \,{\bf e}_i\otimes {\bf e}_j=a^{pq}_{ij}{\bf e}_p\otimes
{\bf e}_q$.  Then identities \eqref{r1} - \eqref{r5} are equivalent to
$$R^{12}=-R^{21},\qquad R^{23}R^{12}+R^{31}R^{23}+R^{12}R^{31}=0,$$
$$A^{12}A^{31}=A^{31}A^{12},$$
$$\sigma^{23}A^{13}A^{12}=A^{12}R^{23}-R^{23}A^{12},$$
$$A^{32}A^{12}=R^{13}A^{12}-A^{32}R^{13}.$$
Here all operators act on ${\bf V}\otimes {\bf V}\otimes {\bf V}$, by $\sigma^{ij}$ we
mean transposition of $i$-th and $j$-th components of the tensor product
and $A^{ij},~R^{ij}$ mean operators $A,~R$ acting on the tensor product of
the $i$-th and $j$-th components.

The equivalence transformation
 \eqref{basis} corresponds to $A \to G A G^{-1},
R \to  G R G^{-1}$, where $G=g\otimes g$ and $g\in {\rm End}({\bf V})$.

\begin{definition} (cf.\eqref{admiss})
A vector $\Lambda=(\lambda_1,...,\lambda_N)$ is said to be {\it
admissible} for \eqref{Poisson} if for any $i,j$
$$
(a^{pq}_{ij}  - a^{qp}_{ji} + r^{pq}_{ij}) \lambda_p \lambda_q=0.
$$
\end{definition}

\begin{lemma} For any admissible vector the argument shift $x_i \to x_i+\lambda_i\,
{\rm Id}$ in \eqref{Poisson} yields a linear Poisson bracket  \eqref{Poissonlin}, where
$$
b^{p}_{ij}=(a_{ij}^{qp}+a_{ij}^{pq}+r_{ij}^{pq})\lambda_q,
$$
compatible with  \eqref{Poisson}.
\end{lemma}

\subsubsection{Case ${\bf a}=0$ and anti-Frobenius algebras }

There is a subclass of brackets  \eqref{Poisson} that corresponds to the case when the tensor ${\bf a}$ is equal to $0$. Relations   \eqref{r1},  \eqref{r2} mean 
that the tensor ${\bf r}$ is a constant solution of the associative Yang-Baxter equation (\cite{Agui},\cite{Sch}). Such solutions can be constructed in the following algebraic way.
 
\begin{definition} An {\it anti-Frobenius algebra} is an associative algebra  $\cal A$
(not necessarily with unity) with non-degenerate
anti-symmetric bilinear form $(~,~)$ satisfying the following relation
\begin{equation}\label{af}
(x,\,y\circ z)+(y,\,z\circ x)+(z,\,x\circ y)=0
\end{equation}
for all $x,y,z\in \cal A$. In other words the form $(~,~)$ defines an 1-cocycle on $\cal A$.
\end{definition}

\begin{theorem} There exists a one-to-one correspondence between
solutions of  \eqref{r1},  \eqref{r2}  up to equivalence and exact
representations of anti-Frobenius algebras up to isomorphism.
\end{theorem}

\begin{proof} The tensor ${\bf r}$ can be written as  
$r^{ij}_{kl}=\sum_{\alpha,\beta=1}^pg^{\alpha\beta}y^i_{k,\alpha}y^j_{l,\beta},$
where $g^{\alpha\beta}=-g^{\beta\alpha}$, the matrix $G=(g^{\alpha\beta})$
is non-degenerate and $p$ is the smallest possible.
Substituting this expression into  \eqref{r1},  \eqref{r2}, we obtain that there
exists a tensor $\phi^{\gamma}_{\alpha\beta}$
such that
$y^i_{k,\alpha}y^k_{j,\beta}=\phi^{\gamma}_{\alpha\beta}y^i_{j,\gamma}$. Let
 $\cal A$ be an
associative algebra with basis $y_1,...,y_p$ and product
$y_{\alpha}\circ y_{\beta}=\phi^{\gamma}_{\alpha\beta}\,y_{\gamma}$.
Define the anti-symmetric bilinear form by $(y_{\alpha},y_{\beta})=g_{\alpha\beta},$ where
$\{g_{\alpha\beta}\}=G^{-1}$.
Then  \eqref{r1},  \eqref{r2} are equivalent to the anti-Frobenius property  \eqref{af}.
\end{proof}

\begin{example} \,(cf. \cite{elash}). Let $\cal A$ be the associative algebra of $N\times N$-matrices with zero $N$-th row. For a generic element $l$ of ${\cal A}^{*};$
the bilinear form $(x,y)=l([x,y])$ is a non-degenerate 
anti-symmetric form, which satisfies  \eqref{af}. It can be written as $(x,y)={\rm tr}([x,y]\, k^T),$ where $k\in {\cal A}$. Let us choose $\quad k_{ij}=0,\,\, i\ne j$, $\quad k_{ii}=\mu_i,$ where $i,j=1,...,N-1,$ and $\quad k_{iN}=1,\,\, i=1,...,N-1$. The corresponding bracket  \eqref{Poisson} is given by the following tensor ${\bf r}$:
\begin{equation}\label{elash}
r^{ii}_{Ni}=-r^{ii}_{iN}=1, \qquad r^{ij}_{ij}=r^{ji}_{ij}=r^{ii}_{ji}=-r^{ii}_{ij}=\frac{1}{\mu_i-\mu_j}, 
\end{equation}
where $  i\ne j,\quad i,j=1,...,N-1.$ 
The remaining elements of the tensor ${\bf r}$ and all elements of the tensor ${\bf a}$ are supposed to be zero.  It can be verified that   \eqref{elash} is equivalent to the bracket given by
\begin{equation}\label{ex4}
  r^{\alpha \beta}_{\alpha \beta}=r^{\beta \alpha}_{\alpha \beta}=r^{\alpha\alpha}_{\beta\alpha}=-r^{\alpha\alpha}_{\alpha\beta}=\frac{1}{\lambda_\alpha-\lambda_\beta},
\qquad \alpha\ne \beta,\qquad \alpha,\beta=1,\ldots,N.
\end{equation}
Here $\lambda_1,\ldots,\lambda_N$ are arbitrary pairwise distinct parameters.  Formula  \eqref{Poisson} with zero tensor ${\bf a}$ defines the corresponding trace Poisson bracket for entries of matrices $x_1,\dots, x_N$ of arbitrary size $m$. For $m=1$ we have the  scalar Poisson bracket 
$$
\{x_\alpha,x_\beta\}=\frac{(x_\alpha-x_\beta)^2}{\lambda_\beta-\lambda_\alpha}, \qquad \alpha\ne \beta,\qquad \alpha,\beta=1,...,N.
$$
If $N$ is even, then the trace Poisson structure  \eqref{ex4} is non-degenerate, i.e. the rank of the Poisson tensor $\Pi$ is equal to $
N m^2$. In the odd case ${\rm rank} \, \Pi =(N-1)\, m^2$.
\end{example}
\begin{remark}
Bracket \eqref{ex4} can be directly obtained from the anti-Frobenius algebra
$$
 \mathcal{A}_{N,1} = \{ A \in {\rm Mat}_N \; \mid \sum_{i} a_{ij} = 0 \quad \forall \, j = 1,\ldots,N \}
$$
equipped with the bilinear form
\begin{equation}
\label{form}
 (x, y) = {\rm tr} \left([x,y] \cdot {\rm diag}\,(\lambda_1, \ldots, \lambda_N) \right).
\end{equation}
\end{remark}

In \cite{zobnin} the algebra $\mathcal{A}_{N,1}$ was generalized.
Let $M$ be a proper divisor of $N$. We consider $N(N-M)$-dimensional algebra
$$
 \mathcal{A}_{N,M} = \{ A \in  {\rm Mat}_N \; \mid \sum_{i \equiv r \!\!\pmod{M}} a_{ij} = 0 \quad \forall \, r = 1, \ldots, M, \; \forall \, j = 1,\ldots,N \}
$$
equipped with the bilinear form~\eqref{form}. Suppose that  $\lambda_i$ are pairwise distinct.
One can check that in this case the form $(x, y)$ is non-degenerate \cite{elash}.
 The components of the tensor ${\bf r}$ corresponding to the algebra $\mathcal{A}_{N,M}$ are given by
\begin{gather*}
r^{\alpha \beta}_{\gamma \delta} = 0, \qquad \text{ if } \alpha \not \equiv \delta \quad \text{ or }\quad  \beta \not \equiv \gamma, \\[2mm]
r^{\alpha \alpha}_{\varepsilon \alpha} = -r^{\alpha \alpha}_{\alpha \varepsilon} = \frac{1}{\lambda_\alpha - \lambda_\varepsilon}, \qquad \text{ when } \alpha \ne \varepsilon, \\[2mm]
r^{\alpha \alpha}_{\gamma \delta} = 0, \qquad \text{ if } \gamma \ne \alpha \quad \text{ or } \quad \delta \ne \alpha,\\[3mm]
r^{\alpha \beta}_{\beta \alpha} = 
\frac{1}{\lambda_\alpha - \lambda_\beta}
\left(
\frac{
\prod\limits_{\beta' \equiv \beta, \, \beta' \ne \beta} (\lambda_\alpha - \lambda_{\beta'}) \prod\limits_{\alpha' \equiv \alpha, \, \alpha' \ne \alpha} (\lambda_\beta - \lambda_{\alpha'})}
{\prod\limits_{\alpha' \equiv \alpha, \, \alpha' \ne \alpha} (\lambda_\alpha - \lambda_{\alpha'}) \prod\limits_{\beta' \equiv \beta, \, \beta' \ne \beta} (\lambda_\beta - \lambda_{\beta'})}
- 1\right), \qquad \text{ if } \alpha \ne \beta, \\[4mm]
r^{\alpha \beta}_{\gamma \delta} = 
\frac{1}{\lambda_\alpha - \lambda_\beta} \cdot \frac{
\prod\limits_{\gamma' \equiv \gamma, \, \gamma' \ne \gamma} (\lambda_\alpha - \lambda_{\gamma'}) \prod\limits_{\delta' \equiv \delta, \, \delta' \ne \delta} (\lambda_\beta - \lambda_{\delta'})}
{\prod\limits_{\alpha' \equiv \alpha, \, \alpha' \ne \alpha} (\lambda_\alpha - \lambda_{\alpha'}) \prod\limits_{\beta' \equiv \beta, \, \beta' \ne \beta} (\lambda_\beta - \lambda_{\beta'})} \qquad \text{ otherwise}.
\end{gather*}
Here $x \equiv y$ means that $x = y \! \pmod{M}$.

With these formulas one can construct corresponding quadratic Poisson brackets.
For example, in the case $N = 2M$ and $m=1$ the corresponding scalar Poisson bracket has the form
$$
 \left\{x_{\alpha}, x_{\beta}\right\} = \frac{(x_{\alpha} - x_{\alpha'})(x_{\beta}-x_{\beta'})(\lambda_{\alpha'}-\lambda_{\beta'})}{
 (\lambda_{\alpha} - \lambda_{\alpha'})
 (\lambda_{\beta} - \lambda_{\beta'})},
$$
where for any $\gamma$  the positive integer $\gamma'$ is uniquely defined by the condition $|\gamma' - \gamma| = M$.

\begin{op} Describe all anti-Frobenius algebras ${\cal A}$ of the form 
$$
{\cal A}={\cal S}\oplus {\cal M},
$$
where ${\cal S}$ is a semi-simple associative algebra and ${\cal M}$ is a ${\cal S}$-module such that  ${\cal M}^2=\{0\}.$
\end{op}
\subsubsection{Classification of trace quadratic brackets for $N=2$}

Consider the case $N=2.$ 
\begin{theorem}
Any Poisson bracket  \eqref{Poisson} up to transformations
 \eqref{basis} and to the proportionality is one of the following brackets. Here we present non-zero
components of the tensors ${\bf a}$ and ${\bf r}$ only.\begin{itemize}
\item   $r^{22}_{12}=1,$ $\qquad r^{22}_{21}=-1$;

\item  $r^{21}_{11}=1,$ $\qquad r^{12}_{11}=-1,$ $ \qquad a^{22}_{21}=a^{12}_{11}=-1$;

\item $r^{21}_{11}=-1,$ $\qquad r^{12}_{11}=1,$  $\qquad a^{22}_{12}=a^{21}_{11}=1$;

\item   $r^{22}_{12}=1,$ $\qquad r^{22}_{21}=-1,$ $\qquad a^{12}_{11}=a^{22}_{21}=1$;

\item   $r^{22}_{21}=1,$ $\qquad r^{22}_{12}=-1,$ $\qquad a^{21}_{11}=a^{22}_{12}=1$;

\item $a^{22}_{11}=1$;

\item   $r^{21}_{11}=1,$ $\qquad r^{12}_{11}=-1$.

\end{itemize}
\end{theorem}
For a classification in the case $N=3$, ${\bf a}=0$ see \cite{sokol3}. 
\begin{op}
Describe all trace Poisson brackets  \eqref{Poisson} for $N=3$.
\end{op}

\subsection{Double Poisson brackets on free associative algebra}

In the previous sections we observed that identities \eqref{r1}--\eqref{r5} describe both non-abelian and trace quadratic Poisson brackets. 
Here we show that this is also true for the quadratic double Poisson brackets on the free associative algebra with generators $x_1,...,x_N.$

\begin{definition}  \cite{VdB}. A double Poisson bracket on an associative algebra $A$ is a $\C$-linear map
$\ldb,\rdb : A \otimes A \to  A \otimes A$ satisfying the following conditions:
$$
\ldb u, v\rdb = - \ldb v,u\rdb ^{\circ},
$$
$$
 \ldb u, \ldb v,w \rdb \rdb + \sigma  \ldb v, \ldb w,u \rdb \rdb +\sigma^2  \ldb w, \ldb u,v \rdb \rdb  =0,
$$
and
$$
\ldb u, vw \rdb = (v\otimes 1)  \ldb u,w \rdb  + \ldb u,v \rdb (1\otimes  w).
$$
Here $\quad (u\otimes v)^{\circ}\stackrel{def}{=} v\otimes u$; \quad $\ldb v_1, v_2\otimes v_3 \rdb :=\ldb v_1, v_2 \rdb \otimes v_3$ \quad 
and \quad $\sigma( v_1 \otimes v_2 \otimes v_3 ):= v_{3}\otimes v_{1} \otimes v_{2}$.
\end{definition}

Denote by $\mu$ the multiplication map $\mu:  A\otimes A \to A$ given by $\mu(u\otimes v)=u v.$
We define a $\C$-bilinear bracket operation on $A$ by  $\{\cdot,\,\cdot\}\stackrel{def}{=} \mu(\ldb \cdot,\, \cdot \rdb).$

\begin{proposition} Let $\ldb \cdot,\, \cdot \rdb$ be a double Poisson bracket on $A$. Then $\{\cdot,\, \cdot\}$ is a trace bracket on $A/[A, A]$, which is defined as
$$
\lbrace \bar a, \bar b\rbrace =\overline{\mu(\ldb a,b\rdb}),
$$
where $\bar a$ means the image of $a\in A$ under the natural projection $A\to A/[A,A].$
\end{proposition}

Let ${\cal A}=\C[x_1,\ldots,x_N]$ be the free associative algebra. If the double brackets $\ldb x_i, x_j \rdb$ between all generators $x_i$ are fixed, then the bracket between two arbitrary elements of ${\cal A}$ 
is uniquely determined. Constant, linear, and quadratic double brackets are defined by 
$$
\ldb x_i,x_j\rdb = c_{ij} 1\otimes 1, \qquad c_{i,j}=-c_{j,i},
 $$$$
\ldb x_i,x_j\rdb = b_{ij}^k x_k\otimes 1 - b_{ji}^k1\otimes x_k,
$$
and  
$$
\ldb x_{\alpha}, x_{\beta}\rdb =r_{\alpha \beta}^{u v} \, x_u \otimes x_v+a_{\alpha \beta}^{v u} \, x_u x_v\otimes 1-a_{\beta \alpha}^{u v} \,1\otimes  x_v x_u,  
$$
respectively.
\begin{proposition} These formulas define  double Poisson brackets iff the constants $ c_{ij}, b_{ij}^k, r^{pq}_{ij}, a^{pq}_{ij} $ satisfy the identities of Theorem \ref{th51}.
\end{proposition}

\end{document}